\documentclass[11pt]{article}
\usepackage[margin=1in]{geometry}
\usepackage[page,toc,titletoc,title]{appendix}
\usepackage{amssymb,amsmath,amsthm,amsfonts}
\usepackage{textgreek}
\usepackage{mathrsfs}
\usepackage{mathtools}
\usepackage[makeroom]{cancel}
\usepackage{enumitem}
\usepackage{authblk}
\usepackage{graphicx,booktabs,multirow}
\usepackage[font=small]{caption}
\usepackage[labelformat=simple]{subcaption}
\usepackage{float}
\usepackage[ruled,vlined,linesnumbered]{algorithm2e}
\usepackage{physics}
\usepackage{footnote}
\usepackage[table,usenames,dvipsnames]{xcolor}
\usepackage{xcolor}
\usepackage{bm}
\usepackage{bbm}
\usepackage{lineno}
\usepackage{siunitx}

\usepackage[numbers,comma,sort&compress]{natbib}
\usepackage{multibib}
\newcites{Main}{References}

\usepackage[colorlinks]{hyperref}

\newtheorem{theorem}{Theorem}
\newtheorem{definition}{Definition}
\newtheorem{lemma}{Lemma}
\newtheorem{proposition}{Proposition}

\newtheorem{corollary}{Corollary}
\newtheorem{remark}{Remark}
\newtheorem{assumption}{Assumption}

\usepackage{thm-restate}

\newcommand{\eqn}[1]{(\ref{eqn:#1})}

\newcommand{\rem}[1]{\hyperref[rem:#1]{Remark~\ref*{rem:#1}}}
\newcommand{\thm}[1]{\hyperref[thm:#1]{Theorem~\ref*{thm:#1}}}
\newcommand{\cor}[1]{\hyperref[cor:#1]{Corollary~\ref*{cor:#1}}}
\newcommand{\defn}[1]{\hyperref[defn:#1]{Definition~\ref*{defn:#1}}}
\newcommand{\lem}[1]{\hyperref[lem:#1]{Lemma~\ref*{lem:#1}}}
\newcommand{\prop}[1]{\hyperref[prop:#1]{Proposition~\ref*{prop:#1}}}
\newcommand{\fig}[1]{\hyperref[fig:#1]{Figure~\ref*{fig:#1}}}
\newcommand{\tab}[1]{\hyperref[tab:#1]{Table~\ref*{tab:#1}}}
\newcommand{\algo}[1]{\hyperref[algo:#1]{Algorithm~\ref*{algo:#1}}}
\renewcommand{\sec}[1]{\hyperref[sec:#1]{Section~\ref*{sec:#1}}}
\newcommand{\append}[1]{\hyperref[append:#1]{Appendix~\ref*{append:#1}}}
\newcommand{\fac}[1]{\hyperref[fac:#1]{Fact~\ref*{fac:#1}}}
\newcommand{\lin}[1]{\hyperref[lin:#1]{Line~\ref*{lin:#1}}}
\newcommand{\fnote}[1]{\hyperref[fnote:#1]{Footnote~\ref*{fnote:#1}}}

\newcommand{\assump}[1]{\hyperref[assump:#1]{Assumption~\ref*{assump:#1}}}

\def\>{\rangle}
\def\<{\langle}

\newcommand{\vect}[1]{\ensuremath{\boldsymbol{#1}}}
\newcommand{\x}{\ensuremath{\mathbf{x}}}

\newcommand{\Z}{\mathbb{Z}}
\newcommand{\R}{\mathbb{R}}
\newcommand{\C}{\mathbb{C}}

\newcommand{\Herm}{\text{Herm}}
\renewcommand{\S}{\mathcal{S}}
\newcommand{\simH}{\widetilde{H}}
\renewcommand{\H}{\mathcal{H}}
\newcommand{\Hs}{\widetilde{\H}}

\renewcommand{\d}{\mathrm{d}}

\def\:{\hbox{\bf:}}

\let\oldnl\nl
\newcommand{\nonl}{\renewcommand{\nl}{\let\nl\oldnl}}

\begin{document}

\title{Quantum Hamiltonian Descent\thanks{This work was partially funded by the U.S. Department of Energy, Office of Science, Office of
Advanced Scientific Computing Research, Quantum Testbed Pathfinder Program under Award
Number DE-SC0019040, Accelerated Research in Quantum Computing under Award Number
DE-SC0020273, and the U.S. National Science Foundation
grant CCF-1816695 and CCF-1942837 (CAREER). We also acknowledge the research credits from Amazon Web Services. An accompanying website is at \url{https://jiaqileng.github.io/quantum-hamiltonian-descent/}.}}

\author{Jiaqi Leng$^{1,3}$ \ \ \ Ethan Hickman$^{2,3}$ \ \ \ Joseph Li$^{2,3}$ \ \ \  Xiaodi Wu$^{2,3,\dagger}$ \\
\small{$^1$Department of Mathematics, University of Maryland, College Park, USA} \\
\small{$^2$Department of Computer Science, University of Maryland} \\
\small{$^3$Joint Center for Quantum Information and Computer Science, University of Maryland} \\
\small{$^\dagger$Corresponding author. \href{mailto:xiaodiwu@umd.edu}{xiaodiwu@umd.edu}}
}

\normalsize

\date{}

\maketitle
\thispagestyle{empty}

\begin{abstract}
Gradient descent is a fundamental algorithm in both theory and practice for continuous optimization. Identifying its quantum counterpart would be appealing to both theoretical and practical quantum applications. A conventional approach to quantum speedups in optimization relies on the quantum acceleration of intermediate steps of classical algorithms, while keeping the overall algorithmic trajectory and solution quality unchanged. We propose Quantum Hamiltonian Descent (QHD), which is derived from the path integral of dynamical systems referring to the continuous-time limit of classical gradient descent algorithms, as a truly quantum counterpart of classical gradient methods where the contribution from classically-prohibited trajectories can significantly boost QHD's performance for non-convex optimization. Moreover, QHD is described as a Hamiltonian evolution efficiently simulatable on both digital and analog quantum computers. By embedding the dynamics of QHD into the evolution of the so-called Quantum Ising Machine (including D-Wave and others), we empirically observe that the D-Wave-implemented QHD outperforms a selection of state-of-the-art gradient-based classical solvers and the standard quantum adiabatic algorithm, based on the time-to-solution metric, on non-convex constrained quadratic programming instances up to 75 dimensions. Finally, we propose a “three-phase picture” to explain the behavior of QHD, especially its difference from the quantum adiabatic algorithm.  
\end{abstract}

\vspace{4mm}
Continuous optimization, stemming from the mathematical modeling of real-world systems, is ubiquitous in applied mathematics, operations research, and computer science~\citeMain{nocedal1999numerical,weinan2020machine,burer2012non,lecun2015deep,jain2017non}.
These problems often come with high dimensionality and non-convexity, posing great challenges for the design and implementation of optimization algorithms~\citeMain{glorot2010understanding,chen2019deep}. 
It is a natural question to investigate potential quantum speedups for continuous optimization, which has been actively studied in past decades. 
A conventional approach toward this end is to quantize existing classical algorithms by replacing their components with quantum subroutines, while carefully balancing the potential speedup and possible overheads.
However, proposals (e.g.~\citeMain{brandao2017quantum, van2020quantum, kalev2019quantum, chakrabarti2020quantum, van2020convex, li2019sublinear, zhang2021quantum}) following this approach usually achieve only moderate quantum speedups, and more importantly, they rarely improve the quality of solutions because they essentially follow the same solution trajectories in the original classical algorithms.

Gradient descent (and its variant) is arguably the most fundamental optimization algorithm in continuous optimization, both in theory and in practice, due to its simplicity and efficiency in converging to critical points. 
However, many real-world problems have spurious local optima~\citeMain{safran2018spurious}, for which gradient descent is subject to slow convergence since it only leverages first-order information~\citeMain{ruder2016overview}.
On the other hand, quantum algorithms have the potential to escape from local minima and find near-optimal solutions by leveraging the \textit{quantum tunneling} effect~\citeMain{farhi:quantum, boixo2014evidence}. Therefore, it is desirable to identify a quantum counterpart of gradient descent that is simple and efficient on quantum computers while leveraging the quantum tunneling effect to escape from spurious local optima. With such features, the quality of the solutions is improved.
Prior attempts (e.g.,~\citeMain{rebentrost:quantum}) to quantize gradient descent, which followed the conventional approach, unfortunately fail to achieve the aforementioned goal, which seems to require a completely new approach to quantization. 

\begin{figure}[ht!]
\centering
     \includegraphics[width=16cm]{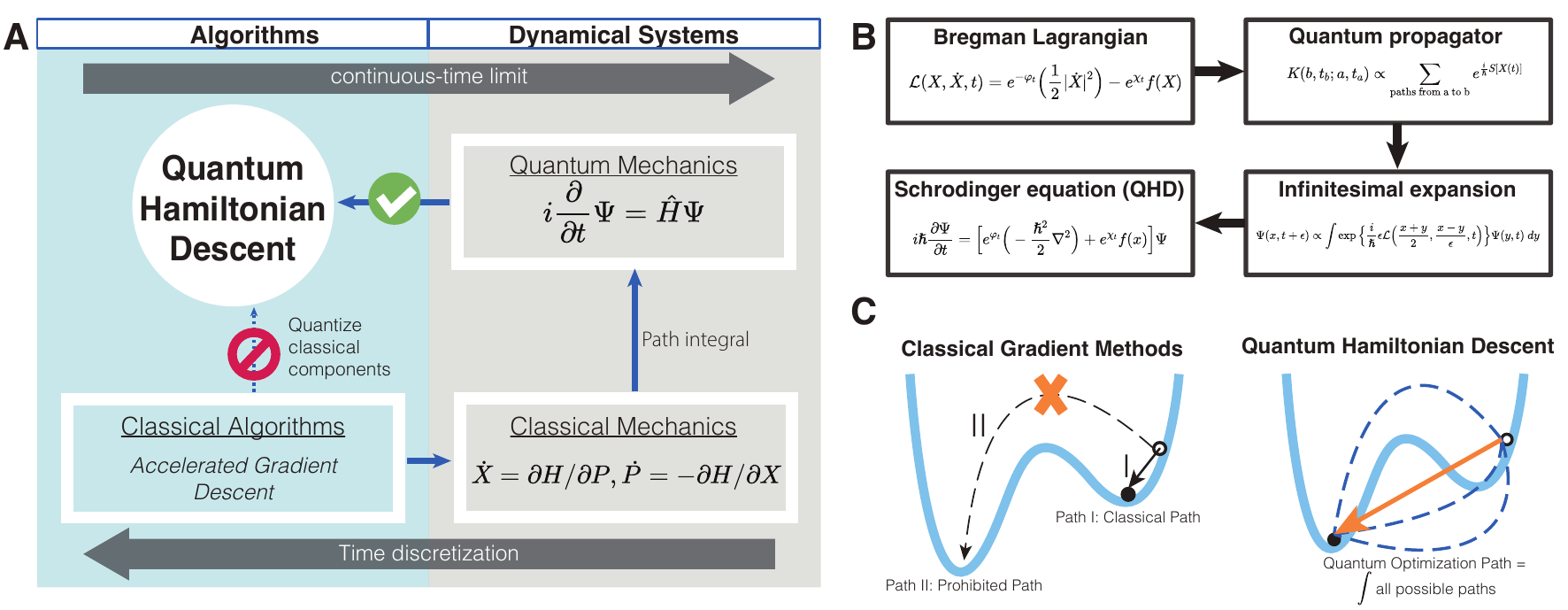}
     \caption{\textbf{Schematic of Quantum Hamiltonian Descent (QHD).}
     \textbf{A.} Road map showing the conventional and our approach of quantizing classical gradient descent (GD). Specifically, QHD is derived through the path-integral quantization of the dynamical system corresponding to the continuous-time limit of classical GD, and hence can be deemed as the path integral of the algorithmic trajectories of classical GD. 
     \textbf{B.} Four major technical steps in the derivation of QHD.
     \textbf{C.} An illustrative example where classical GD with bad initialization will be trapped in a local minimum, while QHD can easily escape and find near-optimal solutions by taking the path integral of trajectories prohibited by classical mechanics.}
     \label{fig:fig1}
\end{figure}

\section*{A genuine quantum gradient descent}
 
Our main observation is a perhaps unintuitive connection between gradient descent and dynamical systems satisfying classical physical laws.
Precisely, it is known that the continuous-time limit of many gradient-based algorithms can be understood as classical physical dynamical systems, e.g., the Bregman-Lagrangian framework derived in~\citeMain{wibisono:variational} to model accelerated gradient descent algorithms. 
Conversely, variants of gradient-based algorithms could be emerged through the time discretization of these continuous-time dynamical systems~\citeMain{betancourt:symplectic}.
This two-way correspondence inspired us a second approach to quantization: instead of quantizing a subroutine in gradient descent, we can quantize the continuous-time limit of gradient descent as a whole, and the resulting quantum dynamical systems lead to quantum algorithms as seen in \fig{fig1}A.
Using the path integral formulation of quantum mechanics (\fig{fig1}B), we quantize the Bregman-Lagrangian framework to obtain a quantum-mechanical system governed by the Schr\"odinger equation $i \frac{\d}{\d t} \Psi(t) = \hat{H}(t) \Psi(t)$, where $\Psi(t)$ is the quantum wave function, and the quantum Hamiltonian reads: 
\begin{align}\label{eqn:main1}
   \hat{H}(t) =e^{\varphi_t}\left(-\frac{1}{2}\Delta\right) + e^{\chi_t}f(x),
\end{align}
where $e^{\varphi_t}$, $e^{\chi_t}$ are \textit{damping parameters} that control the energy flow in the system. We require $e^{\varphi_t/\chi_t}\to 0$ for large $t$ so the kinetic energy is gradually drained out from the system, which is crucial for the long-term convergence of the evolution. Just like the Bregman-Lagrangian framework, different damping parameters in $\hat{H}(t)$ correspond to different prototype gradient-based algorithms. $\Delta$ is the Laplacian operator over Euclidean space. $f(x)$, the objective function to minimize, is assumed to be unconstrained and continuously differentiable. More details of the derivation of QHD are available in \sec{derivation}. The Schr\"odinger dynamics in \eqn{main1} hence generates a family of quantum gradient descent algorithms that we will refer to as \textbf{Quantum Hamiltonian Descent}, or simply \textbf{QHD}.

As desired, QHD inherits simplicity and efficiency from classical gradient descent. QHD takes in an easily prepared initial wave function $\Psi(0)$ and evolves the quantum system described by \eqn{main1}. 
The solution to the optimization problem is obtained by measuring the position observable $\hat{x}$ at the end of the algorithm (i.e., at time $t=T$).
In other words, QHD is no different from a basic Hamiltonian simulation task, which can be done on digital quantum computers using standard techniques as shown in \algo{trotter_qhd} with provable efficiency (\thm{gate-complexity}).
The simplicity and efficiency of QHD on quantum machines potentially make it as widely applicable as classical gradient descent. 

For convex problems, we prove in \thm{convex_convergence} that QHD is guaranteed to find the global solution. In this case, the solution trajectory of QHD is analogous to that of a classical algorithm. Non-convex problems, known to be \texttt{NP}-hard in general, are much harder to solve. Under mild assumptions on a non-convex $f$, we show the global convergence of QHD given appropriate damping parameters and sufficiently long evolution time (\thm{adiabatic_limit_qhd}). \fig{fig1}C shows a conceptual picture of QHD's quantum speedup: intuitively, QHD can be regarded as a \textit{path integral} of solution trajectories, some of which are prohibited in classical gradient descent. Interference among all solution trajectories gives rise to a unique quantum phenomenon called \textit{quantum tunneling}, which helps QHD overcome high-energy barriers and locate the global minimum.

\section*{Performance of QHD on hard optimization problems}

To visualize the difference between QHD and other classical and quantum algorithms, we test four algorithms (QHD, Quantum Adiabatic Algorithm (QAA), Nesterov's accelerated gradient descent (NAGD), and stochastic gradient descent (SGD)) via classical simulation on 22 optimization instances with diversified landscape features selected from benchmark functions for global optimization problems.\footnote{More details of the 22 optimization instances can be found in \sec{2d-test-problems}. {In our experiment, we implement SGD by adding a small Gaussian perturbation to the analytical gradient. Unlike deterministic algorithms like NAGD, this stochastic perturbation seems to help with non-convex optimization.}} QAA solves an optimization problem by simulating a quantum adiabatic evolution~\citeMain{farhi:quantum}, and it has mostly been applied to discrete optimization in the literature. To solve continuous optimization with QAA, a common approach is to represent each continuous variable with a finite-length bitstring so the original problem is converted to combinatorial optimization defined on the hypercube $\{0,1\}^N$, where $N$ is the total number of bits (e.g.,~\citeMain{cohen:portfolio,potok:adiabatic}). In our experiment, we adopt the radix-2 representation (i.e., binary expansion) and assign $7$ bits for each continuous variable. This allows QAA to handle the optimization instances as discrete problems over $\{0,1\}^{14}$.\footnote{Effectively, this means we discretize the continuous domain $[0,1]^2$ into a $128 \times 128$ mesh grid.}  See \sec{qaa} for the detailed setup. 

\begin{figure}[ht!]
    \centering
    \includegraphics[width=16cm]{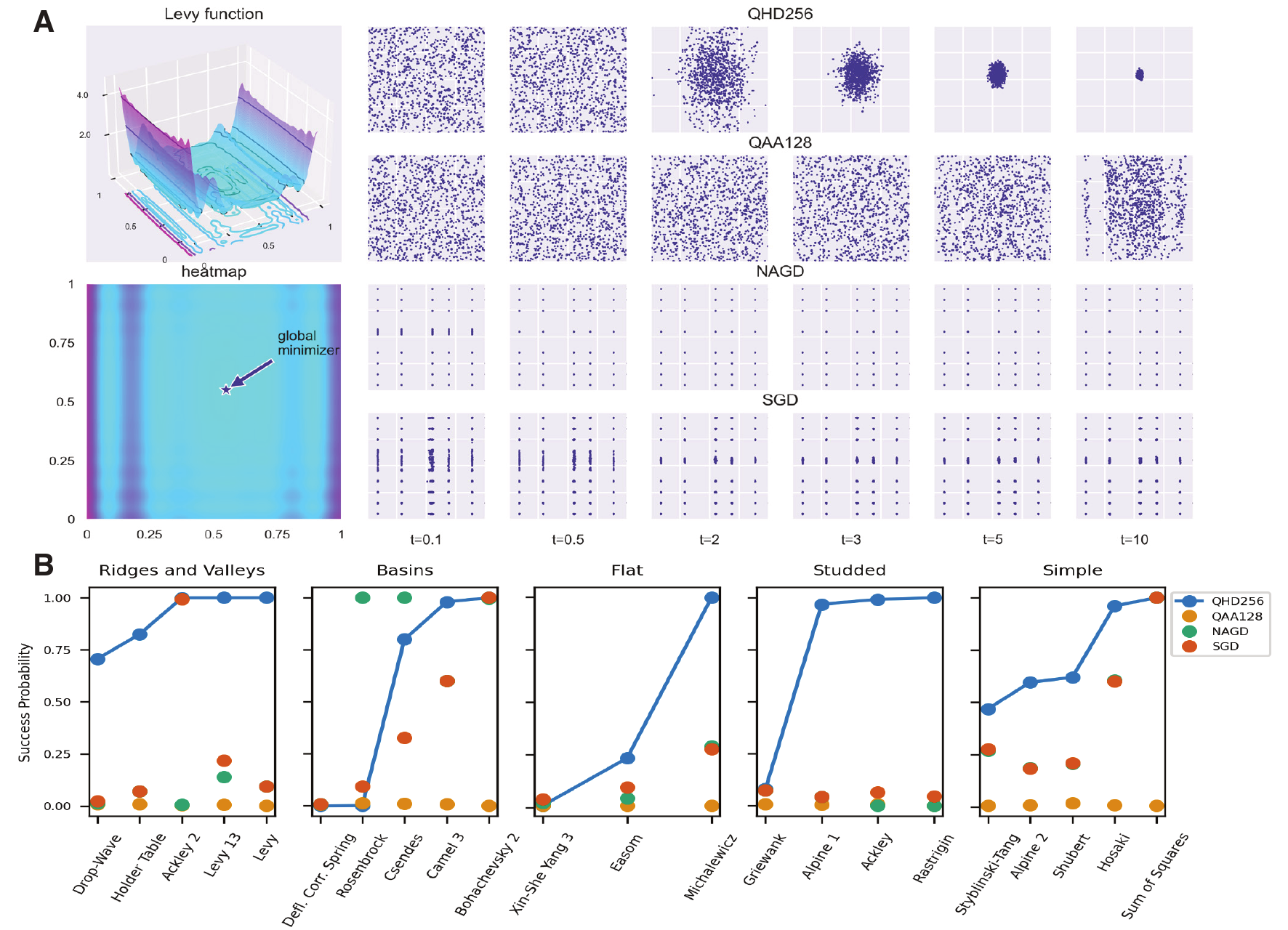}
    \caption{\textbf{Quantum and classical optimization methods for two-dimensional test problems.}
    \textbf{A.} Surface and heatmap plots of the Levy function. Samples from the distributions of QHD, QAA, NAGD, and SGD at different (effective) evolution times $t=0.1, 0.5, 2, 3, 5, 10$ are shown as scatter plots.
    \textbf{B.} Final success probabilities of QHD, QAA, NAGD, and SGD for all 22 instances. {A final solution $x_k$ is considered ``successful'' if $|x_k-x^*|<0.1$, where $x^*$ is the (unique) global minimizer of $f$.} Data are categorized into five groups by landscape features of the objective functions.}
    \label{fig:fig2}
    \end{figure}

In \fig{fig2}A, we plot the landscape of Levy function, and the solutions from the four algorithms are shown for different evolution times $t$. For QHD and QAA, $t$ is the evolution time of the quantum dynamics; for the two classical algorithms, the effective evolution time $t$ is computed by multiplying the step size and the numer of iterations so that it is comparable to the one used in QHD and QAA.
Compared with QHD, QAA converges at a much slower speed and little apparent convergence is observed within the time window.
Although the two classical algorithms seem to converge faster than quantum algorithms, they have lower success probability because many solutions are trapped in spurious local minima. Our observation made with Levy function is consistent with the results of other functions: as shown in \fig{fig2}B, QHD has a higher success probability in most optimization instances at the same choice of total effective evolution time.

\begin{figure}[ht!]
    \centering
    \includegraphics[width=16cm]{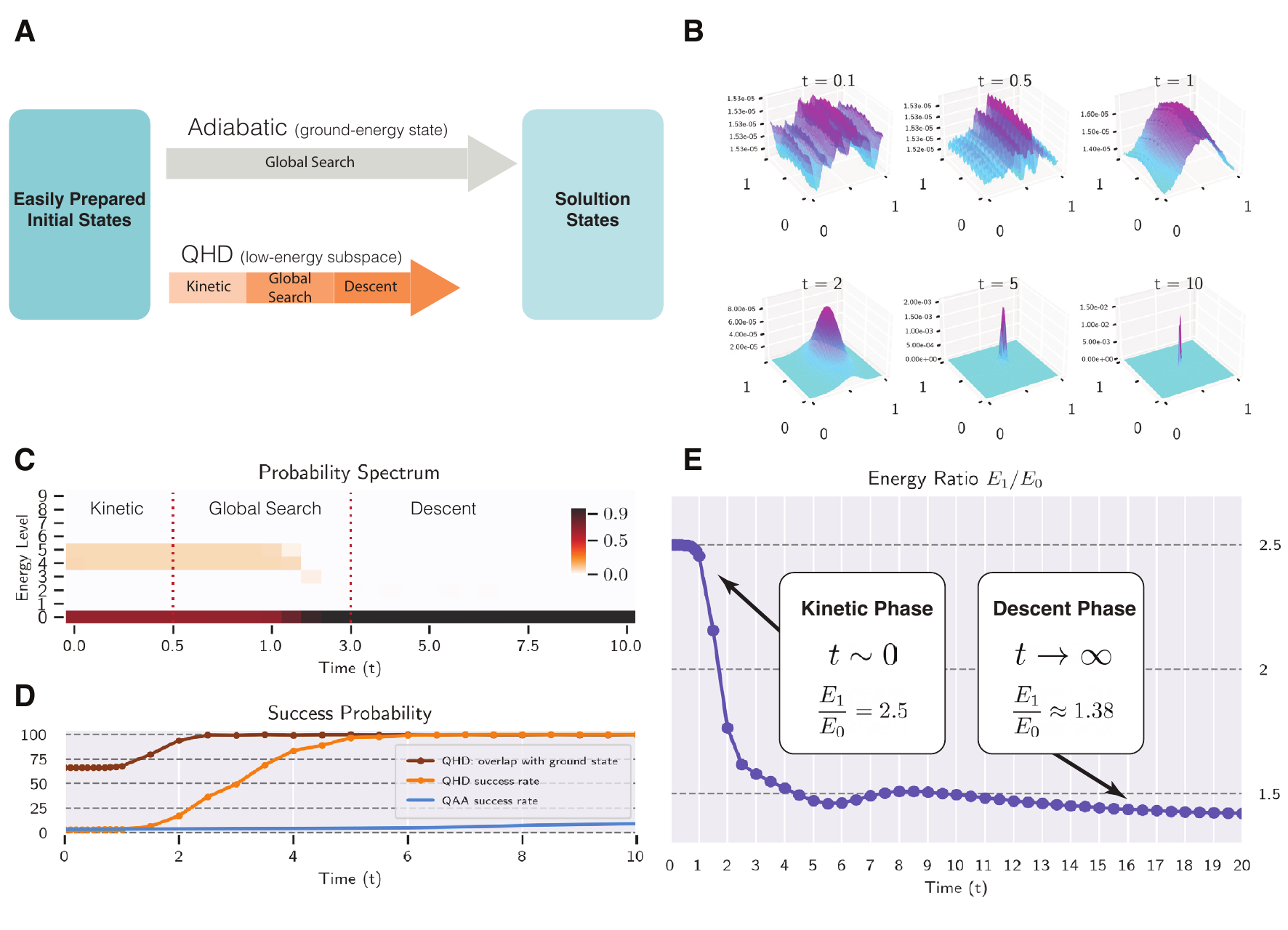}
    \caption{\textbf{The three-phase picture of QHD.}
    \textbf{A.} Schematic of the three-phase picture of QHD. QAA is a long-lasting global search procedure. QHD experiences three different phases and it has faster convergence for continuous problems.
    \textbf{B.} Surface plots of the probability density in QHD for the Levy function.
    \textbf{C.} Probability spectrum of QHD.
    \textbf{D.} Success probabilities of QHD and QAA.
    \textbf{E.} The energy ratio $E_1/E_0$ in QHD shown as a function of time $t$.}
    \label{fig:fig3}
\end{figure}

Focusing on the QHD dynamics, we find rich dynamical properties at different stages of evolution. \fig{fig3}B shows the quantum probability densities of QHD at different evolution times. First, the initial wave function becomes highly oscillatory and spreads to the full search space ($t=0.1,0.5$). Then, the wave function \textit{sees} the landscape of $f$ and starts moving towards the global minimum ($t=1,2$). Finally, the wave packet is clustered around the global minimum and converges like a classical gradient descent ($t=5,10$). This three-stage evolution is not only seen for the Levy function but also observed in many other instances (for details, see our website\footnote{\url{https://jiaqileng.github.io/quantum-hamiltonian-descent/}.}). 
We thus propose to divide QHD's evolution in solving optimization problems into three consecutive phases called the \textbf{kinetic phase}, the \textbf{global search phase}, and the \textbf{descent phase} according to the above observations. 

The three-phase picture of QHD could be supported by several quantitative characterizations of the QHD evolution. One such characterization is the probability spectrum of QHD, which shows the decomposition of the wave function to different energy levels (\fig{fig3}C). 
QHD begins with a major ground-energy component and a minor low-energy component.\footnote{When applied to the Levy function, no high-energy component with energy level $\ge 10$ is found in QHD.} During the global search phase, the low-energy component is absorbed into the ground-energy component, indicating that QHD finds the global minimum (\fig{fig3}D). The energy ratio $E_1/E_0$ is another characterization of the three phases in QHD (\fig{fig3}E), where $E_0$ (or $E_1$) is the ground (or first excited) energy of the QHD Hamiltonian $\hat{H}(t)$. In the kinetic phase, the kinetic energy $-\frac{1}{2}\Delta$ dominates in the system Hamiltonian so we have $E_1/E_0\approx 2.5$, which is the same as in a free-particle system. In the descent phase, the QHD Hamiltonian enters the ``semi-classical regime'' and the energy ratio can be theoretically computed based on the objective function.\footnote{For Levy function, the predicted semi-classical energy ratio reads $E_1/E_0\approx1.38$, which matches our numerical data.}

The three-phase picture of QHD sheds light on why QAA has slower convergence. Compared to QHD, QAA has neither a kinetic phase nor a descent phase. In the kinetic phase, QHD averages the initial wave function over the whole search space to reduce the risk of poor initialization, while QAA remains in the ground state throughout its evolution, so it never gains as much kinetic energy. In the descent phase, QHD exhibits convergence similar to classical gradient descent and  is insensitive to spatial resolution; such fast convergence is not seen in QAA.

In QAA, the use of the radix-2 representation scrambles the Euclidean topology so that the resulting discrete problem is even harder than the original problem (e.g., see \fig{graph_comparison}). Failing to incorporate the continuity structure, QAA is hence sensitive to the resolution of spatial discretization; we observe that higher resolutions often cause worse QAA performance, see \fig{compare_qaa_qhd}.\footnote{Of course, a radix-2 representation is not the only way to discretize a continuous problem. One can lift QAA to the continuous domain by choosing its Hamiltonian over a continuous space in a general way. From this perspective, QHD could be interpreted as a special version of the general QAA with a particular choice of the Hamiltonian. However, some existing results \citeMain{nenciu:linear} suggest that QHD may have fast convergence properties that the general theory of QAA fails to explain.}  See~\sec{three-phase} for details. 

\section*{Large-scale empirical study based on analog implementation}

One great promise of QHD lies in solving high-dimensional non-convex problems in the real world. However, a large-scale empirical study is infeasible with classical simulation due to the curse of dimensionality. 
Although theoretically efficient, implementation of QHD instances of reasonable sizes on digital quantum computers would cost a gigantic number of fault-tolerant quantum gates\footnote{In \tab{resource-analysis}, we show the count of T gates in the digital implementation of QHD. It turns out that solving 50-dimensional problems with low resolution will cost hundreds of millions of fault-tolerant T gates.}, rendering an  empirical study based on digital implementation a dead end in the near term.

Analog quantum computers (or quantum simulators) are alternative devices that directly emulate certain quantum Hamiltonian evolutions without using quantum gates, though they usually have more limited programmability.
However, recent experimental results suggest a great advantage of continuous-time analog quantum devices over the digital ones for quantum simulation in the NISQ era due to their scalability and lower overhead for some simulation tasks. 
Compared with other quantum algorithms, typically described by circuits of quantum gates, a unique feature of QHD is that its description is itself a Hamiltonian simulation task, which makes it possible to leverage near-term analog devices for its implementation. 

A conceptually simple analog implementation of QHD would be building a quantum simulator whose Hamiltonian exactly matches the target QHD Hamiltonian \eqn{main1}, which is, however, less feasible in practice. A more pragmatic strategy is to \textit{embed} the QHD Hamiltonian into existing analog simulators so we can emulate QHD as part of the full dynamics.

To this end, we introduce the \textbf{Quantum Ising Machine} (or simply \textbf{QIM}) as an abstract model for some of the most powerful analog quantum simulators nowadays. It is described by the following quantum Ising Hamiltonian:
\begin{align}\label{eqn:main2}
    H(t) = - \frac{A(t)}{2} \left(\sum_j \sigma^{(j)}_x\right) + \frac{B(t)}{2} \left(\sum_j h_j \sigma^{(j)}_z + \sum_{j>k} J_{j,k} \sigma^{(j)}_z \sigma^{(k)}_z\right),
\end{align}
where $\sigma^{(j)}_x$ and $\sigma^{(j)}_z$ are the Pauli-X and Pauli-Z operator acting on the $j$-th qubit, $A(t)$ and $B(t)$ are time-dependent control functions. The controllability of $A(t), B(t), h_j, J_{j,k}$ represents the programmability of QIMs, which would depend on the specific instantiation of QIM such as the D-Wave systems \citeMain{dwave-sys-doc}, QuEra neutral-atom system \citeMain{wurtz2022industry}, or otherwise.

At a high level, our Hamiltonian embedding technique is as follows: (i) discretize the QHD Hamiltonian \eqn{main1} to a finite-dimensional matrix; (ii) identify an invariant subspace $\mathcal{S}$ of the simulator Hamiltonian for the evolution;
(iii) program the simulator Hamiltonian \eqn{main2} so its restriction to the invariant subspace $\mathcal{S}$ matches the discretized QHD Hamiltonian. 
In this way, we effectively simulate the QHD Hamiltonian in the subspace $\mathcal{S}$ (called the \textit{encoding} subspace) of the simulator's full Hilbert space.
By measuring the encoding subspace at the end of the analog emulation, we obtain solutions to an optimization problem.

Precisely, consider the one-dimensional case of QHD Hamiltonian that is $\hat{H}(t) = e^{\varphi_t}(-\frac{1}{2}\frac{\partial^2}{\partial x^2})+e^{\chi_t}f$. 
Following a standard discretization by the finite difference method \citeMain{morton:numerical}, QHD Hamiltonian becomes $\hat{H}(t)= -\frac{1}{2}e^{\varphi_t}\hat{L}+e^{\chi_t}\hat{F}$ where the second-order derivative $\frac{\partial^2}{\partial x^2}$ becomes a tridiagonal matrix (denoted by $\hat{L}$), and the potential operator $f$ is reduced to a diagonal matrix (denoted by $\hat{F}$).

We identify the so-called \textit{Hamming encoding} subspace $\mathcal{S}_H$ which is spanned by $(n+1)$ \textit{Hamming states} $\{\ket{H_j}:j=0,1,\dots,n\}$ for any $n$-qubit QIM.\footnote{Precisely, 
The $j$-th Hamming state $\ket{H_j}$ is the uniform superposition of bitstring states with Hamming weight (i.e., the number of ones in a bitstring) $j$:
\begin{align*}
    \ket{H_j} = \frac{1}{\sqrt{C_j}}\sum_{|b|=j}\ket{b},
\end{align*}
where $C_j$ is the number of states with Hamming weight $j$. For example, there are $n$ bitstring states with Hamming weight $1$: $\ket{0\dots001}$,$\ket{0\dots010}$,$\dots$,$\ket{1\dots000}$, and the Hamming-$1$ state $\ket{H_1}$ is the uniform superposition of all the $n$ states.}
By choosing appropriate parameters $h_j$, $J_{j,k}$ in \eqn{main2}, the subspace $\mathcal{S}_H$ is invariant under the QIM Hamiltonian. Moreover, the restriction of the first term $\sum^r_{j=1} \sigma^{j}_x$ onto $\mathcal{S}_H$ resembles the tridiagonal matrix $\hat{L}$, and the restriction of the second term in the QIM Hamiltonian (with Pauli-Z and -ZZ operators) represents a discretized quadratic function $\hat{F}$. 
A measurement on $\mathcal{S}_H$ can be conducted by measuring the full simulator Hilbert space in the computational basis and simple post-processing. 
The Hamming encoding construction is readily generalizable to higher-dimensional Laplacian operator $\Delta$ and quadratic polynomial functions $f$. See \sec{analog-implementation} for details. 

Our Hamming encoding enables an empirical study of an interesting optimization problem called quadratic programming (QP) on quantum simulators. Specifically, we consider QP with box constraints:
\begin{subequations}
\begin{align}
    \text{minimize}&\qquad f(\x)=\frac{1}{2} \x^\top \mathbf{Q} \x + \mathbf{b}^\top \x,\\
    \text{subject to}
    &\qquad \mathbf{0} \preccurlyeq \x \preccurlyeq \mathbf{1},
\end{align}
\end{subequations}
where $\mathbf{0}$ and $\mathbf{1}$ are $n$-dimensional vectors of all zeros and all ones, respectively. QP problems are the simplest case of nonlinear programming and they appear in almost all major fields in the computational sciences \citeMain{nocedal1999numerical,dostal2009optimal}. Despite of their simplicity and ubiquity, non-convex QP problems (i.e., ones in which the Hessian matrix $\mathbf{Q}$ is indefinite) are known to be NP-hard \citeMain{burer:nonconvex} in general.

We implement QHD on the D-Wave system\footnote{We access the D-Wave \texttt{advantage\_system6.1} through \href{https://aws.amazon.com/braket/}{Amazon Braket}.}, which instantiates a QIM and allows for the control of thousands of physical qubits with decent connectivity~\citeMain{main:dwave-qpu-characteristics}. While existing libraries of QP benchmark instances (e.g., \citeMain{furini:qplib}) are natural candidates for our empirical study, most of them can not be mapped to the D-Wave system because of its limited connectivity. We instead create a new test benchmark with 160 randomly generated QP instances in various dimensions (5, 50, 60, 75) whose Hessian matrices are indefinite and sparse (see \sec{qp-test-problems} for details) and for which analog implementations are possible on the D-Wave machine (referred as DW-QHD).

We compare DW-QHD with 6 other state-of-the-art solvers in our empirical study: DW-QAA (baseline QAA implemented on D-Wave), IPOPT~\citeMain{kawajir2006introduction}, SNOPT~\citeMain{gill2005snopt}, MATLAB's \texttt{fmincon} (with \texttt{SQP} solver), QCQP~\citeMain{park:general}, and a basic Scipy \texttt{minimize} function (with \texttt{TNC} solver). 
In the two quantum methods (DW-QHD, DW-QAA), we discretize the search space $[0,1]^d$ into a regular mesh grid with $8$ cells per edge due to the limited number of qubits in the D-Wave machine. 
To compensate for the loss of resolution, we post-process the coarse-grained D-Wave results by the Scipy \texttt{minimize} function, which is a local gradient solver mimicking the descent phase of a higher-resolution QHD and only has mediocre performance by itself. 
The choice of classical solvers covers a variety of state-of-the-art optimization methods, including gradient-based local search (Scipy \texttt{minimize}), interior-point method (IPOPT), sequential quadratic programming (SNOPT, MATLAB), and heuristic convex relaxation (QCQP).
Finally, to investigate the quality of the D-Wave machine in implementing QHD and QAA, we also classically simulate QHD and QAA for the 5-dimensional instances (Sim-QHD, Sim-QAA).\footnote{Note that we numerically compute Sim-QHD and Sim-QAA for $t_f = 1 \mu s$, which is much shorter than the time we set in the D-Wave experiment (in DW-QHD and DW-QAA, we choose $t_f = 800 \mu s$.}

We use the \emph{time-to-solution} (TTS) metric \citeMain{ronnow:defining} to compare the performance of solvers. TTS is the number of trials (i.e., initializations for classical solvers or shots for quantum solvers) required to obtain the correct global solution\footnote{For each test instance, the global solution is obtained by Gurobi~\citeMain{gurobi}.} up to $0.99$ success probability:
\begin{align}
    \text{TTS} = t_f \times \Big\lceil\frac{\ln(1-0.99)}{\ln(1-p_s)}\Big\rceil,
\end{align}
where $t_f$ is the average runtime per trial, and $p_s$ is the success probability of finding the global solution in a given trial. We run 1000 trials per instance and compute the TTS for each solver.

\begin{figure}[ht!]
\centering
     \includegraphics[width=16cm]{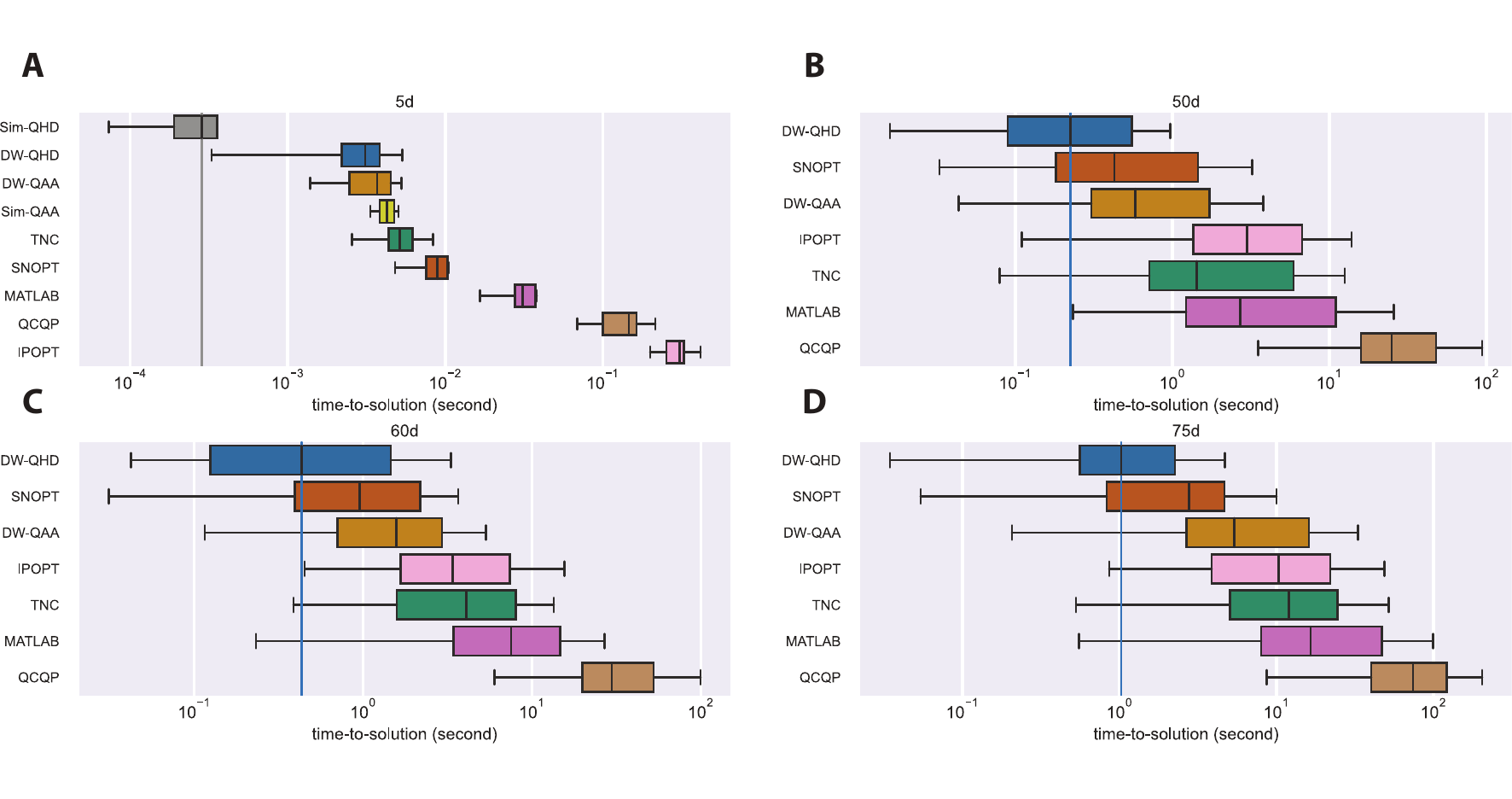}
     \caption{\textbf{Experiment results for quadratic programming problems.} Box plots of the time-to-solution (TTS) of selected quantum/classical solvers, gathered from four randomly generated quadratic programming benchmarks (\textbf{A}: 5-dimensional, \textbf{B}: 50-dimensional, \textbf{C}: 60-dimensional, \textbf{D}: 75-dimensional). The left and right boundaries of a box show the lower and upper quartiles of the TTS data, while the whiskers extend to show the rest of the TTS distribution. The median of the TTS distribution is shown as a black vertical line in the box. In each panel, the median line of the best solver extends to show the comparison with all other solvers.
     }
     \label{fig:fig4}
\end{figure}

In \fig{fig4}, we show the distribution of TTS for different solvers.\footnote{We also provide spreadsheets that summarize the test results for the QP benchmark, see \url{https://github.com/jiaqileng/quantum-hamiltonian-descent/tree/main/plot/fig4/qp_data}.} In the 5-dimensional case (\fig{fig4}A), Sim-QHD has the lowest TTS, and the quantum methods are generally more efficient than classical solvers. Note that with a much shorter annealing time ($t_f=1\mu s$ for Sim-QHD and $t_f=800\mu s$ for DW-QHD) Sim-QHD still does better than DW-QHD, indicating the D-Wave system is subject to significant noise and decoherence. Interestingly, Sim-QAA ($t_f=1\mu s$) is worse than DW-QAA ($t_f=800\mu s$), which shows QAA indeed has much slower convergence. In the higher dimensional cases (\fig{fig4}B,C,D), DW-QHD has the lowest median TTS among all tested solvers. 
Despite of the infeasibility of running Sim-QHD in high dimensions, 
our observation in the 5-dimensional case suggests that an ideal implementation of QHD could perform much better than DW-QHD, and therefore all other tested solvers in high dimensions. See \sec{qp-benchmark} for details. 

It is worth noting that DW-QHD does not outperform industrial-level nonlinear programming solvers such as Gurobi~\citeMain{gurobi} and CPLEX~\citeMain{cplex2009v12}. In our experiment, Gurobi usually solves the high-dimensional QP problems with TTS no more than $10^{-2}$ s. 
These solvers approximate the nonlinear problem by potentially exponentially many linear programming subroutines and use a branch-and-bound strategy for a smart but exhaustive search of the solution.\footnote{In \fig{branch_and_bound}, we show that the runtime of Gurobi scales exponentially with respect to the problem dimension for QP.}
However, the restrictions of the D-Wave machine (e.g., programmability and decoherence) force us to test on very sparse QP instances, which can be efficiently solved by highly-optimized industrial-level branch-and-bound solvers. 
On the other hand, we believe that QHD should be more appropriately deemed as a quantum upgrade of classical GD, which would more conceivably replace the role of GD rather than the entire branch-and-bound framework in classical optimizers.

\vspace*{-0.45em}
\section*{Conclusions}
In this work, we propose Quantum Hamiltonian Descent as a genuine quantum counterpart of classical gradient descent through a path-integral quantization of classical algorithms. 
Similar to classical GD, QHD is simple and efficient to implement on quantum computers. It leverages the quantum tunneling effect to escape from spurious local minima and has a superior performance than the standard quantum adiabatic algorithm, so we believe it could replace the role of classical GD in many optimization algorithms. 
Moreover, with the newly developed Hamiltonian embedding technique, we conduct a large-scale empirical study of QHD on non-convex quadratic programming instances up to 75 dimensions via an analog implementation of QHD on the D-Wave instantiation of a quantum Ising Hamiltonian simulator.  
We believe that QHD could be readily used as a benchmark algorithm for other quantum (e.g., \citeMain{wurtz2022industry,killoran2019strawberry}) or semi-quantum analog devices (e.g.,~\citeMain{inagaki2016coherent}) for both testing the quality of these devices and conducting more empirical study of QHD.

\vspace*{-0.45em}
\subsection*{Acknowledgment}
We thank Christian Borgs, Lei Fan, Shruti Puri, Andre Wibisono, Aram Harrow, Kyle Booth, Peter McMahon, Tamas Terlaky, Daniel Lidar, Wuchen Li, and Tongyang Li for helpful discussions at various stages of the development of this project. 

\bibliographystyleMain{myhamsplain}
\bibliographyMain{ref-main}

\cleardoublepage

\appendix
\numberwithin{equation}{section}

\begin{center}
{\bf\huge Supplementary Materials}    
\end{center}

\tableofcontents

\section{Derivation of QHD}\label{sec:derivation}
\subsection{Review of the Bregman-Lagrangian framework}\label{sec:connections}

The Bregman-Lagrangian framework is a variational formulation for accelerated gradient methods proposed by Wibisono, Wilson and Jordan \cite{wibisono2016variational}. This framework is formulated using Lagrangian mechanics. By defining the Lagrangian function

\begin{align}\label{eqn:bregman_lagrangian}
    \mathcal{L}(X,\dot{X},t) = e^{\alpha_t + \gamma_t}\left(\underbrace{\frac{1}{2}|e^{-\alpha_t }\dot{X}|^2}_\text{kinetic energy}  - \underbrace{e^{\beta_t}f(X)}_\text{potential energy}\right),
\end{align}
where $t \ge 0$ is the time, $X \in \R^d$ is the position, $\dot{X} \in \R^d$ is the velocity, and $\alpha_t, \beta_t, \gamma_t$ are arbitrary smooth functions that control the damping of energy in the system.  

Lagrangian mechanics is formulated by variational principle. For a trajectory of motion \\
$X(t)\colon [0,T] \to \R^d$, we define a functional $S$ as the \textit{action} of this trajectory

\begin{align}
    S[X(t)] = \int^T_{0} \mathcal{L}(X, \dot{X}, t)~ \d t,
\end{align}
where $\mathcal{L}$ is the Lagrangian of the system. A \textit{physical} path $X(t)$ from $a$ to $b$ is the curve $X(t)$ such that $X(0) = a$, $X(T) = b$, and it minimizes the action functional $S[X(t)]$. By calculus of variations, a least-action curve necessarily solves the Euler-Lagrange equation
\begin{align}\label{eqn:euler_lagrange}
    \frac{\d}{\d t} \left(\frac{\partial \mathcal{L}}{\partial \dot{X}}\right) - \frac{\partial \mathcal{L}}{\partial X} = 0.
\end{align}

With the Lagrangian function \eqn{bregman_lagrangian}, the resulting Euler-Lagrange equation is a second-order differential eqaution,
\begin{align}\label{eqn:second-order-ode}
    \ddot{X} + (\dot{\gamma}_t - \dot{\alpha}_t) \dot{X} + e^{2\alpha_t+\beta_t} \nabla f(X)=0.
\end{align}

The convergence of the Bregman-Lagrangian framework for continuously differentiable convex functions is established by constructing a Lyapunov function. Suppose the global minimizer of $f$ is $x^*$, we define the following function $\mathcal{E}_t$:
\begin{align}
    \mathcal{E}_t = \frac{1}{2}\left|e^{-\gamma_t}\dot{X}_t + X_t-x^*\right|^2 + e^{\beta_t}\left(f(X_t) - f(x^*)\right).
\end{align}
When $f$ is convex, it is shown that $\mathcal{E}_t$ is a Lyapunov function of \eqn{second-order-ode}, i.e., non-increasing along the solution trajectory $X_t$. The monotonicity of $\mathcal{E}_t$ leads to the convergence result \cite[Theorem 2.1]{wibisono2016variational},
\begin{align}\label{eqn:bregman-convergence}
    f(X(t)) - f(x^*) \le O(e^{-\beta_t}).
\end{align}

At a first glance, \eqn{bregman-convergence} says that any desired convergence rate can be achieved by choosing different $\beta_t$. This is true for continuous-time flows, as different time-dependent parameters $\alpha_t,\beta_t,\gamma_t$ result in the same path in spacetime while converging at different speeds\footnote{This invariance is called ``time dilation'' and it also holds in the quantum case, see \sec{time_dilation}.}. However, it does not imply that gradient-based methods generated by this dynamics can achieve any arbitrary convergence rate, which contradicts the fact that gradient-based methods can not converge faster than $O(t^{-2})$ in the worst case~\cite{nesterov2013introductory}. The devil is in time discretization: when translating the ODE model to practical gradient-based algorithms, we have to discretize the continuous time into discrete steps. The authors devise a family of gradient-based algorithms via an \textit{ad-hoc} discretization of \eqn{second-order-ode}, and these algorithms fail to convergence when $e^{-\beta_t}$ decays too fast. This matches our intuition that gradient-based optimization is bounded in convergence rate.

It is worth noting that the variational framework in \cite{wibisono2016variational} can be reformulated via Hamiltonian mechanics. The Lagrangian system is equivalently described by its Hamiltonian, the Legendre conjugate of the Lagrangian. In particular, the Hamiltonian corresponding to \eqn{bregman_lagrangian} takes the formal
\begin{align}\label{eqn:bregman_ham}
    \mathcal{H}(X,P,t) = e^{\alpha_t+\gamma_t}\left(\underbrace{\frac{1}{2}|e^{-\gamma_t }P|^2}_\text{kinetic energy}  + \underbrace{e^{\beta_t}f(X)}_\text{potential energy}\right),
\end{align}
where $X$ is the position, and $P$ is the momentum. This Hamiltonian function has the form of the sum of the kinetic and potential energy. The dynamics of the Hamiltonian system is then given by the Hamilton's equations,

\begin{subequations}\label{eqn:hamilton}
    \begin{align}
        \frac{\d X}{\d t} &= \frac{\partial \mathcal{H}}{\partial P} = e^{\alpha_t-\gamma_t} P,\label{eqn:hamilton-a}\\
        \frac{\d P}{\d t} &= -\frac{\partial \mathcal{H}}{\partial X} = - e^{\alpha_t+\beta_t+\gamma_t} \nabla f(X). \label{eqn:hamilton-b}
    \end{align}
  \end{subequations}
  One can show this system of ODEs is identical to the Euler-Lagrange equation \eqn{second-order-ode}.

\vspace{4mm}
Besides the variational formulation of accelerated gradient descent in \cite{wibisono2016variational}, Maddison et. al. \cite{maddison2018hamiltonian} propose a family of gradient-based methods via discretizations of conformal Hamiltonian dynamics, known as \textit{Hamiltonian descent methods}. This framework assumes an additional access to a kinetic energy $k$ that incorporates information about $f$. The resulting continuous-time trajectories achieve linear convergence on convex functions: $f(x_t) - f(x^*) \le O(e^{-ct})$, where $c$ depends on $k$ and $f$. They consider one implicit and two explicit discretization schemes in order to obtain gradient-based optimization algorithms. These algorithms exhibit similar convergence rates as the continuous-time flows.

\subsection{Derivation of QHD}\label{sec:quantization}

Now, we quantize the Lagrangian formulation of accelerated gradient methods by introducing the path integral formulation of quantum mechanics. For simplicity, the derivation works with a one-dimensional objective function $f\colon \R \to \R$, while the generalization to higher dimensions is trivial.

In classical mechanics, only the curves that solve the Euler-Lagrange equation are of interests because they are predicted by the variational principle. All other curves are considered ``unphysical'' because they are not stationary points of the action function $S[X(t)]$. 

For quantum mechanics, however, Feynman postulates that not only the ``physical'' trajectory but all trajectories contribute to the quantum evolution. These trajectories contribute equal magnitudes but different phases to the total amplitude. More precisely, the probability to go from a point $a$ at time $t_a$ to the point $b$ at time $t_b$ is $P(b,a) = |K(b, t_b;a, t_a)|^2$, where the amplitude function $K(b, t_b;a, t_a)$ is the sum of contributions of all paths from $a$ to $b$:

\begin{align}\label{eqn:propagator}
    K(b, t_b;a, t_a) \propto \sum_{\text{paths from $a$ to $b$}} e^{(i/\hbar)S[X(t)]}.
\end{align}
Here, $i$ is the imaginary unit and $\hbar$ is the Planck constant. The amplitude $K(b, t_b;a, t_a)$ is also known as the \textit{propagator} of the quantum dynamics because it can be used to compute the evolution of the wave function from $t_a$ to $t_b$:

\begin{align}\label{eqn:kernel_equation}
    \Psi(x,t_b) = \int_{\R} K(x, t_b;y,t_a) \Psi(y, t_a)~\d y.
\end{align}

To quantize accelerated gradient methods, we begin with the Lagrangian formulation,

\begin{align}\label{eqn:acc-lagrangian}
    \mathcal{L}(X,\dot{X},t) =  e^{-\varphi_t}\left(\frac{1}{2}|\dot{X}|^2\right)  - e^{\chi_t}f(X),
\end{align}
where $\varphi_t$ and $\chi_t$ are time-dependent functions that control the energy dissipation in the system.\footnote{In \cite{wibisono2016variational}, the authors introduced three time-dependent functions $\alpha_t, \beta_t, \gamma_t$ in the Bregman Lagrangian framework. Here, we use a simplified description with only two time-dependent parameters. The original formulation can be recovered by setting $\varphi_t = \alpha_t - \gamma_t$, $\chi_t = \alpha_t + \beta_t +  \gamma_t$.}

Suppose the quantum particle at time $t$ is described by the wave function $\Psi(x,t)$. To get an differential equation of $\Psi(x,t)$, we consider an infinitesimal time interval $[t, t+\epsilon]$. In this short time, the action $S[X(t)]$ is approximately $\epsilon$ times the Lagrangian, 

\begin{align}
    S[X(t)] = \epsilon \mathcal{L}\left(\frac{a+b}{2},\frac{b-a}{\epsilon},t\right),
\end{align}
which is correct correct to first order in $\epsilon$ \cite[2.5]{feynman2010quantum}. And the propagator can be evaluated by 

\begin{align}\label{eqn:first-order-propagator}
    K(x,t+\epsilon;y,t) = \frac{1}{A} \exp{\frac{i}{\hbar}\epsilon \mathcal{L}\left(\frac{x+y}{2},\frac{x-y}{\epsilon},t\right)},
\end{align}
where $A$ is a normalization factor that will be specified later. Plugging \eqn{first-order-propagator} into the \eqn{kernel_equation}, we obtain the following equation:

\begin{align}\label{eqn:infinitesimal}
    \Psi(x,t+\epsilon) = \frac{1}{A}\int^\infty_{-\infty} \exp{\frac{i}{\hbar} \epsilon \mathcal{L}\left(\frac{x+y}{2}, \frac{x-y}{\epsilon},t\right)} \Psi(y,t) ~\d y.
\end{align}

In the infinitesimal time interval, the smooth time-dependent functions  $\varphi_t$, $\chi_t$ in \eqn{acc-lagrangian} can be treated as constant functions. We plug the optimization Lagrangian \eqn{acc-lagrangian} into \eqn{infinitesimal}, and introduce the change of variable $y = x+\eta$. It follows that,

\begin{align}
    \Psi(x,t+\epsilon) = \frac{1}{A} \int^\infty_{-\infty} \exp{\frac{i}{\hbar}\frac{e^{-\varphi_t} \eta^2}{2 \epsilon}}  \exp{-\frac{i}{\hbar} \epsilon f\left(x+\frac{\eta}{2}, t\right)} \Psi(x+\eta,t) ~\d \eta.
\end{align}
Note that we absorb the $e^{\chi_t}$ coefficient into potential field $f(x)$ and write them simply as $f(x,t)$. Then, we expand the wave function $\Psi$ to first order in $\epsilon$ and second order in $\eta$. Note that the term $\epsilon f(x+\eta/2,t)$ is replaced by $\epsilon f(x,t)$ since the error term is of higher order than $\epsilon$. It turns out that
\begin{align}\label{eqn:expansion}
    \Psi(x,t) + \epsilon \frac{\partial \Psi}{\partial t} = \frac{1}{A} \int^\infty_{-\infty} \exp{\frac{i}{\hbar}\frac{e^{-\varphi_t} \eta^2}{2\epsilon}} \left[1 - \frac{i}{\hbar} \epsilon f(x,t)\right] \left[\Psi(x,t) + \eta \frac{\partial \Psi}{\partial x} + \frac{\eta^2}{2} \frac{\partial^2 \Psi}{\partial x^2}\right]~\d \eta.
\end{align}

On the left-hand side of \eqn{expansion}, the $O(\epsilon^0)$\footnote{We use $O(\epsilon^k)$ to indicate the correction terms in the infinitesimal expansion \eqn{expansion}. For example, $O(\epsilon^0)$ means constant term, $O(\epsilon^1)$ means the first-order correction term, etc.} term is $\Psi(x,t)$; meanwhile, the $O(\epsilon^0)$ term on the right-hand side is $\Psi(x,t)$ times the coefficient:
\begin{align}\label{eqn:zero-order-coeff}
    \frac{1}{A}\int^\infty_{-\infty} \exp{\frac{i}{\hbar}\frac{e^{-\varphi_t} \eta^2}{2\epsilon}}~\d \eta = \frac{1}{A}(2\pi i \hbar e^{\varphi_t}\epsilon)^{1/2}.
\end{align}
To match the $O(\epsilon^0)$ term on both sides of \eqn{expansion}, we must have the coefficient \eqn{zero-order-coeff} equal to $1$, i.e., $A = (2\pi i \hbar e^{\varphi_t}\epsilon)^{1/2}$.

Next, we match the higher order terms in \eqn{expansion} (up to $O(\epsilon^1)$ and $O(\eta^2)$). With some algebraic manipulation, we end up with

\begin{align}
    \label{eqn:higher-order-expansion}
    \epsilon \frac{\partial \Psi}{\partial t} = C_1 \frac{\partial \Psi}{\partial x} + C_2 \frac{\partial^2 \Psi}{\partial x^2} + C_3 f\Psi,
\end{align}
where the coefficients $C_1$, $C_2$, $C_3$ can be explicitly evaluated:
\begin{align}
    C_1 &= \frac{1}{A}\int_\R \eta \exp{ie^{-\varphi_t} \eta^2/2 \hbar \epsilon} ~\d \eta = 0,\\
    C_2 &= \frac{1}{A}\int_\R \frac{\eta^2}{2} \exp{ie^{-\varphi_t} \eta^2/2 \hbar \epsilon} ~\d \eta = \frac{1}{2} i \hbar \epsilon e^{\varphi_t},\\
    C_3 &= - \frac{i}{\hbar}\epsilon.
\end{align}

Substituting the coefficients $C_1$, $C_2$, $C_3$ to \eqn{higher-order-expansion}, we obtain the QHD dynamics described by the Schr\"odinger equation,

\begin{align}\label{eqn:schrodinger}
    i\hbar \frac{\partial \Psi}{\partial t} = \left[e^{\varphi_t}\left(-\frac{\hbar^2}{2} \frac{\partial^2 }{\partial x^2}\right) + e^{\chi_t}f(x) \right]\Psi,
\end{align}
which defines the QHD dynamics.

\vspace{4mm}In higher dimensions, the kinetic operator will be replaced by $-\frac{\hbar^2}{2}\nabla^2$. We set $\hbar = 1$, and the general QHD Hamiltonian operator is given by

\begin{align}\label{eqn:qhd_operator}
    \hat{H}(t) = e^{\varphi_t}\left(-\frac{1}{2}\nabla^2\right) + e^{\chi_t} f(x).
\end{align}

In the rest of the paper, we sometimes also consider the QHD with three time-dependent parameters,
\begin{align}\label{eqn:qhd_operator_convex}
    \hat{H}(t) = e^{\alpha_t - \gamma_t}\left(-\frac{1}{2}\nabla^2\right) + e^{\alpha_t+\beta_t+\gamma_t} f(x),
\end{align}
where the parameters $\alpha_t$, $\beta_t$, and $\gamma_t$ satisfy the \textit{ideal scaling condition}:
\begin{subequations}\label{eqn:ideal_scaling}
    \begin{align}
        \dot{\beta}_t \le e^{\alpha_t},\label{eqn:condition-a}\\
        \dot{\gamma}_t = e^{\alpha_t}.\label{eqn:condition-b}
    \end{align} 
\end{subequations}
This setting is closely related to classical accelerated gradient descent because it has the same time-dependent parameters as in the variational formulation in \cite{wibisono2016variational}; see \eqn{bregman_ham}. As we will show in \sec{nonconvex}, quantum and classical gradient descent have radically different behavior even with the same time-dependent parameters. 

The QHD Hamiltonian can be seen as a weighted sum of the kinetic and potential operator. The time-dependent functions $\varphi_t$ and $\chi_t$ contribute to the limiting behavior of this quantum system. If these functions are constant in time, the system Hamiltonian $\hat{H}(t)$ is time-independent and the system energy is conserved. In this case, the wave function is indefinitely oscillatory. This is not what we want: as the word ``descent'' suggests, we want to gradually drain out the kinetic energy from the system so that the quantum particle will eventually land still on the rock bottom of the potential landscape $f(x)$. To achieve this goal, we will consider $\varphi_t$ as a decreasing function and $\chi_t$ as an increasing function so that the potential energy will dominate in the long run.

\paragraph{Connections to Quantum Dynamical Descent (QDD).}
Verdon et. al. propose the \textit{Quantum Dynamical Descent} (QDD) algorithm~\cite{verdon2018universal}, which is formally analogous to our QHD algorithm.
Given the similarity in their formulation, QHD and QDD have little in common in their intuition, derivation, and application. QDD is devised for quantum parametric optimization, which is a special case of continuous optimization. We derive QHD from first principles by quantizing the Bregman-Lagrangian framework, while QDD is constructed by heuristics. We have a systematic theoretical study of QHD, including the convergence for convex and non-convex problems; the analysis of QDD is bound to the adiabatic approximation framework. Moreover, we conduct a large-scale empirical study of QHD based on analog implementation. QDD is implemented on digital quantum computers as like a variational quantum algorithm. To our best knowledge, QDD has not been implemented on any real-world quantum computers.

\subsection{Time dilation}\label{sec:time_dilation}

In this section, we show that QHD is closed under time dilation, i.e., time-dilated QHD evolution is also described by the QHD equation, but with different time-dependent parameters. This means the continuous-time QHD can converge at any speed along the same evolution path. This result generalizes \cite[Theorem 2.2]{wibisono2016variational}.

The quantum evolution $\Psi(x,t)$ for $t\in[0,T]$ forms a curve in the Hilbert space $L^2(\R^d)$. We introduce a smooth increasing function $\tau\colon[0,T] \to \R$ to represent the reparametrization of time. The reparametrized wave function is 
\begin{align}\label{eqn:reparametrized}
    \widetilde{\Psi}(x,t) := \Psi(x, \tau(t)),
\end{align}
which is the same curve in the Hilbert space $L^2(\R^d)$ but with a different speed.

\begin{proposition}[Time dilation]
\label{prop:time_dilation} 
If $\Psi(x,t)$ satisfies the Schr\"odinger equation $i\frac{\d}{\d t}\ket{\Psi_t} = H(t)\ket{\Psi_t}$ with the QHD Hamiltonian 
\begin{align}
    H(t) = e^{\varphi_t}\left(-\frac{1}{2}\nabla^2\right) + e^{\chi_t} f(x),
\end{align}
then the reparametrized wave function $\tilde{\Psi}(x,t)$ defined as in \eqn{reparametrized} satisfies the Schr\"odinger equation $i\frac{\d}{\d t}\ket{\widetilde{\Psi}_t} = \widetilde{H}(t)\ket{\widetilde{\Psi}_t}$ with the time-dilated QHD Hamiltonian 
\begin{align}\label{eqn:rescale-h}
    \widetilde{H}(t) = e^{\varphi_{\tau_t}+\log(\dot{\tau})}\left(-\frac{1}{2}\nabla^2\right) + e^{\chi_{\tau_t}+\log(\dot{\tau})} f(x).
\end{align}
\end{proposition}
\begin{proof}
    \begin{align}
        i \frac{\d}{\d t}\widetilde{\Psi}(x,t) &= i\frac{\d}{\d t}\Psi(x,\tau(t)) = i\frac{\d}{\d \tau}\Psi(x,\tau(t))\frac{\d \tau}{\d t}\\
        &= H(\tau(t))\Psi(x,\tau(t))\dot{\tau}(t)\\
        &= \widetilde{H}(t)\widetilde{\Psi}(x,t),\label{eqn:rescale-last-step}
    \end{align}
    where $\widetilde{H}(t)$ in \eqn{rescale-last-step} is the same as \eqn{rescale-h}.
\end{proof}

\begin{corollary}
If we choose $\varphi_t = \alpha_t - \gamma_t$, $\chi_t = \alpha_t + \beta_t + \gamma_t$ (as in \eqn{qhd_operator_convex}), then the rescaled QHD Hamiltonian is also described by three parameters:
\begin{subequations}\label{eqn:param_dilation}
        \begin{equation}
            \tilde{\alpha}_t = \alpha_{\tau(t)} + \log(\dot{\tau}(t))
        \end{equation}
        \begin{equation}
            \tilde{\beta}_t = \beta_{\tau(t)}
        \end{equation}
        \begin{equation}
            \tilde{\gamma}_t = \gamma_{\tau(t)}.
        \end{equation}
      \end{subequations}
Furthermore, $\alpha, \beta, \gamma$ satisfy the ideal scaling condition \eqn{ideal_scaling} if and only if $\tilde{\alpha}, \tilde{\beta}, \tilde{\gamma}$ do.
\end{corollary}
\begin{proof}
    One can check the corollary by plugging $\varphi_t = \alpha_t - \gamma_t$ and $\chi_t = \alpha_t + \beta_t + \gamma_t$ into \eqn{rescale-h}.
\end{proof}

\begin{remark}
Although we show that QHD convergence speed can be arbitrary fast by dilating the time-dependent functions in the QHD Hamiltonian, it does not mean we have a quantum algorithm that converges at arbitrary fast rate. Too fast time-dependent functions in QHD can make the dynamics unstable in time discretization (in digital implementation) or analog emulation (in analog implementation), thus unable to solve the optimization problem.
\end{remark}

\section{Convergence of QHD}\label{sec:convergence}

In this section, we discuss the convergence of QHD for optimization problems. In \sec{convergence-convex}, we show that QHD has fast convergence in convex optimization problems. This quantum convergence can be seen as a generalization of the classical convergence rate of accelerated gradient descent algorithms \cite{wibisono2016variational}. Besides the convex case, we also prove a global convergence result for QHD under mild assumptions of the objective $f$, see \sec{convergence-nonconvex}. To our best knowledge, the global convergence behavior is not observed in the classical counterpart of QHD. Both convergence results in this section are formulated in continuous time. 

\paragraph{Notations.}
We denote the position and momentum operators as (choosing $\hbar = 1$):
\begin{align}
    \hat{x} = x,~\hat{p} = -i\nabla.
\end{align}
Given a quantum observable $\hat{O}$, its expectation value at time $t$ with respect to the quantum wave function $\ket{\Psi_t}$ is computed by
\begin{align}\label{eqn:expect_O}
    \<\hat{O}\>_t \coloneqq \braket{\Psi_t}{\hat{O}\Big|\Psi_t} = \int \overline{\Psi(t,x)}\hat{O}\Psi(t,x)~\d x. 
\end{align}
In particular, when $\hat{O} = f(x)$ is the objective function, we define
\begin{align}
    \mathbb{E}[f]_{\sim \Psi_t} \coloneqq \<f\>_t
\end{align}
as the (average) loss function at time $t$. 

\subsection{Fast convergence in the convex case}\label{sec:convergence-convex}
To recover the convergence rate shown in~\cite{wibisono2016variational} for QHD, we consider the QHD Hamiltonian with three parameters \eqn{qhd_operator_convex} and these parameters satisfy the ideal scaling condition \eqn{ideal_scaling}.

\begin{restatable}{theorem}{convex-convergence}\label{thm:convex_convergence}
    Assume that $f$ is a continuously differentiable convex function and the ideal scaling condition \eqn{ideal_scaling} holds. Let $x^*$ be the unique local minimizer of $f$. Then, for any smooth initial wave function $\ket{\Psi_0}$, the solution $\ket{\Psi_t}$ to the Schr\"odinger equation $i\ket{\Psi_t} = \hat{H}(t)\ket{\Psi_t}$ with Hamiltonian \eqn{qhd_operator_convex} satisfies
        \begin{align}
            \mathbb{E}[f]_{\sim \Psi_t} - f(x^*) \le O(e^{-\beta_t}).
        \end{align}
\end{restatable}
\begin{proof}
    Without loss of generality, we may assume $x^* = 0$ and $f(x^*)=0$. It suffices to prove that 
    \begin{align}
        \mathbb{E}[f]_{\sim \Psi_t}\le O(e^{-\beta_t}).
    \end{align}
    We take a Lyapunov function approach to prove \thm{convex_convergence}. We construct the following quantum Lyapunov function:
    \begin{align}\label{eqn:quantum-energy-functional}
        \mathcal{W}(t)= \<\hat{J}^2/2\>_t +  e^{\beta_t}\<f\>_t,
    \end{align}
    in which we introduce the new operator $\hat{J}\coloneqq e^{-\gamma_t} \hat{p} + \hat{x}$. $\hat{J}$ is a legal quantum observable because both $\hat{p}$ and $\hat{x}$ are Hermitian operators. In \prop{Wdot}, we show that $\mathcal{W}(t) \le \mathcal{W}(0)$ for any $t \ge 0$. Meanwhile, notice that $\hat{J}^2/2$ is a positive-definite operator, so $\<\hat{J}^2/2\>_t \ge 0$, and we have:
    \begin{align}
        e^{\beta_t}\mathbb{E}[f]_{\sim\Psi_t} \le \mathcal{W}(t) \le \mathcal{W}(0),
    \end{align}
    or equivalently, 
    \begin{align}
        \mathbb{E}[f]_{\sim\Psi_t} \le \mathcal{W}(0)e^{-\beta_t} \le O(e^{-\beta_t}).
    \end{align}
\end{proof}

\begin{lemma}\label{lem:der_expect}
    Given a (time-dependent) quantum observable $\hat{O}(t)$ and let $\ket{\Psi_t}$ be the solution of the Schr\"odinger equation $i \ket{\Psi_t} = \hat{H}(t)\ket{\Psi_t}$, we have
    \begin{align}
        \frac{\d}{\d t}\<\hat{O}(t)\>_t = \<\frac{\d}{\d t}\hat{O}(t)\>_t + \<i[\hat{H}(t), \hat{O}(t)]\>_t.
    \end{align}
\end{lemma}
\begin{proof}
    \begin{align}
        \frac{\d}{\d t}\<\hat{O}(t)\>_t &= \frac{\d}{\d t}\braket{\Psi_t}{\hat{O}(t)\Big|\Psi_t}\\
        &= \braket{\dot{\Psi}_t}{\hat{O}(t)\Big|\Psi_t} + \braket{\Psi_t}{\hat{O}(t)\Big|\dot{\Psi}_t} + \braket{\dot{\Psi}_t}{\frac{\d}{\d t}\hat{O}(t)\Big|\Psi_t}\\
        &= i \braket{\Psi_t}{\hat{H}(t)\hat{O}(t)\Big|\Psi_t} - i \braket{\Psi_t}{\hat{O}(t)\hat{H}(t)\Big|\Psi_t} + \<\frac{\d}{\d t}\hat{O}(t)\>_t\\
        &= \braket{\Psi_t}{i[\hat{H}(t),\hat{O}(t)]\Big|\Psi_t} + \<\frac{\d}{\d t}\hat{O}(t)\>_t.
    \end{align}
\end{proof}

\begin{lemma}[Commutation relations]\label{lem:commutation}
    Let $\hat{x}$, $\hat{p}$ be the position and momentum operators, we have 
    \begin{subequations}
    \begin{align}
        i[\hat{p}^2, \hat{x}^2] &= 4\hat{x}\hat{p} - 2i,\label{eqn:commutation_a}\\
        i[\hat{p}^2, \hat{x}\hat{p}] &= 2\hat{p}^2,\label{eqn:commutation_b}\\
        i[f, \hat{x}\hat{p}] &= -\hat{x}\cdot \nabla f.\label{eqn:commutation_c} 
    \end{align}
    \end{subequations}
\end{lemma}
\begin{proof}
    We note that $[\hat{x},\hat{p}] = i$, or equivalently, $\hat{p}\hat{x} = \hat{x}\hat{p} - i$. To show \eqn{commutation_a}, we compute:
    \begin{align*}
        i[\hat{p}^2, \hat{x}^2] &= i \left(\hat{p}^2 \hat{x}^2 - \hat{x}^2\hat{p}^2\right) =i \left(\hat{p}(\hat{x}\hat{p}-i)\hat{x} - \hat{x}^2\hat{p}^2\right)\\
        &= i\left(\hat{p}\hat{x}(\hat{x}\hat{p}-i) - i\hat{p}\hat{x} - \hat{x}^2\hat{p}^2\right)\\
        &= i\left(\hat{p}\hat{x}^2\hat{p} - 2i\hat{p}\hat{x} - \hat{x}^2\hat{p}^2\right)\\
        &= i\left((\hat{x}\hat{p}-i)\hat{x}\hat{p} - 2i\hat{p}\hat{x} - \hat{x}^2\hat{p}^2\right)\\
        &= i\left(\hat{x}(\hat{x}\hat{p}-i)\hat{p} - i \hat{x}\hat{p} -2i\hat{p}\hat{x} - \hat{x}^2\hat{p}^2\right)\\
        &= 2\hat{x}\hat{p} + 2\hat{p}\hat{x} = 4\hat{x}\hat{p} - 2i.
    \end{align*}
    
    Similarly, one can exploit the commutation relation of $\hat{x}$ and $\hat{p}$ to show \eqn{commutation_b}. To show \eqn{commutation_c}, we introduce a smooth $L^2$ test function $\varphi$:
    \begin{align*}
        i[f, \hat{x}\hat{p}]\varphi &= i\left(f\hat{x}(-i\nabla) \varphi - \hat{x}(-i\nabla)(f\varphi)\right)\\
        &= f\hat{x}\nabla \varphi - \hat{x}\left(\nabla f \varphi + f \nabla \varphi\right)\\
        &= - \hat{x}\nabla f \varphi,
    \end{align*}
    so it follows that $i[f, \hat{x}\hat{p}] = - \hat{x}\nabla f$.
\end{proof}

\begin{proposition}\label{prop:Wdot}
    The function $\mathcal{W}(t)$ is non-increasing in time, i.e., $\mathcal{W}(t) \le \mathcal{W}(0)$ for any $t \ge 0$. 
\end{proposition}
\begin{proof}
    By \lem{der_expect}, we have
    \begin{align}
        \frac{\d}{\d t}\mathcal{W}(t) &= \frac{\d}{\d t}\<e^{-2\gamma_t}\hat{p}^2/2 + \hat{x}^2/2 + e^{-\gamma_t}(\hat{x}\hat{p}-i/2) + e^{\beta_t}f\>_t\\
        &= \<-\dot{\gamma}_t e^{-2\gamma_t} (\hat{p}^2) - \dot{\gamma}_te^{-\gamma_t}(\hat{x}\hat{p}-i/2) + \dot{\beta}_t e^{\beta_t}f\>_t \\
        &+ \<i[H(t), e^{-2\gamma_t}\hat{p}^2/2 + \hat{x}^2/2 + e^{-\gamma_t}(\hat{x}\hat{p}-i/2) + e^{\beta_t}f]\>_t,\label{eqn:W_commutator}
    \end{align}
    where $H(t) = e^{\alpha_t-\gamma_t}\hat{p}^2/2 + e^{\alpha_t+\beta_t+\gamma_t}f$ is the QHD Hamiltonian as in \eqn{qhd_operator_convex}. We expand the commutator \eqn{W_commutator} and simplify it using \lem{commutation}:
    \begin{align}
        &i[H(t), e^{-2\gamma_t}\hat{p}^2/2 + \hat{x}^2/2 + e^{-\gamma_t}(\hat{x}\hat{p}-i/2) + e^{\beta_t}f]\\ 
        &= i\Big(\cancel{\frac{1}{2}e^{\alpha_t+\beta_t-\gamma_t}[f,\hat{p}^2]} + \frac{1}{4}e^{\alpha_t-\gamma_t}[\hat{p}^2,\hat{x}^2]\\ &+ \frac{1}{2}e^{\alpha_t-2\gamma_t}[\hat{p}^2, \hat{x}\hat{p}] + e^{\alpha_t+\beta_t}[f,\hat{x}\hat{p}]+\cancel{\frac{1}{2}e^{\alpha_t+\beta_t-\gamma_t}[\hat{p}^2,f]}\Big)\\
        &= e^{\alpha_t -\gamma_t}(\hat{x}\hat{p}-i/2) + e^{\alpha_t-2\gamma_t}\hat{p}^2 + e^{\alpha_t + \beta_t}\left(-\hat{x}\cdot \nabla f\right).\label{eqn:H_comm}
    \end{align}
    
    Inserting \eqn{H_comm} into \eqn{W_commutator}, we have:
    \begin{align}
        \frac{\d}{\d t}\mathcal{W}(t) = \left(e^{\alpha_t}-\dot{\gamma}_t\right)\left[e^{-2\gamma_t}\<\hat{p}^2\>_t + e^{-\gamma_t}\<\hat{x}\hat{p}-i/2\>_t\right] + \dot{\beta}_te^{\beta_t}\<f\>_t + e^{\alpha_t + \beta_t}\<-\hat{x}\cdot \nabla f\>_t.
    \end{align}
    
    By the ideal scaling condition \eqn{ideal_scaling}, we have $e^{\alpha_t}-\dot{\gamma}_t = 0$ and $\dot{\beta}_t \le e^{\alpha_t}$, therefore, we have
    \begin{align}
        \frac{\d}{\d t}\mathcal{W}(t) \le e^{\alpha_t+\beta_t}\<f - x\cdot \nabla f\>_t.
    \end{align}
    Note that $f$ is continuously differentiable and convex, and $x^* = 0$, we have $f(x) \le x \cdot \nabla f(x)$ for any $x \in \R^d$. We conclude that $\frac{\d}{\d t}\mathcal{W}(t) \le 0$,
\end{proof}

\subsection{Global convergence in the non-convex case}\label{sec:convergence-nonconvex}

Non-convex optimization problems usually have many local minimum, and the theoretical result for convex problems does not naturally generalize. Nevertheless, we can prove the global convergence of QHD for non-convex problems under certain realistic assumptions. 

In this section, we consider the QHD Hamiltonian with two parameters \eqn{qhd_operator} because the convergence rate is not guaranteed for non-convex problems. Also, for simplicity, we only consider the one-dimensional case, i.e., $f\colon \R\to\R$. However, our argument is readily generalized to arbitrary finite dimensions. We denote $\Pi_N(t)$ as the orthogonal projection onto the subspace spanned by the first $N$ eigenstates of the operator $\hat{H}(t)$, given a fixed positive integer $N$.

\begin{assumption}\label{assump:obj}
The objective function $f$ is continuously differentiable, unbounded at infinity, and has a unique global minimum at $x^*$. Without loss of generality, we assume $f(x^*)=0$.
\end{assumption}

\begin{assumption}\label{assump:semiclassical}
The time-dependent parameters are slow-varying in time (i.e., $|\dot{\varphi}_t|$,$|\dot{\chi}_t| \ll 1$) and $\lim_{t\to\infty}e^{\varphi_t-\chi_t} = 0$. The Hamiltonian $\hat{H}(t)$ has no crossing eigenvalues for $t \ge 0$. Moreover, the semi-classical approximation holds. This means that in the regime $e^{\varphi_t - \chi_t} \ll 1$, the first $N$ eigenstates of $\hat{H}(t)$ are approximated by the first $N$ eigenstates of the quantum Harmonic oscillator:
\begin{align}
    \hat{H}_{HO}\coloneqq e^{\varphi_t}\left(-\frac{1}{2}\nabla^2\right) + \frac{e^{\chi_t}}{2}f''(x^*)(x-x^*)^2.
\end{align}
\end{assumption}
\begin{remark}
    In what follows, we will call $f_q\coloneqq f(x^*) + \frac{1}{2}f''(x^*)(x-x^*)^2$ as the quadratic model of the function $f$.
\end{remark}

\begin{assumption}\label{assump:init}
The initial wave function $\ket{\Psi_0}$ is in the low-energy subspace of the QHD Hamiltonian at $t=0$, i.e., 
\begin{align}
    \Pi_N(0) \ket{\Psi_0} = \ket{\Psi_0}.
\end{align}
\end{assumption}

Now, we justify why the assumptions we made are realistic for many non-convex optimization problems.

\begin{itemize}
    \item \assump{obj} is a very standard assumption on the objective function in optimization theory. We assume $f(x)$ has a unique global minimum to avoid degenerate cases when we do semi-classical approximation. It is worth mentioning that the uniqueness of global minimum is just a technical assumption and it does not mean QHD can not solve optimization problems with multiple global minima!
    \item The semi-classical approximation in \assump{semiclassical} is a well-known result in the spectral theory of Schr\"odinger operators. For a Schr\"odinger operator $H = -\hbar \Delta + \frac{1}{\hbar}f$, it is shown that when $\hbar \to 0$, the low-energy spectrum is well approximated by that of a quantum harmonic oscillator $H' = -\hbar \Delta + \frac{1}{\hbar}f_q$ (i.e., the potential field $f$ is replaced by its quadratic model), see \cite[Theorem 11.3, informal]{hislop2012introduction}. In our assumption, we take a stronger form so not only the low-energy spectrum but also the low-energy eigenstates of $H$ are approximated by those of $H'$ as well. In \fig{low_energy_subspace}, we show the low-energy subspace of the Schr\"odinger operator $H$ and its semi-classical limit $H'$. The numerical evidence suggests that our assumption is valid in practical problems.
    \item Though it is difficult to prepare $\ket{\Psi_0}$ to satisfy \assump{init} without enough knowledge of $f$, we observe that $(I-\Pi_N(0))\ket{\Psi_0}$ is often very small for $N \approx 10$. For example, in our 2D experiments, we usually have $\|(I-\Pi_N(0))\ket{\Psi_0}\|\le 10^{-4}$ for $\ket{\Psi_0}$ as a uniform superposition state. The reason is that the eigenstates of $\hat{H}(0)$ form a complete basis set, so it is always possible to find an integer $N$ such that the projection of $\ket{\Psi_0}$ onto the low-energy subspace is close enough to $\ket{\Psi_0}$ itself.
\end{itemize}

\begin{figure}[ht]
    \centering
    \includegraphics[width=16cm]{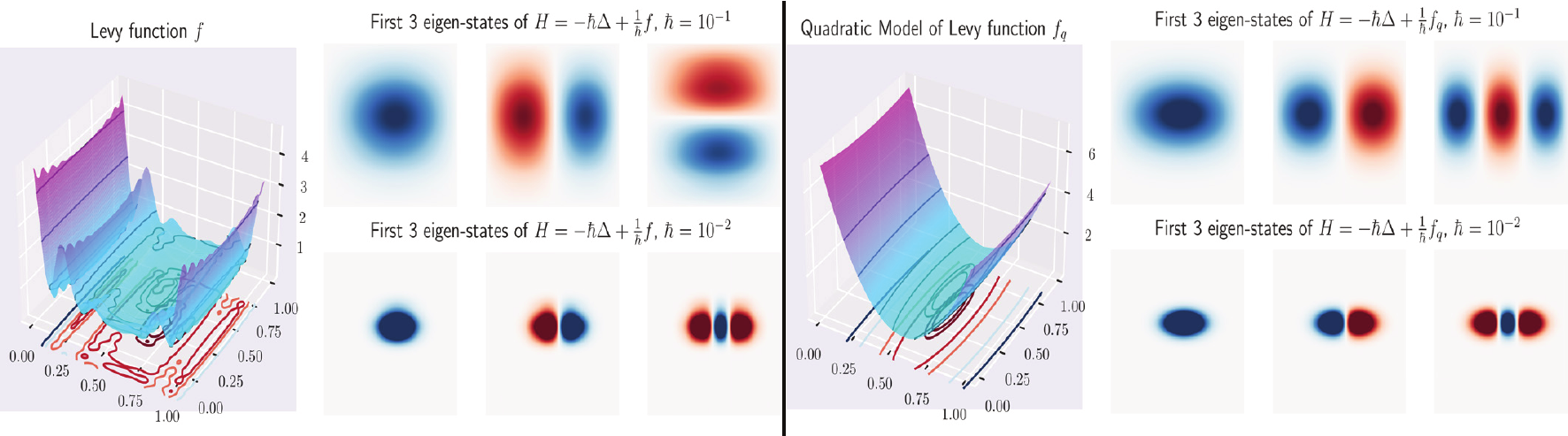}
    \caption[Low-energy subspace of QHD]{The low-energy subspace of the Schr\"odinger operator, illustrated with Levy function $f$ and its quadratic model $f_q$. The ground state, the first excited state, and the second excited eigenstate of the Schr\"odinger operator are plotted for $\hbar = 0.1, 0.01$. When $\hbar$ is small, it is clear that the low-energy subspae of the Schr\"odinger operator converges to its semi-classical limit (i.e., with the quadratic model as potential field).}
    \label{fig:low_energy_subspace}
\end{figure}

\begin{theorem}\label{thm:adiabatic_limit_qhd}
    Suppose \assump{obj}, \assump{semiclassical}, and \assump{init} are satisfied. Then, QHD will converge to the global minimum of $f$ when $t$ is large enough, i.e., 
    \begin{align}
        \lim_{t \to \infty} \mathbb{E}[f]_{\sim \Psi_t} = f(x^*).
    \end{align}
\end{theorem}

\begin{remark}
    In the proof of the theorem, we invoke the adiabatic approximation. Since the same technique is also used in the analysis of Quantum Adiabatic Algorithm (QAA)\footnote{For more details, see \sec{qaa}.}, we are no worse than QAA (in terms of convergence time) in the worst case. In fact, as \cite{nenciu1993linear} suggests, QHD is likely to converge much faster than the standard QAA because the QHD Hamiltonian is more regular.
\end{remark}

\begin{proof}
    Since $f$ is continuously differentiable and unbounded at infinity, by \cite[Theorem 10.7]{hislop2012introduction}, the spectrum of $\hat{H}(t)$ is purely discrete for all $t \ge 0$. Denote the eigen-pairs of $H(t)$ by $\{E_n(t), \ket{n;t}\}_n$. 
    We now represent the QHD wave function in this basis:
        \begin{align}\label{eqn:ansatz}
        \ket{\Psi_t} = \sum_n c_n(t) e^{i \theta_n(t)} \ket{n;t},
    \end{align}
    where $\theta_n(t)$ are real-valued functions.
    
    Inserting \eqn{ansatz} into the Schr\"odinger equation $i \partial_t \ket{\Psi_t} = \hat{H}(t)\ket{\Psi_t}$, we obtain the following equations:

    \begin{align}\label{eqn:ode_form}
        \dot{c}_m(t) = \underbrace{-c_m(t) \bra{m;t}\left[\frac{\partial}{\partial t} \ket{m;t}\right]}_{\text{Part~I.}} - \underbrace{\sum_{n\neq m} c_n(t) e^{i(\theta_n - \theta_m)} \frac{\bra{m;t}\dot{H}\ket{n;t}}{E_n(t) - E_m(t)}}_{\text{Part~II.}}.
    \end{align}

    No crossing eigenvalues in $\hat{H}(t)$ implies that $E_n(t) - E_m(t) \neq 0$. By the slow-varying assumption on the time-dependent parameters, we have $\braket{m;t}{\dot{H}\Big|n;t} \ll 1$. Now, we invoke the adiabatic approximation \cite[Section 5.6]{sakurai2011modern} and neglect Part II in \eqn{ode_form}. 

    As for Part I in \eqn{ode_form}, we note that $\bra{m;t}\left[\frac{\partial}{\partial t} \ket{m;t}\right]$ is purely imaginary for all $m$. To see this, we differentiate the identity $\braket{m;t}{m;t}=1$ w.r.t the time $t$ and it turns out that
    \begin{align}
        0 = \frac{\partial}{\partial t} \braket{m;t}{m;t} = 2 \Re\left\{\bra{m;t}\left[\frac{\partial}{\partial t} \ket{m;t}\right] \right\}.
    \end{align}
    Therefore, \eqn{ode_form} is simplified to the following first-order linear ODE:
    \begin{align}\label{eqn:simplified_ode}
        \dot{c}_m(t) = -i c_m(t) \mathscr{A}_m(t),
    \end{align}
    in which $\mathscr{A}_m(t)$ are real-valued functions.
    This simplification leads to a closed-form solution for $c_m(t)$:
    \begin{align}
        c_m(t) = c_m(0) e^{-i\int^t_0 \mathscr{A}(s)\d s},
    \end{align}
    which implies that $|c_m(t)| = |c_m(0)|$ for all $t \ge 0$. Due to \assump{init}, we conclude that $\ket{\Psi_t}$ remains in the low-energy subspace during the evolution:
    \begin{align}
        \ket{\Psi_t} = \sum_{m < N} c_m(t) e^{i\theta_m(t)} \ket{m;t}.
    \end{align}
    With this finite expansion of $\ket{\Psi_t}$, we compute the expectation value:
    \begin{align}\label{eqn:expectation_f}
        \mathbb{E}[f]_{\sim \Psi_t} = \braket{\Psi_t}{f|\Psi_t} = \sum_{m < N} |c_m|^2 \braket{m;t}{f|m;t} + \sum_{\substack{m\neq n\\m,n < N}} \overline{c_m}c_n e^{-i(\theta_m-\theta_n)} \braket{m;t}{f|n;t}.
    \end{align}
    
    For large enough $t$, the parameters $e^{\varphi_t-\chi_t}$ is so small that we can invoke the semi-classical approximation. We denote the eigenstates of $\hat{H}_{HO}(t)$ as $\ket{e_n;t}$. By \assump{semiclassical}), we approximate the eigenstates of $\hat{H}(t)$ by those of $\hat{H}_{HO}(t)$. We write $V = \frac{e^{\chi_t}}{2}f''(x^*)(x-x^*)^2$, by \lem{expect_potential} (using $\hbar^2 = e^{\varphi_t}$),
    \begin{align}
        \braket{m;t}{V|m;t} \approx \braket{e_m;t}{V|e_m;t} = \frac{1}{2} \sqrt{e^{\varphi_t+\chi_t}f''(x^*)}\left(m + \frac{1}{2}\right).
    \end{align}
    Note that $f = e^{-\chi_t}V$ in the semi-classical limit, then it follows that
    \begin{align}\label{eqn:main_inner}
        \braket{m;t}{f|m;t} = \frac{1}{2} \sqrt{e^{\varphi_t-\chi_t}f''(x^*)}\left(m + \frac{1}{2}\right).
    \end{align} 
    
    Next, we use the Cauchy-Schwarz inequality to bound the cross terms in \eqn{expectation_f}:
    \begin{align}
        |\braket{m;t}{V|n;t}|  &= \left|\int_{\R} \overline{e_m(x)} \left(\frac{e^{\chi_t}}{2} \omega^2 x^2\right) e_n(x)~\d x \right| \\
        & \le \left[\left(\int_{\R} V |e_m|^2~\d x\right)\left(\int_{\R} V |e_n|^2~\d x\right)\right]^{1/2} \\
        & = \frac{1}{2} \sqrt{e^{\varphi_t+\chi_t}f''(x^*)}\sqrt{(m+1/2)(n+1/2)},
    \end{align}
    and similarly, we insert $f = e^{-\chi_t}V$ and end up with
    \begin{align}\label{eqn:cross_inner}
        |\braket{m;t}{f|n;t}|\le \frac{1}{2} \sqrt{e^{\varphi_t-\chi_t}f''(x^*)}\sqrt{(m+1/2)(n+1/2)}.
    \end{align}
    
    Substituting both \eqn{main_inner} and \eqn{cross_inner} to \eqn{expectation_f}, we get an upper bound of $\mathbb{E}[f]_{\sim\Psi_t}$:
    \begin{align}
        \Big|\mathbb{E}[f]_{\sim\Psi_t}\Big| &\le \frac{1}{2} \sqrt{e^{\varphi_t-\chi_t}f''(x^*)}\left(\sum_{m <N} \left(m+\frac{1}{2}\right) + \sum_{\substack{m\neq n\\m,n < N}}\sqrt{(m+1/2)(n+1/2)}\right)\\
        & \le \frac{1}{2} \sqrt{e^{\varphi_t-\chi_t}f''(x^*)} \left(\frac{1}{2}N^2 + \frac{1}{2}N^3 - \frac{1}{2}N^2\right)= \frac{N^3}{4}\sqrt{e^{\varphi_t-\chi_t}f''(x^*)},
    \end{align}
    which will converge to the global minimum $f(x^*)=0$ because $N$ and $f''(x^*)$ are fixed positive numbers and $\lim_{t \to \infty}e^{\varphi_t - \chi_t} = 0$.
\end{proof}

\begin{lemma}\label{lem:expect_potential}
    Consider the quantum harmonic oscillator 
    \begin{align}
        \hat{H} = \hat{K} + \hat{V},
    \end{align}
    where $\hat{K} = -\frac{\hbar^2}{2} \frac{\partial^2}{\partial x^2}$ is the kinetic operator and $\hat{V} = \frac{1}{2}\omega^2 x^2$ is the quadratic potential field. We denote the eigenstates of $\hat{H}$ as $\{\ket{e_n}\}^\infty_{n=0}$. Then, we have
    \begin{align}
        \braket{e_n}{\hat{K}\Big|e_n} = \braket{e_n}{\hat{V}\Big|e_n} = \frac{1}{2}\hbar \omega \left(n+\frac{1}{2}\right),
    \end{align}
    \begin{align}
        \braket{e_n}{\hat{K}\Big|e_{n\pm 1}} =  \braket{e_n}{\hat{V}\Big|e_{n\pm 1}} = 0.
    \end{align}
\end{lemma}

\begin{proof}
    We introduce the ladder operators:
    \begin{align}\label{eqn:ladder}
        \hat{a}_{\pm} \coloneqq \frac{1}{\sqrt{2\hbar \omega}} \left(\mp i\hat{p}+ \omega x\right),
    \end{align}
    where $\hat{p} = -i \frac{\partial}{\partial x}$ is the momemtum operator. The $\hat{a}_{+}$ (or $\hat{a}_{-}$) is known as the \textit{raising} (or \textit{lowering}) operator because
    \begin{align*}
        \hat{a}_{+}\ket{e_n} \propto \ket{e_{n+1}},~\hat{a}_{-}\ket{e_n} \propto \ket{e_{n-1}}.
    \end{align*}
    And in particular, we have (c.f.~\cite[Section 2.3]{griffiths2018introduction})
    \begin{align}
        \hat{a}_{+}\hat{a}_{-}\ket{e_n} = n \ket{e_n},~ \hat{a}_{-}\hat{a}_{+}\ket{e_n} = (n+1) \ket{e_n}.
    \end{align}

    We use the definition of ladder operators \eqn{ladder} to express $\hat{p}$ and $x$,
    \begin{align}
        \hat{p} = i\sqrt{\frac{\hbar \omega}{2}} \left(\hat{a}_{+} - \hat{a}_{-}\right),~x = \sqrt{\frac{\hbar}{2\omega}} \left(\hat{a}_{+} + \hat{a}_{-}\right).
    \end{align}

    We are interested in $\hat{K}$ and $\hat{V}$:
    \begin{align}
        \hat{K} &= \frac{1}{2} \hat{p}^2 = - \frac{\hbar \omega}{4}\left[(\hat{a}_{+})^2 - (\hat{a}_{+}\hat{a}_{-}) - (\hat{a}_{-}\hat{a}_{+}) + (\hat{a}_{-})^2\right],\\
        \hat{V} &= \frac{\omega^2}{2} x^2 = \frac{\hbar \omega}{4}\left[(\hat{a}_{+})^2 + (\hat{a}_{+}\hat{a}_{-}) + (\hat{a}_{-}\hat{a}_{+}) + (\hat{a}_{-})^2\right].
    \end{align}
  
    Note that $(\hat{a}_{+})^2$ in $\hat{K}$ or $\hat{V}$ will raise $\ket{e_n}$ to $\ket{e_{n+2}}$ and thus has no overlap with $\ket{e_n}$. Similarly, $\bra{e_n}(\hat{a}_{-})^2\ket{e_n}=0$. It turns out that,
    \begin{align}
        \braket{e_n}{\hat{K}\Big|e_n} = \frac{\hbar \omega}{4} \bra{e_n}\left[(\hat{a}_{+}\hat{a}_{-}) + (\hat{a}_{-}\hat{a}_{+})\right]\ket{e_n} = \frac{\hbar \omega }{2}\left(n+\frac{1}{2}\right).
    \end{align}
    Similarly, $\braket{e_n}{\hat{V}\Big|e_n} = \frac{\hbar \omega }{2}\left(n+\frac{1}{2}\right)$.

    As for $\braket{e_n}{\hat{K}\Big|e_{n+1}}$, we note that $\ket{K}$ maps the state $\ket{e_{n+1}}$ to a linear combination of $\ket{e_{n-1}}$, $\ket{e_{n +1}}$, and  $\ket{e_{n+3}}$. It has no overlap with the state $\ket{e_n}$, thus $\braket{e_n}{\hat{K}\Big|e_{n+1}}=0$. The same argument applies to show that $\braket{e_n}{\hat{K}\Big|e_{n-1}} = 0$ and $\braket{e_n}{\hat{V}\Big|e_{n\pm 1}} = 0$.
\end{proof}

\subsection{QHD for quadratic model functions}
The Schr\"odinger equation describing QHD is often too complicated to solve analytically. Fortunately, we can find a closed-form solution of QHD when $f$ is a quadratic function. In the following calculation, we consider the one-dimensional case $f(x) = \frac{1}{2}x^2$ for simplicity. It is worth noting that the same method also applies to general finite-dimensional quadratic forms $f(\vect{x}) = \frac{1}{2}\vect{x}^T A \vect{x}$ with a positive semidefinite matrix $A$\footnote{When $A$ is not positive semidefinite, the quadratic form $f(\vect{x})$ is not bounded from below. We think this is not an interesting minimization problem to study.}. It turns out that, if we choose the same time-dependent parameters as in Nesterov's accelerated gradient descent, the convergence rate is $\mathbb{E}[f]_{\sim \Psi_t} = O(t^{-3})$. This rate is at par with the convergence rate of the continuous-time model of Nesterov's method~\cite{su2014differential}.

\begin{remark}
    Our result does not mean QHD can not achieve faster convergence for quadratic objective functions. In fact, one can use a linear function for $\beta_t$, and the convergence rate can be exponentially fast. Here, our goal is to compare QHD to the classical ODE model of Nesterov's method.
\end{remark}

We choose the following time-dependent parameters in the QHD Hamiltonian \eqn{qhd_operator}:
\begin{align}\label{eqn:nesterov_params}
    \alpha_t = \log (2/t), \beta_t = 2\log(t), \gamma_t = 2\log(t),
\end{align}
and the QHD Hamiltonian is 
\begin{align}\label{eqn:quadratic_model_qhd}
    \hat{H}(t) = -\frac{1}{t^3}\frac{\partial^2}{\partial x^2} + t^3 x^2.
\end{align}

Note that our choice of the time-dependent parameters \eqn{nesterov_params} is the same as the ones in the ODE model of the Nesterov's accelerated gradient descent algorithm \cite{wibisono2016variational}. It has been shown that this ODE model exhibits $O(t^{-3})$ convergence rate for strongly convex functions such as $\frac{1}{2}x^2$~\cite{su2014differential}.

\begin{proposition}\label{prop:quadratic_rate}
    Let $f$, $\hat{H}(t)$ be the same as above, and we consider the initial wave function,
    \begin{align}\label{eqn:init_gaussian}
        \Psi_0(x) = \left(\frac{1}{2\pi}\right)^{1/4} \exp\left(\frac{-x^2}{4}\right).
    \end{align}
    Let $\ket{\Psi_t}$ be the solution to the Schr\"odinger equation $i\partial_t\ket{\Psi_t}= \hat{H}(t)\ket{\Psi_t}$, then we have
    \begin{align}
        \mathbb{E}[f]_{\sim \Psi_t} = \Theta(t^{-3}).
    \end{align}
\end{proposition}
\begin{proof}
    Since $\hat{H}(t)$ has a quadratic potential field, the quantum wave packet will remain Gaussian in the time evolution. We now introduce the following solution ansatz
    \begin{equation} \label{eqn:ansatz_gaussian}
        \psi(t,x) = \left(\frac{1}{\pi}\right)^{1/4} \frac{1}{\sqrt{\delta(t)}}\exp(-i\theta(t)) \exp\left(\frac{-x^2}{2\delta(t)^2}\right),
    \end{equation}
    with $\theta(0) = 0$, $\delta(0) = \sqrt{2}$, and the probability amplitude $p_\lambda(t,x)$ is given by
    \begin{equation} \label{eqn:density}
        p(t,x):= |\Phi(t,x)|^2 = \frac{1}{\sqrt{\pi}} \frac{1}{|\delta(t)|} \exp\Big(2 \operatorname{Im}(\theta(t))\Big) \exp\Big(-x^2 \operatorname{Re}(1/ y(t))\Big),
    \end{equation}
    where $y(t) = \delta^2(t)$. Since the wave function \eqn{ansatz_gaussian} solves the Schr\"odinger equation, we have $\|\Phi(t,x)\|^2 = 1$ for all $t \geq 0$. It turns out that \eqref{eqn:density} is the density function of a Gaussian random variable with zero mean and variance
    \begin{equation} \label{eqn:variance0}
        \sigma^2(t) = \frac{1}{2\operatorname{Re}(1/ y(t))}.
    \end{equation}
    
    Substituting the ansatz \eqn{ansatz_gaussian} to the Schr\"odinger equation $i\partial_t\ket{\Psi_t}= \hat{H}(t)\ket{\Psi_t}$, and we take change of variables $y(t) = \delta^2(t)$. The equation ends up being:
    \begin{subequations}\label{eqn:ode}
        \begin{equation}
        \label{eqn:ode-a}
         y' +i\left(t^3 y^2 - t^{-3}\right) = 0
    \end{equation}
    \begin{equation}
        \label{eqn:ode-b}
        \theta' = \frac{i}{4}\frac{y'}{y} + \frac{1}{2t^3}\frac{1}{y}
    \end{equation}
    \begin{equation}
        \label{eqn:ode-c}
        \theta(0) = 0, ~y(0) = 2.
      \end{equation}
    \end{subequations}
    
    Since the phase $\theta(t)$ merely contributes to a normalization factor in the probability density, we can ignore \eqn{ode-b} and focus on the first equation \eqn{ode-a}, which simplifies to the following Riccati equation:
    \begin{align}\label{eqn:riccati}
        \dot{y} = (-it^3)y^2 + it^{-3}.
    \end{align}
    Introducing the change of variable $y = -i\dot{u}/t^3 u$, \eqn{riccati} reduces to a linear second-order ODE:
    \begin{align}\label{eqn:intermediate}
        \ddot{u} - \frac{3}{t}\dot{u} + u = 0.
    \end{align}
    Finally, we introduce a time scaling $u = t^2v$ so \eqn{intermediate} is transformed to the standard Bessel equation:
    \begin{align}\label{eqn:bessel}
        t^2 \ddot{v} + t\dot{v} + (t^2-4)v = 0.
    \end{align}
    A solution of the Bessel equation \eqn{bessel} is a linear combination of Bessel functions $J_2(t)$ and $Y_2(t)$, i.e., 
    \begin{align}
        v(t) = c_1 J_2(t) + c_2Y_2(t).
    \end{align}
    
    It is tricky (but still possible) if we want to determine the exact coefficients $c_1$ and $c_2$ since $Y(t)$ has a singularity at $t = 0$. Fortunately, since we are more interested in the asymptotic behavior of $y(t)$, we do not have to compute the coefficients $c_1$ and $c_2$. Instead, we consider the asymptotic form of Bessel functions:
    \begin{align}
        J_k (t) = \sqrt{\frac{2}{\pi t}} \left(\cos(t- \frac{k\pi }{2} - \frac{\pi}{4}) + O(1/t)\right), \\
        Y_k (t) = \sqrt{\frac{2}{\pi t}} \left(\sin(t- \frac{k\pi }{2} - \frac{\pi}{4}) + O(1/t)\right), 
    \end{align}
    and it follows that $v(t) \propto t^{-1/2}$ when $t \to \infty$. 

    To compute the derivative of $v$, we make use of the following recursive relation for Bessel functions $Z_k$ ($Z = J$ or $Z = Y$):
    \begin{align}
        \dot{Z}_k(t) = \frac{1}{2}\left(Z_{k-1}(t) - Z_{k+1}(t)\right).
    \end{align}
    Therefore, we can expand the derivatives of Bessel functions as
    \begin{align}
        \dot{v} = c_1 \dot{J}_2 + c_2 \dot{Y}_2 = \frac{c_1}{2}(J_1 - J_3) + \frac{c_1}{2}(J_1 - J_3),
    \end{align}
    and it turns out that $\dot{v}(t) \propto t^{-1/2}$ as well. In terms of $y(t)$, we have
    \begin{align}
        y(t) = \frac{-i}{t^3}\left[\frac{2}{t} + \frac{\dot{v}}{v}\right] = \Theta(t^{-3}).
    \end{align}
    Note that $\dot{v}$ and $v$ decays at the same order so we expect $\dot{v}/v$ converges to a constant as $t \to \infty$. It is reasonable to assume that the real and imaginary part of $y(t)$ are at the same magnitude as $t$ becomes large, and it follows that $\sigma^2(t) = \Theta(t^{-3})$. Recall that the objective/potential function is $f(x) = \frac{1}{2}x^2$, the expectation value of $f$ is
    \begin{equation}\label{eqn:quadratic_convergence}
        \mathbb{E}[f]_{\sim \Psi_t} = \int_{\R}\frac{1}{2}x^2 |\Psi(x,t)|^2~\d x = \frac{1}{2} \sigma^2(t) = \Theta(t^{-3}).
    \end{equation}
\end{proof}

\prop{quadratic_rate} indicates that QHD has the same asymptotic convergence rate as its ODE counterpart in \cite[Section 3.2]{su2014differential} when applied on the quadratic model function $f$. Note that we set $\beta_t = \log(t^2)$. \thm{convex_convergence} implies that $\mathbb{E}[f]_{\sim \Psi_t} \le O(e^{-\beta_t}) = O(t^{-2})$. So QHD actually does better on the quadratic problem $f$ than the worst-case theoretical guarantee. However, \prop{quadratic_rate} also rules out faster convergence.

\section{Quantum and Classical Algorithms for Nonconvex Problems}\label{sec:nonconvex}
\subsection{Test problems}\label{sec:2d-test-problems}
We select 22 optimization instances in the literature \cite{jamil2013literature, simulationlib, Al-Roomi2015}. The optimization problems are split into five categories by the landscape features, see \tab{BenchmarkSet}.

\begin{table}[!ht]
    \centering
    \begin{tabular}{|p{0.2\textwidth}|c|p{0.6\textwidth}|}
        \hline
            \multicolumn{3}{|c|}{Custom Classification (no repetition)} \\
        \hline
            \textbf{Features} & \textbf{Count} & \textbf{Function names}\\
        \hline
            Ridges or Valleys & 5 & Ackley 2, Dropwave, Holder Table, Levy, Levy 13\\
        \hline
            Basin & 5 & Bohachevsky 2, Camel 3, Csendes, Deflected Corrugated Spring, Rosenbrock\\
        \hline
            Flat & 3 & Bird, Easom, Michalewicz\\
        \hline
            Studded & 4 & Ackley, Alpine 1, Griewank, Rastrigin\\
        \hline
            Simple & 5 & Alpine 2, Hosaki, Shubert, Styblinski-Tang, Sum of Squares\\
        \hline
    \end{tabular}
    \caption[2D test functions]{A classification of the functions, separated by features we believe to be interesting to investigate in the comparison between QHD and alternate optimization algorithms.}
    \label{tab:BenchmarkSet}
\end{table}

The category ``Ridges or Valleys'' is characterized by small deep sections and/or steep walls that divide the domain. ``Basin'' is characterized by flatness at function values near the global optimum. ``Flat'' is characterized by flatness at high function values. ``Studded'' is characterized by a base shape with higher frequency perturbation. Local search algorithms may be confused by many local minima and high gradient magnitudes. ``Simple'' functions are those on which standard gradient descent algorithms work efficiently.\footnote{These functions serve as a `sanity check' on the behavior of the gradient methods and QHD} We also plot two representative instance from each of the five categories, as shown in \fig{representative_table}.

\begin{figure}[ht!]
    \centering
     \includegraphics{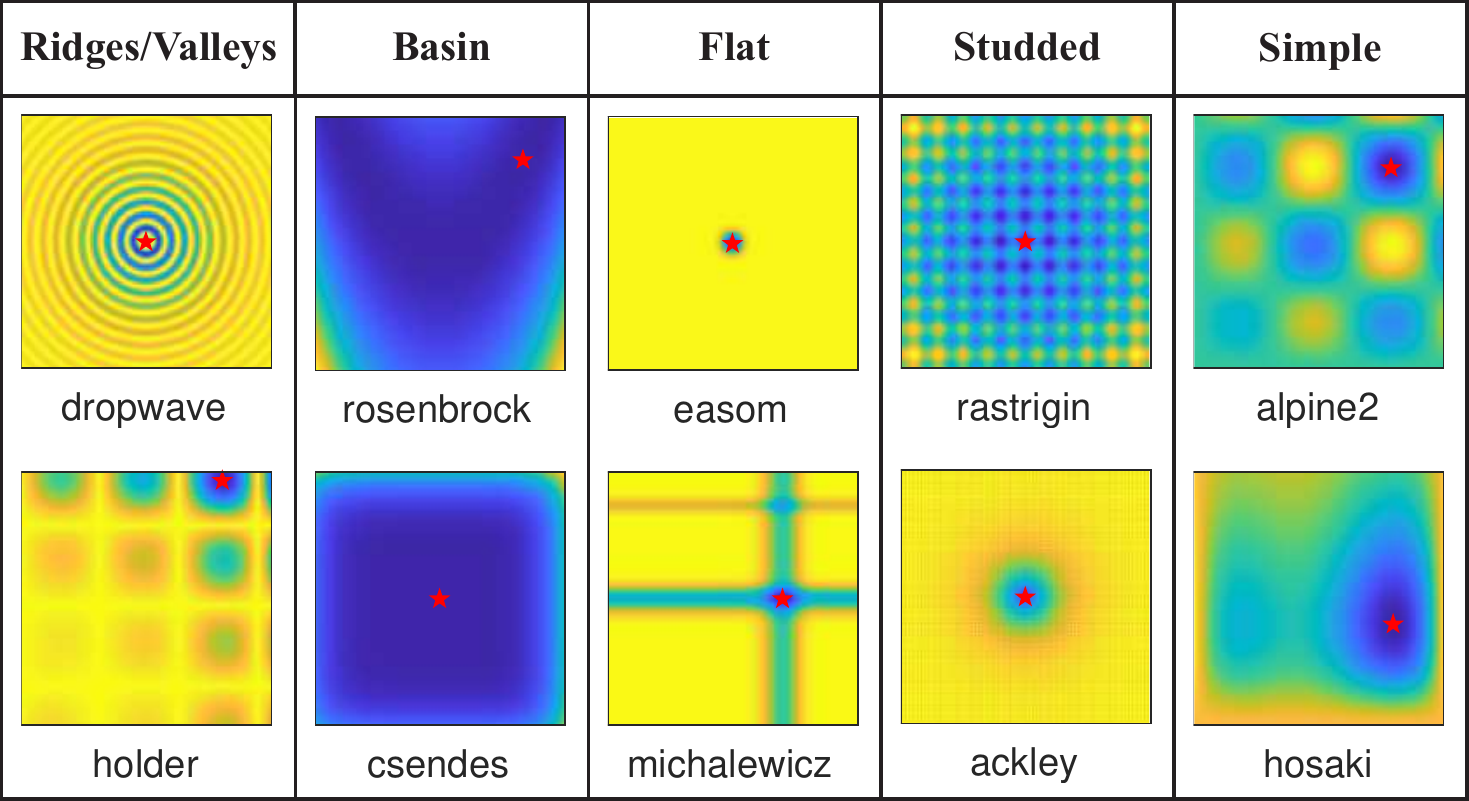}
     \caption[Representatives of 2D test functions]{Representatives from the benchmark set of test functions. High function values shown in yellow, down to low values in blue. The unique global minimum is marked with a red star.}
     \label{fig:representative_table}
\end{figure}

To unify the simulation accuracy and comparability of the functions, we truncate the objective function $f(x,y)$ over a squared domain $[a,b]^2$, and then rescale the function so it is supported on the unit suqare $[0,1]^2$. The rescaled objective function is given by:
\begin{align}
    \tilde{f}(x,y) = \frac{1}{L}f(a+Lx, b+Ly) - f_{\min},
\end{align}
where $L = b-a$ and $f_{\min}$ is the global minimum value of $f$. The global minimum of the rescaled function $ \tilde{f}$ is 0. The normalization factor $1/L$ is introduced to preserve the size of gradient and Hessian of $f$. We do this to avoid affecting the performance of the (classical) gradient-based algorithms, preserving fair comparison.

\subsection{Quantum algorithms}

We test two quantum algorithms, Quantum Hamiltonian Descent (QHD) and Quantum Adiabatic Algorithms (QAA), for the 22 two-dimensional non-convex optimization instances. Both quantum algorithms are simulated numerically on classical computers so we can visualize the evolution of the probability distribution.

\subsubsection{Quantum Hamiltonian Descent (QHD)}\label{sec:simulating-qhd}

As a reminder, the QHD Hamiltonian \eqn{qhd_operator} is 
\begin{equation*}
    H(t) = e^{\varphi_t} \left(-\frac{1}{2} \nabla^2\right) + e^{\chi_t}f(x).
\end{equation*}

The QHD evolution from $t=0$ to $t=T$ is simulated by iteratively applying the product formula, 
\begin{align}\label{eqn:prod_qhd}
    \ket{\Psi_{j+1}} = \Big[ \exp(-i s a_j (-\nabla^2/2)) \exp(-is b_j f)\Big]\ket{\Psi_j},
\end{align}
where $j = 0,...,N-1$, $s=T/N$ is the stepsize, and $a_j=e^{\varphi_{js}}$, $b_j=e^{\chi_{js}}$. We choose the initial state $\ket{\Psi_0}$ to be the the uniform random distribution over the domain $[0,1]^2$. 

We observe that the operator $f$ is diagonal in the position basis, and $\nabla^2$ is diagonalized by the Fourier basis, i.e., $-\frac{1}{2}\nabla^2 = \mathbf{F}_s \mathbf{L}\left(\mathbf{F}_s\right)^{-1}$, where $\mathbf{F}_s$ is the shifted Fourier transform (SFT)\footnote{SFT is a close variant of the Fast Fourier Transform (FFT) and can be readily implemented from FFT.}, and $\mathbf{L}$ is a diagonal matrix with eigenvalues of $(-\nabla^2/2)$. These eigenvalues can be computed from the frequency number of a Fourier basis function \cite[Eq.(26)]{childs2022quantum}. This means that each iteration step can be implemented by 
\begin{align}\label{eqn:prod_qhd_fft}
    \ket{\Psi_{j+1}} = \Big[ \mathbf{F}_s \exp(-i s a_j \mathbf{L})\left(\mathbf{F}_s\right)^{-1} \exp(-is b_j f)\Big]\ket{\Psi_j}.
\end{align}
In the literature, this method is often known as the pseudo-spectral method~\cite{boyd2001chebyshev} and it is a standard numerical method for Schr\"odinger equations. 

\paragraph{Experiment parameters.}
We choose the time-dependent functions $\varphi_t$ and $\chi_t$ so that they are the same as in the ODE model of the Nesterov's accelerated gradient descent algorithm~\cite{su2014differential}. To avoid the singularity at $t=0$ when performing numerical simulation, we slightly modify the first time-dependent function and the QHD Hamiltonian becomes
\begin{equation}\label{eqn:nonconvex_QHD_h}
    H(t) = \left(\frac{2}{s + t^3}\right) \left(-\frac{1}{2} \nabla^2\right) + 2t^3f(x),
\end{equation}
where $s$ is the stepsize. We choose $s=0.001$ for all test instances, and the total evolution time is $T =10$, hence the maximal iteration number $N = T/s = 10^4$. We try two resolutions, 128 and 256, in the spatial discretization of QHD, i.e., we simulate QHD on a $128\times128$ and a $256\times256$ mesh grid over the unit square $[0,1]^2$. The initial state is the equal uniform superposition of all points.

\subsubsection{Quantum Adiabatic Algorithm (QAA)}\label{sec:qaa}

Quantum Adiabatic Algorithm (QAA)~\cite{farhi2001quantum} is a well-known quantum algorithm for general optimization problems. QAA is formulated as a quantum adiabatic evolution described by the Schr\"odinger equation ($0 \le t \le T$):
\begin{align}
    i \frac{\d}{\d t} \ket{\psi(t)} = \Big[\big(1 - g(t)\big)H_0 + g(t)H_1\Big]\ket{\psi(t)},
\end{align}
where $g(t)$ is the annealing schedule such that $g(0) = 0$ and $g(T) = 1$, $H_0$ is the initial Hamiltonian, and $H_1$ is the problem Hamiltonian that encodes the problem of interest. If we choose the initial state $\ket{\psi(0)}$ as the ground state of $H_0$ and let the evolution time $T$ be large enough, the quantum state $\ket{\psi(t)}$ remains the ground state of the system in the whole evolution, and the final state $\ket{\psi(T)}$ corresponds to a solution to the problem of interest.

Conventionally, QAA is defined over discrete domains and thus applicable to discrete optimization problems.\footnote{Note that there is no direct or natural counterpart of gradient descent/QHD in discrete domains.} The performance of QAA heavily depends on the choice of the initial Hamiltonian $H_0$. Prior to QHD, the standard approach for solving continuous problems using QAA is to adopt Radix-2 representation of real numbers (i.e., binary numeral system) in the continuous domain so that the original problem is converted to a discrete optimization problem defined over $\{0,1\}^N$, e.g.,~\cite{potok2021adiabatic, cohen2020portfolio}. In this case, the initial Hamiltonian is chosen to be the ``sum of Pauli-X'' operator:
\begin{align}\label{eqn:qaa_hamiltonian}
    H_0 = -\sum^N_{j=1} \sigma_x^{(j)},
\end{align}
whose ground state is the uniform superposition state, and $H_1$ is a diagonal matrix whose diagonal elements are evaluated from the objective functions. We will refer to this traditional approach of solving continuous problems using QAA as the \textbf{baseline QAA}, or simply QAA.

In our experiment, we simulate the baseline QAA using the leapfrog scheme for time integration. Leapfrog integrators are time-reversible and symplectic \cite{quillen2018integrators}, so lend themselves well to simulation of unitary dynamics.

The choice of the annealing schedule $g(t)$ has a huge impact on the performance of QAA. In practice, it is common to use the linear interpolation schedule $g(t) = \frac{t}{T}$, while sometimes it fails to achieve optimal results. For example, to achieve the quadratic quantum speedup for the unstructured search problem, one needs to use the local adiabatic schedule (see \fig{qaa_schedules}A) instead of the plain linear interpolation~\cite{roland2002quantum}. For our 2D test problems, however, we find that the linear interpolation schedule works better than the local adiabatic schedule, potentially due to the limited evolution time and the continuous nature of the original problem. In \fig{qaa_schedules}B, we show the success probability of the quantum adiabatic evolution (applied to the Levy function). It is clear that the local adiabatic schedule from~\cite{roland2002quantum} does not keep the state in the ground-energy subspace and it is outperformed by the standard linear interpolation schedule at $T = 10$. Actually, we also tried the prolonged evolution time $T = 100$ for the local adiabatic schedule but it is still worse than the linear interpolation one. For other test instances, we also observe similar behavior of the local adiabatic schedule. Therefore, we choose the linear interpolation schedule in our numerical experiment.

\begin{figure}[ht!]
    \centering
    \includegraphics[width=16cm]{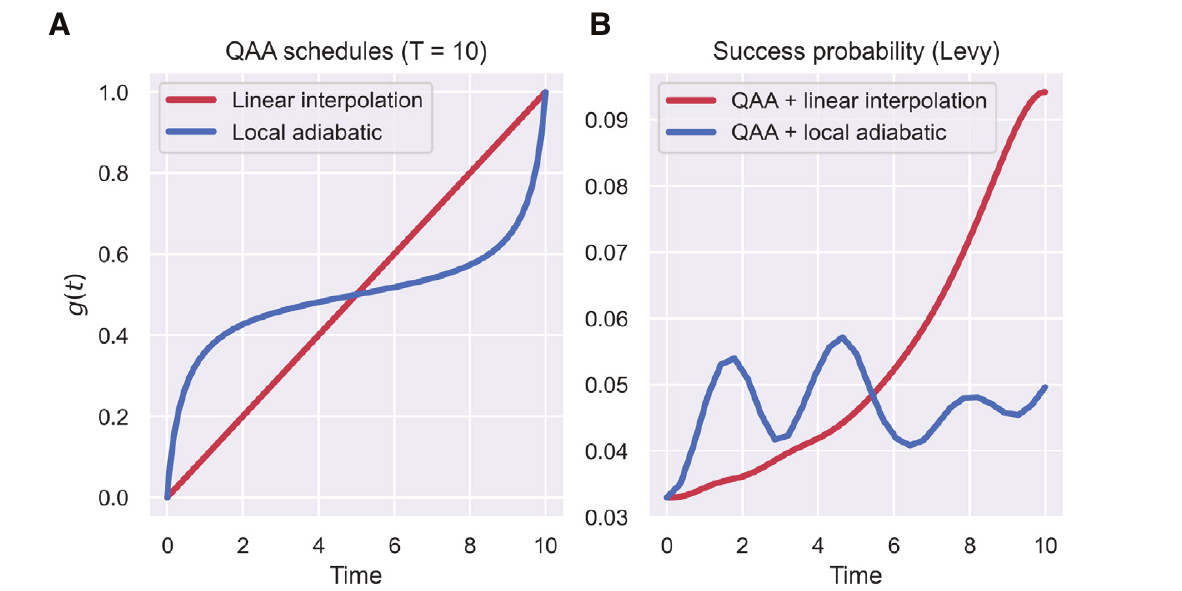}
    \caption[Annealing schedules for QAA]{\textbf{A:} the linear interpolation schedule and the local adiabatic schedule for $T = 10$.
    \textbf{B:} the success probability in the quantum adiabatic evolution. We choose spatial resolution $64$ and simulate QAA by leapfrog integrator.}
    \label{fig:qaa_schedules}
\end{figure}

\begin{remark}
Although QAA is usually applied to discrete problems, one can lift QAA to continuous domains by choosing $H_0$, $H_1$ over a continuous space. From this perspective, QHD can be regarded as a special version of this general definition of QAA (by choosing $H_0$ as the kinetic operator and $H_1$ as the potential operator). However, unlike the continuous-version (baseline) QAA, we have a more refined theoretical analysis of QHD because of the rich structures in the QHD Hamiltonian. Not as many theoretical results are known for QAA due to the weak structure of general discrete optimization problems. Meanwhile, we also observe that QHD converges much faster than the baseline QAA for continuous problems in our empirical study.
\end{remark}

\paragraph{Experiment parameters.} We try two resolutions, 64 and 128, in the simulation of QAA. A higher resolution (e.g., 256) will make the problem too large for classical computers to solve in reasonable time.\footnote{Doubling the resolution in 2D therefore requires four times the storage (so four times the work per step) and eight times the number of steps.} Also, the discrete problem becomes significantly harder with more grid points and QAA needs a super long time to return meaningful solution. To compare with the performance of QHD, we choose the total evolution time $T=10$ for all instances. The initial state is the equal superposition of all grid points, and we choose the linear interpolation schedule $g(t) = t/T$.

\subsection{Classical algorithms}
\subsubsection{Nesterov's accelerated gradient descent algrothms (NAGD)}

In 1983, Nesterov proposed an accelerated gradient method \cite{nesterov1983method}. It begins with initial guesses $x_0 = y_0$ and inductively computes the update sequence,
\begin{subequations}\label{eqn:nesterov}
    \begin{align}
        x_k &= y_{k-1} - s \nabla f(y_{k-1}),\label{eqn:nesterov-a}\\
        y_k &= x_k + \frac{k-1}{k+2}(x_k - x_{k-1}).\label{eqn:nesterov-b}
    \end{align}
\end{subequations}
We will refer to this method as Nesterov's accelerated gradient descent algrothms, or NAGD, in what follows. It has been shown that, for any step size $s \le 1/L$, where $L$ is the Lipschitz constant of $\nabla f$, this method achieves the convergence rate $O(k^{-2})$ for continuously differentiable and convex $f$. This rate is faster than the standard gradient descent. 

In the continuous-time limit (i.e., $s\to 0$), NAGD is modelled as a second-order differential equation, $ \ddot{X} + \frac{3}{t}\dot{X} + \nabla f(X) = 0$; see~\cite{su2014differential}. In the continuous-time model, the effective evolution time at the $k$-th step is $t_k = ks$. Equivalently, this system is described by the following Hamiltonian-mechanical system:
\begin{align}
    H(X,P,t) = \frac{2}{t^3}\left(\frac{1}{2}|P|^2\right) + 2t^3 f(X).
\end{align}

\paragraph{Experiment parameters.} We choose the stepsize $s = 0.001$ for all instances. We also fix the total effective evolution time $T=10$, so the maximal number of iterations is $N=T/s=10^4$. We draw 1000 initial points $x_0=y_0$ uniformly random from the feasible domain $[0,1]^2$ and compute the optimization updates paths by \eqn{nesterov}, respectively.

\subsubsection{Stochastic Gradient Descent (SGD)}

Stochastic gradient descent (SGD) is widely used in continuous  non-convex optimization, and has been very successful in training large-scale systems such as neural networks. SGD is a gradient-based method and it updates each step with the rule:
\begin{align}\label{eqn:sgd_update}
    x_{k+1} = x_k - s \tilde{\nabla}f(x_k),
\end{align}
where $s$ is a fixed stepsize (or learning rate), $\tilde{\nabla} f$ is a noisy gradient evaluated from a single data point or a mini-batch. 

SGD is a discrete-time algorithm. The continuous-time limit of SGD turns out to be a first-order stochastic differential equation (SDE), as derived in \cite{shi2020learning},
\begin{align}
    d X_t = - \nabla f(X_t)\d t + \sqrt{s}\d W_t.
\end{align}
With the SDE model, we can compare the evolution of the distribution in SGD to the evolution in QHD. In particular, we compute the effective evolution time in SGD at the $k$-th step by $t_k = sk$. 

\paragraph{Experiment parameters.} We model the noisy gradient in \eqn{sgd_update} by adding unit Gaussian noise to the exact gradient, i.e., $\tilde{\nabla} f = \nabla f + \mathcal{N}(0, 1)$. We set the stepsize $s=0.001$ for all instances, and we stop the algorithm after $10^4$ iteration steps, i.e., the total effective evolution time for SGD is $T = 10$. We draw 1000 initial points $x_0$ uniformly random from the feasible domain $[0,1]^2$ and compute the optimization updates paths by \eqn{sgd_update}, respectively.

\begin{figure}[ht!]
    \centering
    \includegraphics[width=16cm]{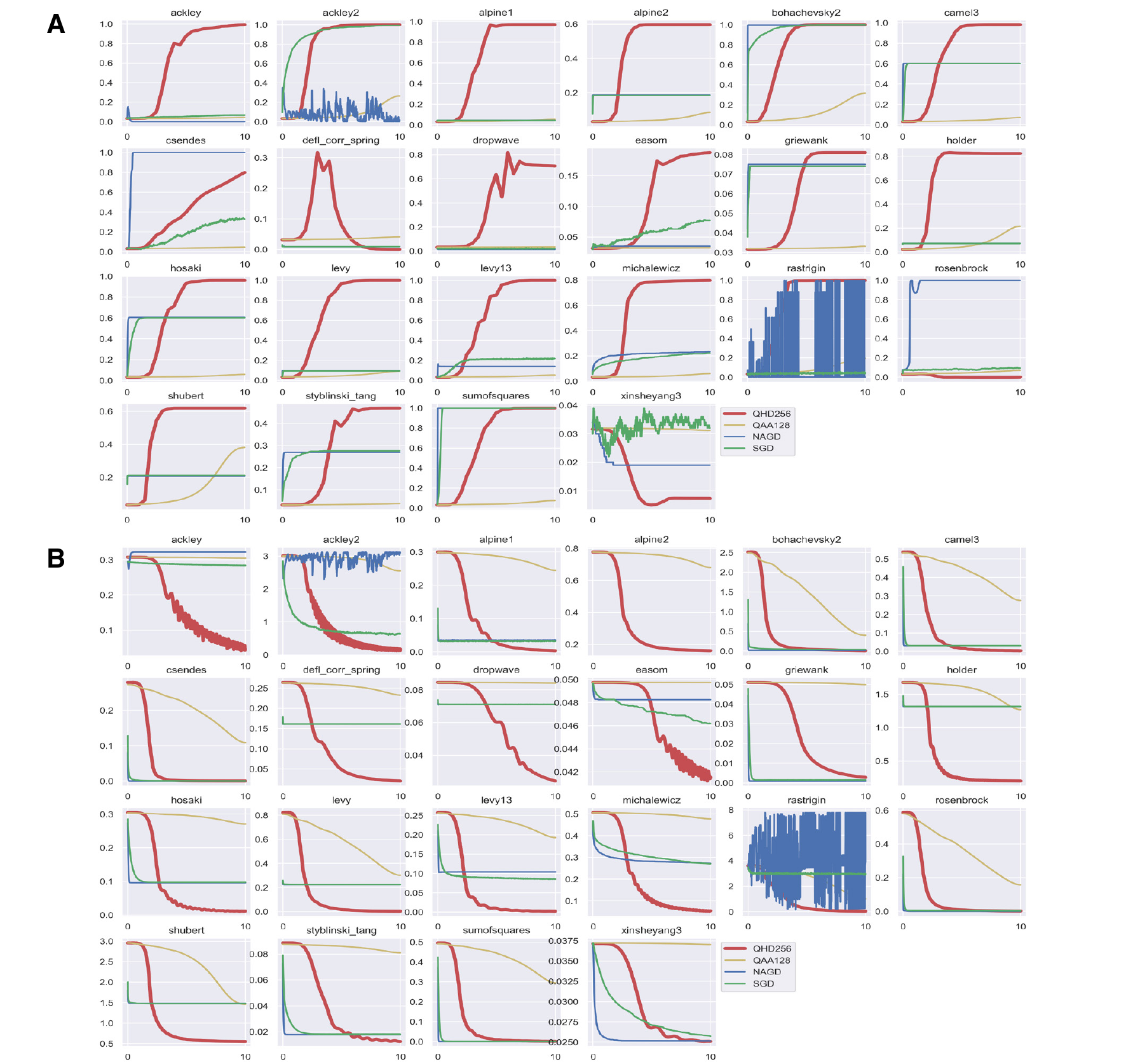}
    \caption[Experiment results for 2D test problems]{The performance of four tested algorithms for all 22 problems. \textbf{A:} success probability, \textbf{B:} loss curve (expected function value). In both panels, the x data in each subplot is the evolution time $t \in [0,10]$.}
    \label{fig:2d_experiment_result}
\end{figure}

\subsection{Experiment results}

As described above, we fix the same total evolution time $T=10$ among all tested algorithms. For Levy function, we make the scatter plot showing the distributions of solutions in panel A of \fig{fig2}. The final success probability (measured at $t=10$) for all 22 instances is shown in panel B of \fig{fig2}. We observe QHD has a higher success probability in most optimization instances. The full success probability data of all 22 instances is shown in \fig{2d_experiment_result}A.

We also find that QHD is more stable than Nesterov and SGD. We use the same stepsize $s = 0.001$ for QHD and the two classical iterative methods. However, we observe that NAGD fails to converge for Ackley2 and Rastrigin functions, while QHD manages to converge for both instances. This observation suggests that our quantization approach not only improves the solution quality, but also the stability of the resultant quantum dynamics in time discretization.

\begin{figure}[ht!]
    \centering
    \includegraphics[width=12cm]{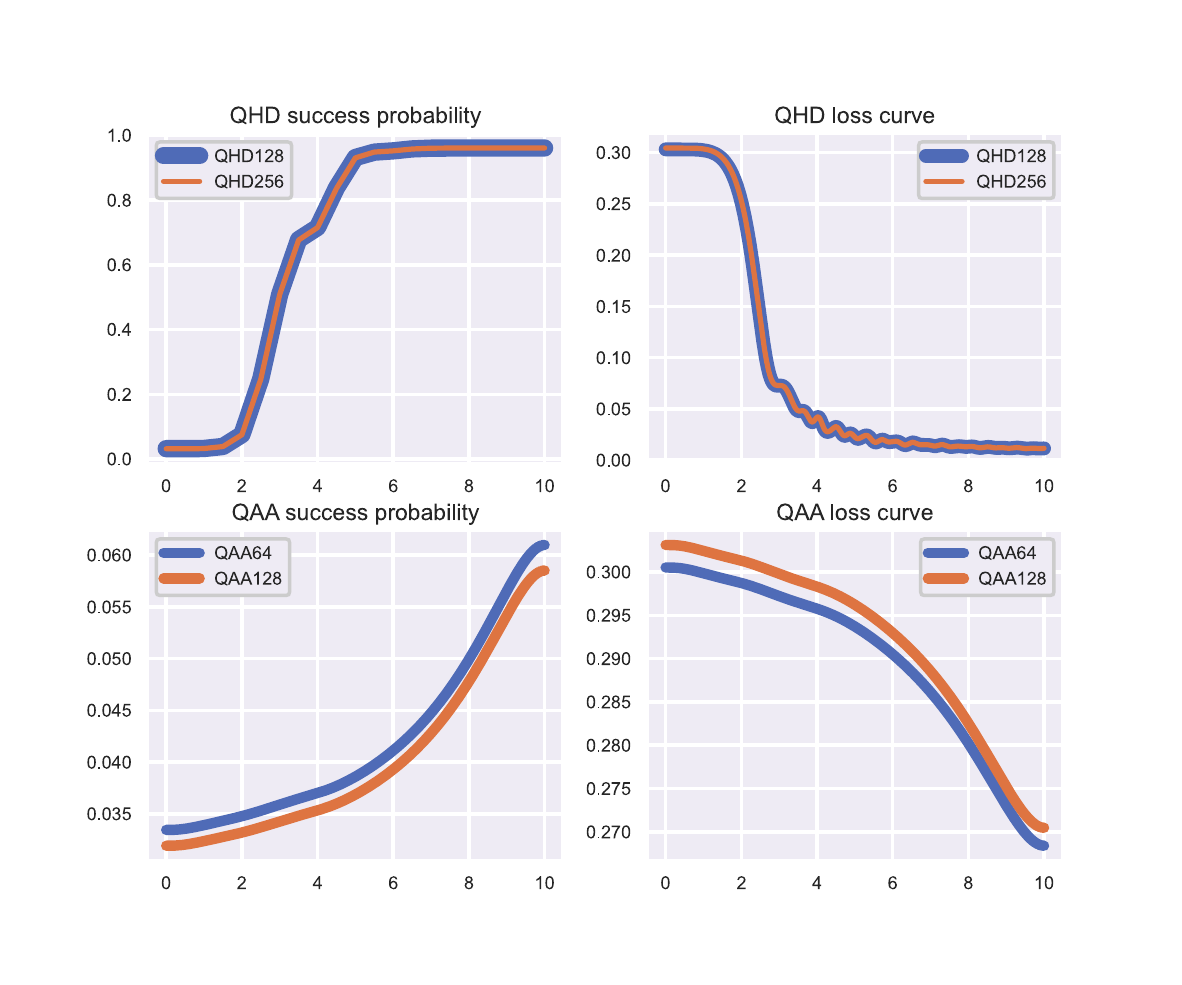}
    \caption[Performance of QHD and QAA in different resolutions]{Performance of QHD and QAA in different rsuesolution (QHD: 128, 256; QAA: 64, 128), illustrated with Hosaki function. The performance of QHD is not affected by the resolution, while QAA does worse when the resolution is higher.}
    \label{fig:compare_qaa_qhd}
\end{figure}

In \fig{compare_qaa_qhd}, we show that the performance of QAA highly depends on the spatial resolution. For Levy function, we see the high resolution (128) in QAA causes worse performance. This observation also holds for several other instances. This is because the Radix-2 representation used in QAA does not preserve the Euclidean topology in the original problem, therefore the discrete problem can become much harder for higher resolution. On the contrary, QHD does not suffer this issue -- the performance of QHD is consistent in both low and high resolution cases.

\section{Three Phases in QHD} \label{sec:three-phase}
In \thm{adiabatic_limit_qhd}, we have shown that there exists time-dependent parameters in QHD such that the convergence to global minimum is guaranteed, regardless of the shape of $f$. This is because the quantum state in QHD stays in the low-energy subspace in the evolution, and the low-energy subspace will eventually settle at the global minimizer of $f$. In this section, we introduce a more refined description of the global convergence in QHD. In particular, we look into the energy exchange between energy levels in the quantum evolution. We observe that there are \textbf{three phases} in QHD evolutions: namely, the \textbf{kinetic} phase, \textbf{global search} phase, and \textbf{descent} phase. We also develop a quantitative analysis for the three-phase picture of QHD. Serving as an explanative framework of QHD, the three-phase picture demystifies why QHD usually does better than QAA in solving non-convex problems.

\begin{figure}[ht!]
    \centering
         \includegraphics[width=1\textwidth]{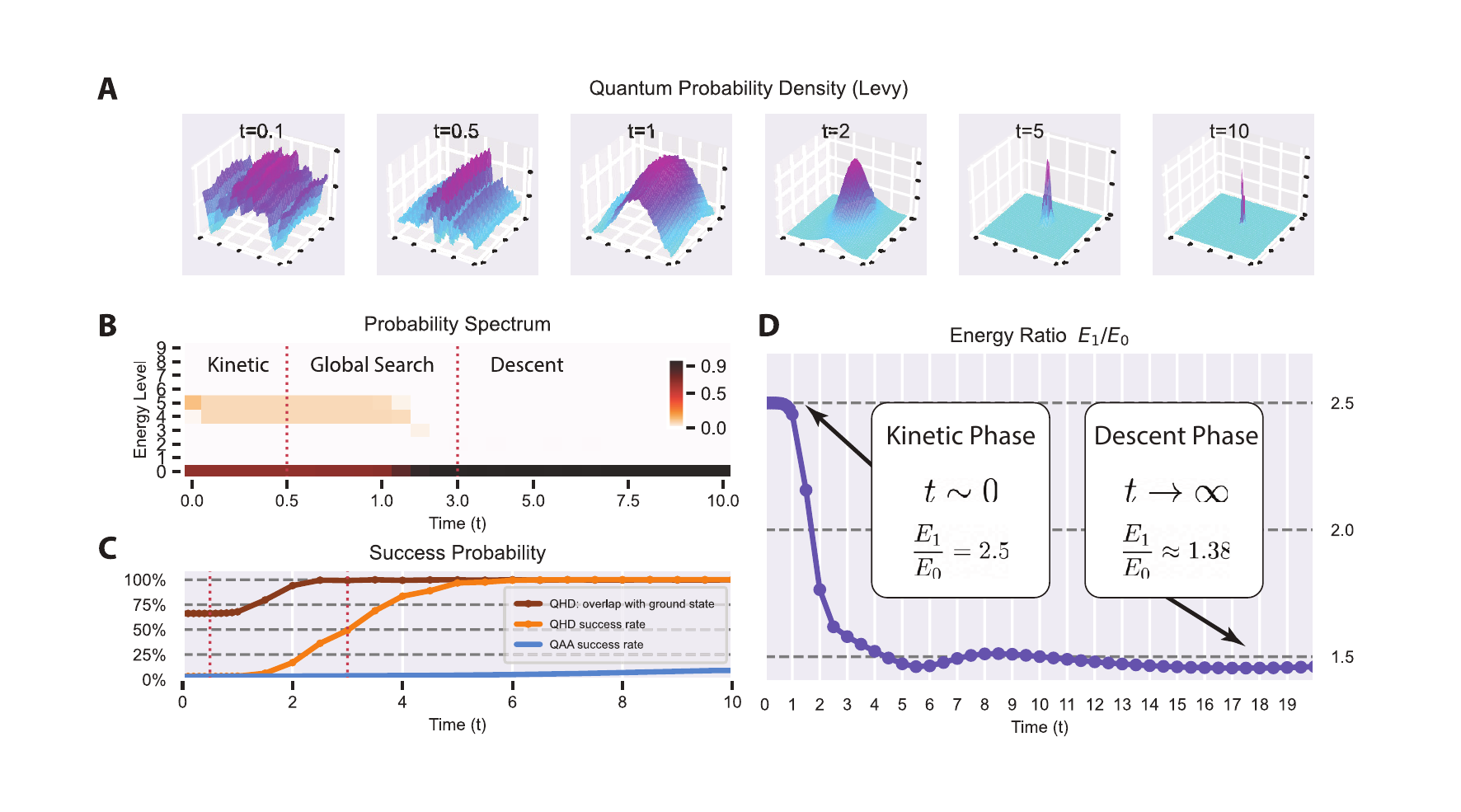}
         \caption[Three-phase picture of QHD]{The three-phase picture of QHD, illustrated with Levy function.}
         \label{fig:three_phase}
    \end{figure}

\subsection{Probability spectrum of QHD}
In \sec{convergence-nonconvex}, we introduce the notion of low- and high-energy subspaces in the QHD evolution. We now generalize this idea and define the \textit{probability spectrum} of QHD.

Let $H(t)$ be the QHD Hamiltonian. Suppose $E_n(t)$ and $\ket{n;t}$ are the $n$-th eigenvalue and eigenstate of $H(t)$, we can express the Hamiltonian as $H(t) = \sum_{n} E_n(t) \ket{n;t}\bra{n;t}$. Here the integer $n$ represents the energy levels of the system. Given a wave function $\ket{\Psi_t}$ at time $t>0$, it can be written as a superposition of eigenstates of $H(t)$: $\ket{\Psi_t} = \sum_n c_n(t)\ket{n;t}$, where $c_n(t) = \braket{n;t}{\Psi_t}$ is the probability amplitude.

\begin{definition}[Probability spectrum]
    The modulus squared of the amplitude (i.e., $|c_n(t)|^2$) represents the probability density of the wave function with respect to the energy eigenstates. We call the sequence $\{|c_n(t)|^2 = |\braket{n;t}{\Psi_t}|^2\}_n$ as the \textit{probability spectrum} of $\ket{\Psi_t}$ at time $t \ge 0$. 
\end{definition}

Since the quantum wave function has unit $L^2$ norm, the probability spectrum of wave functions at any $t \ge 0$ sum up to $1$: $\sum_n |c_n(t)|^2 = 1$. With this definition, the projection of $\ket{\Psi_t}$ onto the $N$-low-energy subspace of Hamiltonian $H(t)$ is the sum of the first $N$ probability spectrum.

The probability spectrum of QHD is a dynamical property of the system, and we can use it to characterize different stages in the quantum evolution. In \fig{three_phase}~(B), we numerically calculate the probability spectrum of QHD for Levy function. We also plot the quantum wave packet at various $t$ in \fig{three_phase}~(A). The wave function data is from the numerical simulation of QHD as in \sec{nonconvex}. The probability spectrum is depicted on a heatmap with horizontal axis being the time span $t$ and vertical axis being the energy levels $n$. Each cell in the heatmap is colored by the corresponding numerical value of $|c_n(t)|^2 = |\braket{n;t}{\Psi_t}|^2$. 

We observe an interesting probability transition pattern in \fig{three_phase}~(B). In the very beginning of QHD, there are two clusters in the probability spectrum: one low-energy cluster at $n=0$ and one high-energy cluster at $n = 4,5$. There are some probability exchange in the high-energy cluster but it does not interact with the low-energy one. After $t = 1$, the high-energy cluster starts to get absorbed into the low-energy subspace. As time goes by, the high-energy cluster evaporates and the quantum state completely sits in the low-energy subspace. The same trend is repeatedly observed in the probability spectrum of QHD (with other test functions). Motivated by the findings, we introduce the three-phase picture of QHD.

\subsection{Three phases in QHD}
We divide the course of QHD evolution into the following three phases based on the behavior of the wave function and the direction of energy/probability flows:

\paragraph{Kinetic phase.}
In the first phase of QHD, the wave function is of ample kinetic energy and it rapidly bounces within the whole search space (see $t=0.1, 0.5$ in \fig{three_phase}~(A)). While the majority of probability spectrum is in the low-energy subspace, a mid- or high-energy cluster in the probability spectrum prevails. The two energy clusters co-exist and almost do not interact. We call this phase the \textit{kinetic} phase of QHD because it is characterized by the mobility of wave functions (as a result of the dominating kinetic energy term).

\paragraph{Global search phase.}
In the second phase of QHD, the kinetic energy in the system starts to drain out. The wave function becomes less oscillatory and shows a selectivity toward the global minimum of $f$ (see $t=1,2$ in \fig{three_phase}~(A)). In the probability spectrum, the high-energy cluster in the wave function is driven toward the low-energy subspace, and this trend does not reverse. We call this phase the \textit{global search} phase because the quantum system manages to locate the global minimum of $f$ after the screening of the whole search domain in the first phase. This phase separates QHD from other classical gradient methods.

\paragraph{Descent phase.}
In the last phase of QHD, the wave function settles and becomes increasingly concentrated near the global minimizer of $f$ (see $t=5, 10$ in \fig{three_phase}~(A)). The wave function stays in the low-energy subsapce. In this phase, the quantum evolution enters the \textit{semi-classical regime} in which the potential energy term dominates. QHD starts to behave like classical gradient descent as it converges to the global minimizer $x^*$. We call this phase the \textit{descent} phase of QHD.

\subsection{Quantitative analysis}
The probability spectrum of QHD witness a structural change only in the global search phase, in which a significant probability transition from the high-energy subspace to the low-energy subspace takes place. While in the kinetic phase and descent phase of QHD, we do not see an outstanding shift in the probability spectrum. This phenomenon can be explained by quantitative arguments.

\subsubsection{Suppressed probability transition in the kinetic and descent phase}

First, we focus on the underlying mechanism in QHD that discourages the probability transition in the kinetic and descent phase. For simplicity, we assume the dimension $d=1$ in the following discussion, while it can be readily generalized to arbitray finite dimensions.

The QHD Hamiltonian takes the form $H(t) = e^{\varphi_t} \left(-\frac{1}{2}\frac{\partial^2}{\partial x^2}\right) + e^{\chi_t}f(x)$. Suppose the Hamiltonian is written as $H(t) = \sum_n E_n(t) \ket{n;t}\bra{n;t}$ where $E_n(t)$, $\ket{n;t}$ are eigenvalues and eigenstates of the system at time $t$. As in the proof of \thm{adiabatic_limit_qhd}, we represent the wave function as $\ket{\Psi_t} = \sum_n c_n(t) e^{i\theta_n(t)}\ket{n;t}$, where $\theta_n(t) = -\int^t_0 E_n(s)~\d s$. Recall from \eqn{ode_form} that the amplitudes $c_n(t)$ are determined by the equations:
\begin{align*}
    \dot{c}_n(t) = \underbrace{-c_n(t) \bra{n;t}\left[\frac{\partial}{\partial t} \ket{n;t}\right]}_{\text{Part~I.}} - \underbrace{\sum_{k\neq n} c_k(t) e^{i(\theta_k - \theta_n)} \frac{\bra{n;t}\dot{H}\ket{k;t}}{E_k(t) - E_n(t)}}_{\text{Part~II.}}.
\end{align*}
As we have shown, Part I does not contribute to the change in amplitudes because $\bra{n;t}\left[\frac{\partial}{\partial t} \ket{n;t}\right]$ is pure imaginary. A further understanding of the amplitude evolution amounts to investigating Part II in the equation.

In the kinetic phase, it expected that the kinetic operator $\left(-\frac{1}{2}\frac{\partial^2}{\partial x^2}\right)$ plays a dominant role in the system, while the potential operator $f(x)$ is minor. In this case, QHD behaves like free particle evolution:
\begin{align}
    H^{kinetic}(t) \propto \left(-\frac{1}{2}\frac{\partial^2}{\partial x^2}\right).
\end{align}
We then approximate the eigenstates of $H(t)$ by those of the kinetic operator. Suppose the QHD is simulated in the unit interval $[0,1]$ (with vanishing boundary conditions $\Psi(0,t) = \Psi(1,t)=0$), eigenstates of the kinetic operator are sine waves: $\ket{n;t} \propto \sin(n\pi x)$. Therefore, we can compute the inner product in Part II:
\begin{align}
    \bra{n;t}\dot{H}\ket{k;t} \propto \dot{\varphi}_t e^{\varphi_t} \int^1_0 \sin(n\pi x) \left(-\frac{1}{2}\frac{\partial^2}{\partial x^2}\right) \sin(k\pi x)~\d x = 0,
\end{align}
for $k \neq n$. Given the non-degeneratcy of $H(t)$, i.e., $E_k(t)- E_n(t) \neq 0$, we find that Part II is actually very close to zero in the kinetic phase. This fact implies that there is no amplitude exchange in the kinetic phase, which is consistent with our numerical evidence. 

In the descent phase, we do not see significant amplitude exchange in the probability spectrum as well. However, the underlying mechanisms are quite different. In this phase, QHD runs into the \textit{semi-classical} regime as the potential operator becomes the major contributor to the quantum dynamics. As shown in \sec{convergence-nonconvex}, the eigenstates of the QHD Hamiltonian can be approximated by those of a harmonic oscillator: 
\begin{align}
    H^{descent}(t) \approx \hat{H}_{HO}(t) = \hat{K}(t) + \hat{V}(t),
\end{align}
where $\hat{K}(t) = e^{\varphi_t}\left(-\frac{1}{2} \frac{\partial^2}{\partial x^2}\right)$ and $\hat{V}(t) = e^{\chi_t} \frac{f''(x^*)}{2} (x-x^*)^2$. As a result, $\dot{H}(t) \approx \dot{\varphi}_t \hat{K} + \dot{\chi}_t \hat{V}$ and we use \lem{expect_potential} to compute the inner product between adjacent energy levels in Part II:
\begin{align}
    \bra{n;t}\dot{H}\ket{n\pm 1;t} = \dot{\varphi}_t \bra{n;t}\dot{K}\ket{n\pm 1;t} + \dot{\chi}_t\bra{n;t}\dot{V}\ket{n\pm 1;t} = 0,
\end{align}
while the contributions from non-adjacent energy levels are minor due to the wider energy gaps (i.e., a bigger denomenator $(E_k - E_n)$). We conclude that, in the descent phase, Part II in the amplitude equation is very small so there is little energy exchange in QHD.

\subsubsection{Energy ratio and phase transitions}
We consider the energy ratio $E_1(t)/E_0(t)$, i.e., the first excited energy $E_1(t)$ to the ground energy $E_0(t)$. The energy ratio serves as a good indicator of the phase transitions in QHD; meanwhile, it also sheds light on the mechanism of probability transition in the global search phase. In what follows, we assume QHD is defined for two-dimensional problems so that we can refer to our experiment results in \sec{nonconvex}. 

We begin with two extreme cases. At the very beginning of the kinetic phase ($t = 0$), the system Hamiltonian is dominated by the kinetic term $(-\nabla^2/2)$. This means that we can ignore the structure of $f$ and push the Hamiltonian to the kinetic limit, $H(0) \approx -\nabla^2/2$. We denote the eigenstates of the kinetic operator $(-\nabla^2/2)$ as $\ket{\psi_{\vect{n}}}$, where the index is $\vect{n} = (n_x,n_y)$ with $n_x,n_y = 1,2,...$:
\begin{align*}
    \ket{\psi_{\vect{n}}} = \sqrt{2}\sin(n_x \pi x)\sin(n_y \pi y),
\end{align*}
and correspondingly, the eigenvalues are $E_{\vect{n}} = (n^2_x+n^2_y)\pi^2/2$. In particular, the ground energy is $E_0 = \pi^2$ and the first excited energy (with degeneracy) is $E_1 = 5\pi^2/2$. The energy ratio is $E_1/E_0 = 2.5$. Clearly, the deviation of the energy ratio from the kinetic limit $2.5$ indicates a transition from the kinetic phase to the global search phase; see \fig{three_phase}~(D).

In the descent phase, the Schr\"odinger operator $H(t)$ enters the semi-classical regime so it can be effectively approximated by a quantum harmonic oscillator (see \sec{convergence-nonconvex}). The eigenvalues of two-dimensional harmonic oscillators are 
\begin{align}\label{eqn:qho_eigvals}
    E_{\vect{n}} = \frac{\hbar}{2} \left[(n_1+1)\omega_1+(n_2+1)\omega_2\right],
\end{align}
with $n_1, n_2 = 0,1,2,...$.\footnote{The $\hbar$, $\omega$ here are determined by the time-dependent parameters and the neighborhood information of $x^*$. Details can be found in \sec{convergence-nonconvex}.} Therefore, the ground energy is $E_0 = \hbar (\omega_1+\omega_2)/2$ and the first excited energy (assuming $\omega_1 \ge \omega_2$) is $E_1 = \hbar (\omega_1+3\omega_2)/2$. Hence, the energy ratio at the semi-classical limit is 
\begin{align}\label{eqn:semiclassical_ratio}
    \lim_{t\to \infty}\frac{E_1}{E_0} = \frac{\omega_1 + 3\omega_2}{\omega_1 + \omega_2}.
\end{align}
The convergence of the energy ratio to \eqn{semiclassical_ratio} indicates the transition from the global search phase to the descent phase; see \fig{three_phase}~(E).

In the global search phase, the energy ratio changes and it has different values than the two extreme cases. Usually, it decreases from the initial value $2.5$ and gradually converges to the semi-classical limit. Intuitively, the shrinking energy gap between $E_0$ and $E_1$ makes the probability transition between energy levels much easier. This can be a reason why the high-energy cluster is absorbed into the low-energy cluster in the global search phase.

In fact, from the energy ratio, we can theoretically predict the times at which phase transition happens. For example, we can approximately predict the two phase-transition times in QHD for Levy functions (see \fig{three_phase}~(B)): the energy ratio significantly deviates from the initial value $2.5$ after $t = 0.5$, which marks the first phase transition (kinetic to global search); after $t = 5$, the energy ratio curve starts to converge to the semi-classical limit, indicating the termination of the global search phase.

\paragraph{Energy ratio for Levy function.}
The Levy function (for dimenson 2) is defined by 
\begin{align}
    f(w_1, w_2) = \sin(\pi w_1)^2 + (w_1-1)^2\left[1+10\sin(\pi w_1+1)^2\right] + (w_2-1)^2\left[1+\sin(2\pi w_2)^2\right],
\end{align}
where $w_j = 1 + (x_j-1)/4$ for $j = 1,2$. When used as an optimization test problem, the Levy function is usually evaluated on the square $[-10,10]^2$. It has a unique global minimum $ f(1,1) = 0$. The Hessian matrix of the Levy function at $(1,1)$ is a diagonal matrix with two positive eigenvalues: $\lambda_1 = (\pi^2 + 1 + 10\sin(1)^2)/8$, and $\lambda_2 = 1/8$. Therefore, the second-order Taylor expansion (quadratic model) of $f$ at the global minimizer $x^*=(1,1)$ is:
\begin{align}
    f(x_1,x_2) \approx \frac{1}{2}\omega^2_1(x_1-1)^2 + \frac{1}{2}\omega^2_2(x_2-1)^2,
\end{align}
where $\omega_1 = \sqrt{\lambda_1}\approx 1.4979$, $\omega_2 = \sqrt{\lambda_2} \approx 0.3536$. By \eqn{semiclassical_ratio}, the energy ratio at the semi-classical limit is $E_1/E_0 \approx 1.3819$, which is quite close to the numerical results shown in \fig{three_phase}~(D).

\subsection{Comparison with QAA}
In our experiment, we test two quantum algorithms (QHD and QAA) and two classical algorithms (NAGD and SGD). Although both quantum algorithms are formulated as continuous-time quantum evolutions, QHD appears to outperform QAA in most test problems. The three-phase picture of QHD lends us a pivot to understand why QHD is better than QAA in solving non-convex optimization problems.

\begin{figure}[ht!]
    \centering
    \includegraphics[width=14cm]{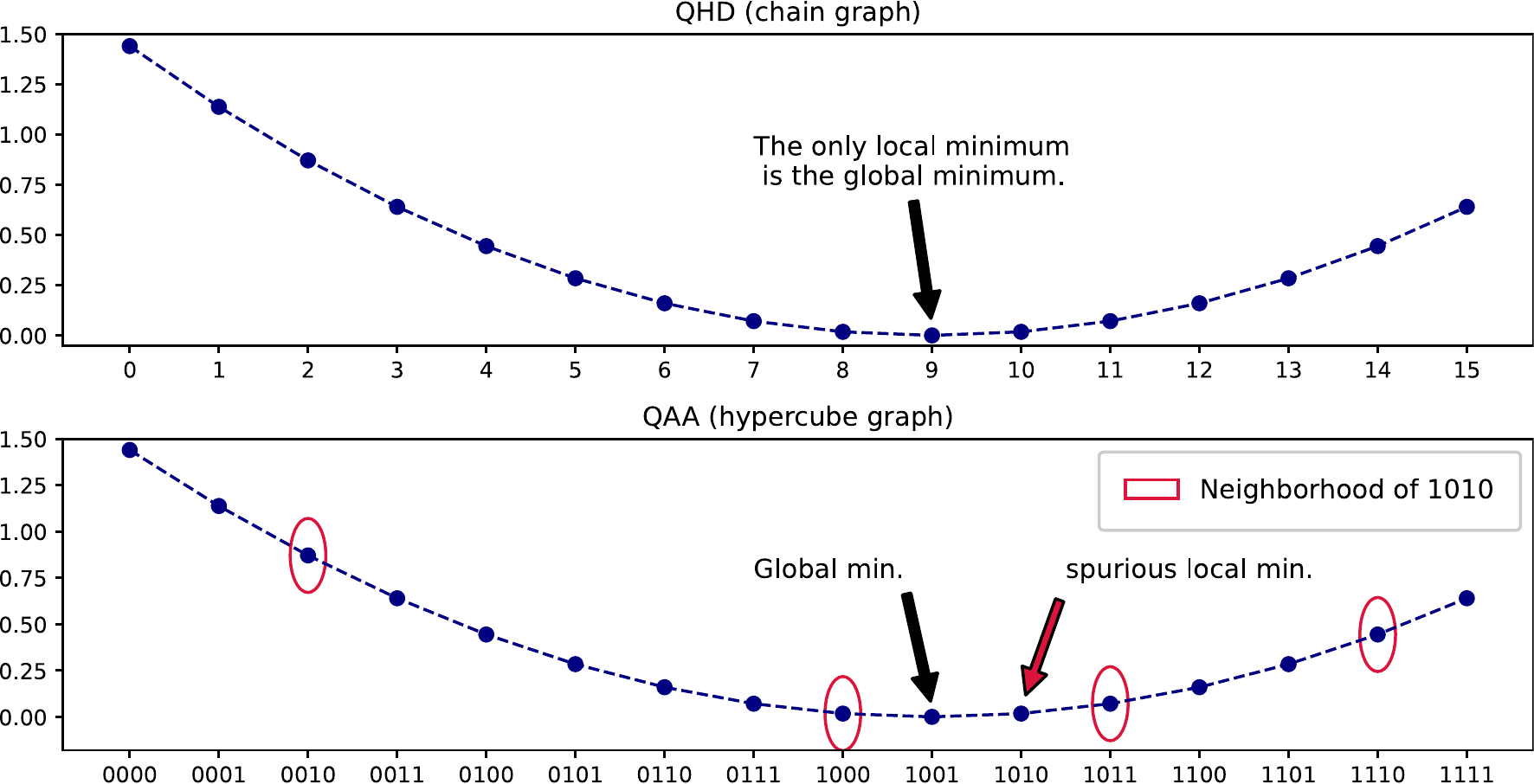}
    \caption[Quantum adiabatic algorithms do not preserve convexity for continuous problems.]{The hypercube search structure in QAA twists the landscape of $f$. \textbf{Upper:} in QHD (problem discretized as an optimization on an one-dimensional chain graph), each node has two neighborhood nodes, and there is only one local minimum. \textbf{Lower:} in QAA (problem discretized as an optimization on a 4-dimensional hypercube graph), the node \texttt{1010} becomes a spurious local minimizer as it has the least function value compared to its neighborhood. This means QAA maps a \textit{convex continuous} problem to a \textit{non-convex discrete} problem.}
    \label{fig:graph_comparison}
\end{figure}

The key difference between the two quantum algorithms is that they have different ``kinetic'' operator in the system Hamiltonian. The kinetic operator $(-\nabla^2/2)$ in QHD comes from the Euclidean distance function in the Bregman-Lagrangian framework. This differential operator stands for the \textit{quantized} Euclidean distance and it allows QHD to \textit{see} the continuous structure of $f$ -- the continuity of $f$ is defined by the Euclidean topology in $\R^d$. When the spatial discretization is applied, the finite difference discretization of the kinetic operator $(-\nabla^2/2)$ makes it resembles the graph Laplacian of a regular lattice. However, in QAA, the ``kinetic operator'' is the initial hamiltonian $H_0$ (see \eqn{qaa_hamiltonian}), interpreted as the adjacency matrix of the hypercube graph. This means QAA regards the task as a discrete optimization problem on a hypercube graph. This misinterpretation of the underlying topology of continuous optimization problems leads to several serious consequences:

\begin{enumerate}
    \item QAA twists the optimization landscape of $f$, making the problem much harder to solve. It is shown in \fig{graph_comparison} that, even for a convex $f$, the hypercube search graph in QAA creates many spurious local minimum in the problem. This makes the discretized problem much more challenging than the original continuous problem. 
    \item The performance of QAA heavily depends on the resolution in the discretization. Since the naive discretization in QAA does not preserve the continuous structure of the original problem, a higher resolution can result in a harder problem, thus even worse performance in QAA. For example, when applied to the Hosaki function, the performance of QAA becomes significantly worse when we choose a higher resolution (128), while QHD does not suffer from this issue, see \fig{compare_qaa_qhd}. We observe similar patterns in other instances as well.
    \item Compared to QHD, QAA does not have the kinetic phase. In the kinetic phase, QHD averages the initial wave function over the whole feasible domain, thus reducing the risk of poor random initializations. In QAA, however, the quantum state is supposed to remain the ground state of the Hamiltonian, to there is no preparatory/warm-up stage.

    \item QAA does not have a descent phase, either. The semi-classical approximation theory fails to hold in QAA so that the minimal spectral gap (i.e., $\lim_{0\le t\le T}E_1(t) - E_0(t)$) can be arbitrarily small when the spatial discretization stepsize shrinks.\footnote{To see this, note that $\lim_{0\le t\le T}E_1(t) - E_0(t) \le E_1(T)-E_0(T)$. At $t=T$, the QAA Hamiltonian is just the problem Hamiltonian $H_P$, which is a discretization of the continuous function $f$. Note that the spectral gap in $H_P$ corresponds to the resolution in the discretization of $f$.} In contrast, the energy gaps in QHD is insensitive to spatial discretization. Therefore, we can not use fast-varying time-dependent functions in QAA since they will drive the quantum state out of the ground-energy subspace. This fact prevents us from achieveing faster convergence in QAA.
\end{enumerate}

\section{Digital Implementation of QHD}
\label{sec:digital-implementation}
In this section, we discuss how to implement QHD on digital (i.e., circuit-model) quantum computers. We assume the access to the quantum evaluation oracle $O_f$ which allows us to query the function value of $f$ in a coherent way:

\begin{assumption}\label{assump:oracle}
    We assume a unitary $O_f: \R^d\times \R \to \R^d \times \R$ such that for any $x \in \R^d$, and $z \in \R$, 
    \begin{align*}
        O_f \ket{x, z} \to \ket{x, f(x)+z}.
    \end{align*}
\end{assumption}

\subsection{Implementing QHD using product formulae}
In the quantum simulation literature, several algorithms have been proposed for simulating real-space quantum dynamics, including Trotter-type algorithms (product formulae) \cite{an2021time}, truncated Taylor series (with finite difference scheme \cite{kivlichan2017bounding} or Fourier spectral method \cite{childs2022quantum}), and interaction picture simulations \cite{childs2022quantum}. Although methods like truncated Taylor series or interaction picture simulations achieve better asymptotic scaling in terms of the error dependence, we will consider Trotter-type algorithms for implementing QHD. Quantum simulation algorithms based on product formulae usually take a simpler form and thus more resource-efficient in practice \cite{childs2018toward}. Besides, product formulae are closely related to a family of numerical schemes known as \textit{symplectic integrators} \cite{hatano2005finding}, which are proven to be the key in the time discretization of the Bregman-Lagrangian framework \cite{betancourt2018symplectic}.

\vspace{4mm}The product-formula-based implementation of QHD basically follows from \cite[Section 2.3]{childs2022quantum}. It takes two steps:

\paragraph{Spatial discretization.} We introduce a standard mesh-grid discretization of the continuous space. Suppose the objective function $f$ is unbounded at infinity, then there must exist a large enough region $\Omega = [a,b]^d$ such that the global minimizer of $f$ is within $\Omega$. We consider a regular mesh grid $\mathcal{M}$ in the domain $\Omega=[a,b]^d$ and divide each edge into $N$ cells, i.e., 
\begin{align}
    \mathcal{M} = \{(x_1,...,x_d): x_j \in \{a, a+\Delta x, ..., b-\Delta x\}\},
\end{align}
where $\Delta x = (b-a)/N$ is the stepsize of the mesh $\mathcal{M}$. Let $N = 2^q$ for an integer $q$, and we assign a $q$-qubit register to represent a single edge: computational basis of the quanutm register enumerates mesh coordinates $x_j \in \{a, a+\Delta x, ..., b-\Delta x\}$. The full full mesh $\mathcal{M}$ is represented by concatenating $d$ such registers: the mesh point $\vect{x} = (x_1,...,x_d)$ is represented by the basis state $\ket{x_1,...,x_d}$. With this mesh-grid discretization, a quantum wave function in the real space $\Psi(x)\colon \Omega \to \C$ can be represented as a quantum state in the register:
\begin{align}\label{eqn:wave-rep}
    \ket{\psi} = \sum_{\vect{x} \in \mathcal{M}}c(x_1,...,x_d)\ket{x_1,...,x_d},~\|\psi\|^2_2 = \sum_{\vect{x} \in \mathcal{M}}|c(x_1,...,x_d)|^2 = 1.
\end{align}

\paragraph{Time discretization.} We leverage the product formula to to decompose the quantum evolution into an alternating sequence of kinetic and potential evolutions. Given a Hamiltonian of the form $H = A+B$, the first-order Trotter-Suzuki formula \cite{suzuki1991general} is defined as $\mathscr{T}_1(t) = e^{-itA}e^{-itB}$. Recall the QHD Hamiltonian takes the form $H(t) = e^{\varphi_t}(-\nabla^2/2) + e^{\chi_t}f$, let $U(0,T)$ be the time-evolution operator in QHD, we can approximate $U(0,T)$ by a sequence of product formulae: 

\begin{align}\label{eqn:trotterized_qhd_evolution_operator}
    U(T,0) \approx \mathcal{T} \prod^{N-1}_{j=0} \Big[ \exp(-i \Delta t a_j (-\nabla^2/2)) \exp(-i\Delta t b_j f)\Big],
\end{align}
where $\mathcal{T}$ is the time-ordering operator, $a_j = e^{\varphi_{t_j}}$, and $b_j = e^{\chi_{t_j}}$.

\vspace{4mm}Assuming the domain $\Omega$ is of periodic boundary condition, one observes that the Laplacian operator $\nabla^2$ is diagonalized by the Fourier basis. This observation leads to a neat implementation of the exponential of kinetic operator $(-\nabla^2/2)$ \cite[Section 3.2]{childs2022quantum}:

\begin{align}
    -\frac{1}{2}\nabla^2 \to \mathbf{F}_s \mathbf{L}\left(\mathbf{F}_s\right)^{-1},~ \exp(-i \Delta t a_j (-\nabla^2/2)) \to \mathbf{F}_s \exp(-i \Delta t a_j \mathbf{L})\left(\mathbf{F}_s\right)^{-1},
\end{align}
where $\mathbf{F}_s$ is the \textit{quantum shifted Fourier transform} (QSFT)\footnote{QSFT is a close variant of the Quantum Fourier Transform (QFT) and can be efficiently implemented on quantum computers. According to \cite[Lemma 5]{childs2021high}, QSFT can be performed to an $k$-dimensional vector with gate complexity $\log(k) \log\log (k)$.}, and $\mathbf{L}$ is a diagonal matrix with eigenvalues of $(-\nabla^2/2)$. These eigenvalues can be computed from the frequency number of a Fourier basis function \cite[Eq.(26)]{childs2022quantum}. On the other hand, the operator $\exp(-i\Delta t b_j f)$ can be readily computed because $f$ is diagonal and can be directly queried from the oracle $O_f$ at any point $\vect{x} \in \mathcal{M}$.

We now summarize the full procedures of the product formula implementation of QHD in \algo{trotter_qhd}. We also analyze the gate complexity of this algorithm, see \thm{gate-complexity}.

\begin{algorithm}[htbp]
    \caption{Implementation of QHD by product formulae}
    \label{algo:trotter_qhd}
    Specify a hypercubic domain $\Omega = [a,b]^d$ where $R$ is prefixed or determined by the knowledge of $f$, and perform the regular-mesh discretization $$\mathcal{M} = \{(x_1,...,x_d): x_j \in \{a, a+\Delta x, ..., b-\Delta x\}, \text{~where~}\Delta x = (b-a)/2^q\}.$$\\
    Initialize a $dq$-qubit quantum register to $\ket{0}^{\otimes dq}$, where the parameter $d$ is the dimension of the problem, and $q$ is a pre-fixed precision parameter. Note that all grid points in $\mathcal{M}$ are enumerated by computational basis $\ket{x_1,...,x_d}$ of the quantum register.\\
    Prepare the \textit{initial guess state} $\ket{\psi_0}$. For instance, it can be the uniform superposition state (i.e., uniform distribution) or a Gaussian state (i.e., Gaussian distribution).\\
    Specify the time-dependent parameters $\varphi_t$, $\chi_t$, the number of iterations $R$, and the learning rate $s$. The (effective) end time is $T = Ns$.\\
    Inductively apply the following update rule to the quantum state for $j = 0, 1, ..., N-1$: $$\ket{\psi_{j+1}} = \mathbf{F}_s \exp(-is a_j \mathbf{L}) \mathbf{F}^{-1}_s \exp(-is b_j \mathbf{V})\ket{\psi_j},$$
    where $a_j = e^{\varphi(t_j)}$, $b_j = e^{\chi(t_j)}$, $t_j = js$, $\mathbf{L}$ is defined as above, and $\mathbf{V}$ is a diagonal operator such that $\mathbf{V}\ket{\vect{x}} = f(\vect{x})\ket{\vect{x}}$ for any $\ket{\vect{x}}\in \mathcal{M}$. We end up with a quantum state $\ket{\psi_N}$.\\
    Measure the quantum state $\ket{\psi_N}$ with computational basis (i.e., position observable), and the outcome is a grid point $\ket{\vect{x}} \in \mathcal{M}$.
\end{algorithm}

\subsection{Asymptotic complexity}
We show that the gate complexity of \algo{trotter_qhd} is linear in the dimension $d$ of the problem and the number of iterations $R$. Our result is exponentially better than the previous quantum gradient descent algorithm~\cite{rebentrost2019quantum} in terms of the iteration number $R$. 

\begin{theorem}\label{thm:gate-complexity}
    Let $d$ be the dimension of the problem, $q$ be a pre-fixed precision parameter, and $R$ be the number of iterations. The gate complexity of \algo{trotter_qhd} is $\tilde{O}(dqR)$.
\end{theorem}
\begin{proof}
    First, we study the cost of each iteration step. Let $n = dq$ be the total number of qubits. Note that the operator $\mathbf{V}$ is diagonal in the computational basis and is given by the oracle $O_f$, the evolution operator $\exp(-is b_j \mathbf{V})$ can be computed by standard techniques (see for example Rule 1.6 in \cite[Section 1.2]{childs2004quantum}) with gate complexity $\tilde{O}(n)$. The quantum shifted Fourier transform $\mathbf{F}_s$ is performed with gate complexity $O(n\log(n))$ \cite[Lemma 5]{childs2021high}, so is its inverse $\mathbf{F}^{-1}_s$. The operator $\mathbf{L}$ is diagonal in the computational basis and the diagonal entries are evaluated from a closed-form formula \cite[Section 2.1]{childs2022quantum}. Similar as $\exp(-is b_j \mathbf{V})$, the evolution operator $\exp(-is a_j \mathbf{L})$ is computed with gate complexity $\tilde{O}(n)$. It turns out that the gate complexity of each iteration is $\tilde{O}(n)$.
    
    As the algorithm iterations in $R$ steps, and each step has exactly the same gate complexity, the overall gate complexity is $\tilde{O}(dqR)$.
\end{proof}

\section{Analog Implementation of QHD for quadratic programming}
\label{sec:analog-implementation}
In this section, we discuss how to use QHD to solve quadratic programming problems (with box constraints) on analog quantum computers. Quadratic programming with box constraints is formulated as follows:
\begin{subequations}\label{eqn:qp-analog}
\begin{align}
    \text{minimize}&\qquad f(\x)=\frac{1}{2} \x^\top \mathbf{Q} \x + \mathbf{b}^\top \x, \label{eqn:qp-obj-analog}\\
    \text{subject to}
    &\qquad \mathbf{0} \preccurlyeq \x \preccurlyeq \mathbf{1},
\end{align}
\end{subequations}
where $\vect{x}\in \R^d$,  $\mathbf{0}$ and $\mathbf{1}$ are $d$-dimensional vectors of all zeros and all ones, respectively.

\subsection{Quantum Ising Machines}
Among various types of analog quantum simulators, we will be focusing on an important class of simulators whose evolutions are effectively described by quantum Hamiltonians in the Ising form: 
\begin{align}\label{eqn:ising-machine}
    H(t) = - \frac{A(t)}{2} \left(\sum_j \sigma^{(j)}_x\right) + \frac{B(t)}{2} \left(\sum_j h_j \sigma^{(j)}_z + \sum_{j>k} J_{j,k} \sigma^{(j)}_z \sigma^{(k)}_z\right),
\end{align}
where the operator $\sigma^{(j)}_\alpha$ for $\alpha = x,y,z$ stands for Pauli operator acting on the $j$-th qubit, $A(t)$ and $B(t)$ are time-dependent functions. Notable examples in this class include superconducting qubits spin glass simulators \cite{harris2018phase} and Rydberg atoms \cite{saffman2010quantum}. The technology of building such quantum simulators has seen a great advancement in the past decades, leading to several cloud-based commercial quantum processors (e.g., D-Wave, QuEra, etc). In order to facilitate further theoretical discussions, we propose a formal abstraction called the \textit{quantum Ising machine} for quantum simulators with Ising-type Hamiltonian.

\begin{definition}\label{defn:quantum-ising-machine}
    A \textbf{Quantum Ising Machine} embodies an $n$-qubit quantum register which permits the following operations:
    \begin{enumerate}
        \item \textbf{Initialization}: the quantum register is initialized to certain quantum state, e.g., the all-zero state $\ket{0}^{\otimes n}$ or the uniform superposition state $\ket{+}^{\otimes n}=\ket{\psi_0}=\frac{1}{\sqrt{2^n}}\sum_{b\in\{0,1\}^n} \ket{b}$.
        \item \textbf{Simulation}: the evolution of the quantum register is described by the Schrodinger equation,
        \begin{equation}\label{eqn:analog-evolution}
            i h\frac{\d }{\d t} \ket{\psi_t} = H(t) \ket{\psi_t},
        \end{equation}
        where $h$ is the Planck constant, and the Hamiltonian $H(t)$ is given by \eqn{ising-machine}. We assume $A(t)$, $B(t)$, $h_j$, $J_{j,k}$ are programable parameters.
        \item \textbf{Measurement}: the quantum register is measured in computational basis, i.e., we can compute $|\braket{b}{\psi_t}|^2$ with $b \in \{0,1\}^n$ at certain time $t \ge 0$.
    \end{enumerate}
\end{definition}

In practice, it is often convenient to divide $h$ on both sides of \eqn{analog-evolution} and consider time-dependent functions $A(t)/h$ and $B(t)/h$. Usually, the time-dependent functions are given in giga-Hertz (1GHz = $10^9$Hz); accordingly, the simulation time is given in nano-second (1ns=$10^{-9}$s).

\subsection{Implementing QHD on Quantum Ising Machines}\label{sec:restricted_implementation}

\subsubsection{Finite difference discretization of PDEs}
\label{sec:fdm}
Suppose the feasible domain of an optimization problem is the $[0,1]^d$, QHD is formulated as the partial differential equation (PDE):
\begin{align}\label{eqn:standard-qhd}
    i \frac{\partial}{\partial t}\Psi(t,x) = \left[\underbrace{-\frac{e^{\varphi_t}}{2}\nabla^2}_{\text{Kinetic part}} + \underbrace{e^{\chi_t}f(x)}_{\text{Potential part}}\right]\Psi(t,x),
\end{align}
where $\Psi(t,x)\colon [0,T]\times [0,1]^d\to \C$ with the vanishing boundary condition $\Psi(t,x)=0$ for $t\in[0,T]$, $x\in \partial [0,1]^d$. To simulate the real-space dynamics as described by \eqn{standard-qhd}, the first step is to discretize the PDE so that it becomes a finite-dimensional ODE. This goal is achieved by introducing the finite difference method (FDM) \cite{morton2005numerical}.

We begin with the one-dimensional case. Suppose the wave function $\Psi(t,x)\colon [0,T]\times [0,1]$ and $\Psi(t,0)=\Psi(t,1)=0$ for $t\in [0,T]$. We divide the unit interval $[0,1]$ into a $r$-cells mesh grid: $\mathcal{M}=\{a_j =j/r: j = 0,1,...,r\}$. At time $t$, the wave function $\Psi(t, x)$ is sampled on the mesh grid $\mathcal{M}$ and discretized to an $(n+1)$-dimensional vector:
\begin{align}\label{eqn:state-discretize}
    \Psi(x) \mapsto \ket{\phi} = \frac{1}{C} \sum^r_{j=0} \Psi(a_j) \ket{j},
\end{align}
where $C = \sqrt{\sum^r_{j=0} |\Psi(a_j)|^2}$ is a normalization constant such that $\|\phi\| = 1$. Due to the unit probability $\int^1_0|\Psi(x)|^2~\d x = 1$, we have that $C \approx \sqrt{r}$ for large $r$ because of the convergent Riemann sum:
\begin{align}
    \lim_{r\to \infty} \sum^r_{j=0} \frac{|\Psi(x_j)|^2}{r} = \int^1_0|\Psi(x)|^2~\d x = 1.
\end{align}

The second-order derivative of $\Psi$ is approximated by the central finite difference method:
\begin{align}
    \frac{\partial ^2}{\partial x^2} \Psi \Big|_{x=a_j} \approx \frac{\Psi(a_{j+1})-2\Psi(a_{j})+\Psi(a_{j-1})}{(1/r)^2}.
\end{align}
Since the potential operator $f$ acts on the wave function $\Psi$ by pointwise multiplication, the discretization of $f$ turns out to be a diagonal matrix. Therefore, the finite-difference-discretized PDE takes the form:
\begin{align}\label{eqn:1d-qhd-discrete}
    i \frac{\d}{\d t}\ket{\phi_t} = \left[\underbrace{-\frac{e^{\varphi_t}}{2}\hat{L}}_{\text{Kinetic part}} + \underbrace{e^{\chi_t}\hat{F}}_{\text{Potential part}}\right]\ket{\phi_t},
\end{align}
where $\ket{\phi_t}\colon[0,T]\to \C^{r+1}$ is the discretized wave function, $\hat{L}$ and $\hat{F}$ are $(r+1)$-by-$(r+1)$ matrices:
\begin{align}\label{eqn:discretize-pde}
    \hat{L} = \frac{1}{(1/r)^2} \begin{bmatrix} -2 & 1 & & \\
    1 & -2 & 1 & \\
    ...& ... & ... &...\\
    & 1 & -2 & 1\\
    & & 1 &-2\\
    \end{bmatrix},
    \hat{F} = \begin{bmatrix} f(a_0) &  & & \\
         & f(a_1) &  & \\
        ...& ... & ... &...\\
        & & f(a_{r-1}) & \\
        & & & f(a_r)\\
    \end{bmatrix}.
\end{align}

\begin{remark}\label{rem:chain-adjacency}
    We can write the matrix $\hat{L} = \frac{1}{(1/r)^2} (\hat{A}-2 I)$, where $\hat{A}$ is a tridiagonal matrix with zero main diagonals and ones on the first diagonal above and the first diagonal below. $I$ is the identity matrix. It is worth noting that $\hat{A}$ is the adjacency matrix of a 1-dimensional chain graph with $(r+1)$ nodes. 
\end{remark}

In the end of QHD evolution, we want to measure the wave function with position observable $\hat{x}$. Since the wave function $\Psi(t,x)$ is discretized to $\ket{\phi_t} \in \C^{r+1}$, the position observable must be discretized accordingly. It turns out that the discretized position observable, denoted by $\hat{X}$, is of the form:
\begin{align}\label{eqn:observable-discretize}
    \hat{X} = \sum^r_{j=0} a_j \ket{j}\bra{j} = \sum^r_{j=0} \frac{j}{r} \ket{j}\bra{j}.
\end{align}
To see why \eqn{observable-discretize} is valid, it is sufficent to verify that
\begin{align}
    \lim_{r\to \infty} \bra{\phi}\hat{X}\ket{\phi} &= \lim_{r\to \infty} \frac{1}{\sum^r_{j=0} |\Psi(x_j)|^2} \sum^r_{j=0}a_j|\Psi(a_j)|^2,\\
    & =  \int^1_0 x|\Psi(x)|^2~\d x = \braket{\Psi}{\hat{x}|\Psi}.
\end{align}

The finite difference discretization is readily generalized to arbitrary finite dimensions by forming a regular mesh grid over $[0,1]^d$. Introducing a resolution parameter $r \in \Z$, which is the number of cells on each edge of the feasible domain $[0,1]^d$, the mesh can be written as $\mathcal{M}_d=\{0,1/r, 2/r,...,1\}^d$. A continuous-space wave function $\Psi(x)$ is discretized to an $(r+1)^d$-dimensional unit vector:
\begin{align}
    \Psi(x) \mapsto \ket{\phi} = \frac{1}{C} \sum^{r}_{j_1,...,j_d=0} \Psi(a_{j_1},...,a_{j_d}) \ket{j_1,j_2,...,j_d},
\end{align}
where $C = \sqrt{\sum_{j_1,...,j_d} |\Psi(a_{j_1},...,a_{j_d})|^2}$ is a normalization constant, and $\lim_{n\to \infty} C = \sqrt{dr}$.

Similarly, the finite differece discretization of the $d$-dimensional Laplacian operators $\nabla^2 = \sum^d_{k=1}\frac{\partial^2}{\partial x^2_k}$ is $\frac{1}{(1/r)^2} \hat{A}_d$ (ignoring a global phase $-2d$), where $\hat{A}_d$ is the adjacency matrix of a $d$-dimensional regular lattice graph:
\begin{align}\label{eqn:d-lattice-adjacency}
    \hat{A}_d = \sum^d_{k=1} I \otimes ... \otimes \underbrace{\hat{A}}_{\text{the $k$-th operator}} \otimes ... \otimes I,
\end{align}
and $A$ is the the same as defined in \rem{chain-adjacency}. The discretized potential operator $\hat{F}_d$ is given by
\begin{align}
    \hat{F}_d \ket{j_1,j_2,...,j_d}= f(a_{j_1},...,a_{j_1}) \ket{j_1,j_2,...,j_d}.
\end{align}
Moreover, the discretized position observable in $d$ dimension is a $d$-tuple $(\hat{X}^{(1)},...,\hat{X}^{(d)})$, where
\begin{align}\label{eqn:d-observable}
    \hat{X}^{(k)} = I \otimes ... \otimes \underbrace{\hat{X}}_{\text{the $k$-th operator}} \otimes ... \otimes I,
\end{align}
for $k = 1,2,...,d$. It follows that the finite difference discretization of QHD \eqn{standard-qhd} is
\begin{align}\label{eqn:finite-qhd-discrete}
    i \frac{\d}{\d t}\ket{\phi_t} = \left[-\frac{e^{\varphi_t}}{2}\hat{A}_d + e^{\chi_t}\hat{F}_d\right]\ket{\phi_t},
\end{align}
where $\ket{\phi_t}\colon[0,T]\to \C^{(r+1)^d}$ is the discretized wave function. When $f(x)$ is given as a quadratic function as in \eqn{qp-obj-analog}, the discretized $f$ takes the form:
\begin{align}\label{eqn:finite-quadratic-obj}
    \hat{F}_d = \frac{1}{2}\sum_{1\le p,q\le d} Q_{p,q}X^{(p)}X^{(q)} + \sum_{1\le p \le d} b_p X^{(p)},
\end{align}
where $\mathbf{Q} = (Q_{p,q})$ for $1\le p,q\le d$, and $\mathbf{b} = (b_p)$ for $1\le p \le d$.

\subsubsection{A relaxed implementation of QHD for quadratic programming}
Ideally, an analog implementation of QHD for quadratic programming problems amounts to simulate \eqn{finite-qhd-discrete} with the objective function \eqn{finite-quadratic-obj} on analog quantum computers. However, we do not find an exact approach to do so on Quantum Ising Machines because $A_d$ is very different from the off-diagonal part in the Quantum Ising Machine Hamiltonian; see \eqn{ising-machine}. Instead, we consider a \textit{relaxed} version of \eqn{finite-qhd-discrete}:
\begin{align}\label{eqn:relaxed-qhd-discrete}
    i \frac{\d}{\d t}\ket{\phi_t} = \left[-\frac{e^{\varphi_t}}{2}\hat{A}'_d + e^{\chi_t}\hat{F}_d\right]\ket{\phi_t},
\end{align}
in which have 
\begin{align}\label{eqn:d_relaxed_adjacency}
    \hat{A}'_d \coloneqq \sum^d_{k=1} I \otimes \ldots \otimes \underbrace{\hat{A}'}_{\text{the $k$-th operator}} \otimes \ldots \otimes I,
\end{align}
with 
\begin{align}\label{eqn:relaxed_adjacency}
    \hat{A}' \coloneqq \sum^{r}_{j=0}\sqrt{\frac{(j+1)(r+1-j)}{r+1}}\Big(\ket{j+1}\bra{j} + \ket{j}\bra{j+1}\Big).
\end{align}
Note that $\hat{A} = \sum^{r}_{j=0}\Big(\ket{j+1}\bra{j} + \ket{j}\bra{j+1}\Big)$, so the difference between $\hat{A}$ and $\hat{A}'$ is that $\hat{A}'$ has non-uniform tridiagonal elements. 

\begin{remark}
    In fact, the modified kinetic operator $A'$ may bring us extra benefits in solving constrained optimization problems (such as quadratic programming). Unlike the exact discretization $A$, the off-diagonal elements in $A'$ are not identical: except for the two end points $j = 0, r$, the values of the matrix elements are slightly larger than those in $A$. This uneven distribution of middle-range off-diagonal elements in $A'$ means that the wave function will gain more momentum (thus more significant quantum tunneling) near the center of the feasible domain $[0,1]$. Since the solutions of constrained optimization problems are usually close to the boundary of the feasible domain, more quantum tunneling in the middle of the feasible domain is likely to speed up the solution process.
\end{remark}

We will show that the relaxed QHD \eqn{relaxed-qhd-discrete} can be mapped to Quantum Ising Machines through a non-trivial encoding called \textit{Hamming encoding}. In \tab{HammingEncoding}, we summarize the one-to-one correspondence between the mathematical model \eqn{relaxed-qhd-discrete} and the physical implementation on Quantum Ising Machines.

\begin{table}[!ht]
    \centering
    \begin{tabular}{|c|c|c|}
        \hline
            \textbf{Parameters} & \textbf{Relaxed QHD} & \textbf{Quantum Ising Machine}\\
        \hline
            Initial states & $\ket{\phi_0}$: discrete Gaussian &  $\ket{\psi_0}$: uniform superposition\\
        \hline
            Kinetic operator & $-\frac{1}{2}\hat{A}'_d$ & $-\frac{1}{2r^{1/2}}S_x$, see \eqn{relaxed_Ad}\\
        \hline
            Potential operator & $\hat{F}_d$ & $H_P$, see \eqn{ising-format-qp}\\
        \hline
            Time-dependent function 1 & $e^{\varphi_t}$ & $A(t)/h$, see \eqn{time-dep-fun}\\
        \hline
            Time-dependent function 2 & $e^{\chi_t}$ & $B(t)/h$, see \eqn{time-dep-fun}\\
        \hline
            Measurement elements & $\{E^{(k)}_j\}$, see \eqn{Ekj} & $\{\mathscr{M}^{(k)}_j\}$, see \eqn{Mkj}\\
        \hline
    \end{tabular}
    \caption[Dictionary of Hamming encoding]{The dictionary of implementing relaxed QHD on Quantum Ising Machines.}
    \label{tab:HammingEncoding}
\end{table}

In what follows, we elaborate on some parameters shown in \tab{HammingEncoding}. Let $r\in\Z^+$ be a resolution parameter, $\lambda > 0$ be a time-rescaling parameter. $T$ is the total evolution time, $d$ is the dimension of the problem $f$, and $n=dr$ is the total number of qubits required.
 
\begin{itemize}
    \item Initial states:
    \begin{align}
        \ket{\phi_0} = \left(\sum^r_{j=0}\sqrt{\binom{n}{j}/2^r} \ket{j}\right)^{\otimes d},~\ket{\psi_0} = \ket{+}^{\otimes n} = \frac{1}{\sqrt{2^n}}\sum_{b\in{0,1}^n} \ket{b}.
    \end{align}

    \item Time-dependent functions:
        \begin{align}\label{eqn:time-dep-fun}
        A(t)/h = \lambda r^{3/2} e^{\varphi(t)},~B(t)/h = 2\lambda e^{\chi(t)}.
        \end{align}

    \item Problem Hamiltonian:
    \begin{align}\label{eqn:ising-format-qp}
        H_P = \sum^{n-1}_{j=0} h_j \sigma^{(j)}_z + \sum_{\substack{j,k=0 \\j> k}}^{n-1} J_{j,k} \sigma^{(j)}_z\sigma^{(k)}_z + g,
    \end{align}
    with $h_j = -\left[\left(\sum^d_{q=1}Q_{p,q}\right)+2 b_p\right]/4r$ for $(p-1)r \le j \le pr-1$, $J_{j,k} = \frac{Q_{p,q}}{4r^2}$ for $(p-1)r \le j \le pr-1$, $(q-1)r \le k \le qr-1$, and $g = \frac{1}{8}(1+\frac{1}{r})\sum_p Q_{p,p} + \frac{1}{4}\sum_{p>q} Q_{p,q} + \frac{1}{2}\sum_p b_p$.

    \item QHD measurement elements: $\{E^{(k)}_j: 1 \le k \le d, ~0 \le j \le r\}$, in which 
    \begin{align}\label{eqn:Ekj}
        E^{(k)}_j \coloneqq I \otimes ... \otimes \underbrace{E_j}_{\text{the $k$-th operator}} \otimes ... \otimes I,
    \end{align}
    and $E_j \coloneqq \ket{j}\bra{j}$ for $j=0,1,\dots,r$. The measurement elements $\{E^{(k)}_j\}$ are in computational basis and used to determine the ``position'' of the quantum state.\footnote{To see this, we note that the (discretized) position operator $\hat{X}^{(k)} = \sum^r_{j=0} \frac{j}{r} E^{(k)}_j$.}
    
    \item Quantum Ising Machine measurement elements: $\{\mathscr{M}^{(k)}_j: 1 \le k \le d, ~0 \le j \le r\}$, in which
    \begin{align}\label{eqn:Mkj}
        \mathscr{M}^{(k)}_j \coloneqq I \otimes ... \otimes \underbrace{\Pi_j}_{\text{the $k$-th operator}} \otimes ... \otimes I,
    \end{align}
    and $\Pi_j \coloneqq \sum_{|b|=j}\ket{b}\bra{b}$, where $|b|$ is the Hamming weight of the bitstring $b$, $j =0,1,\dots,r$. Note that this set of measurements $\{\mathscr{M}^{(k)}_j\}$ is also in computational basis. 
\end{itemize}

\begin{remark}
    We will have a detailed discussion on how to compute the problem Hamiltonian $H_P$ in \sec{theory_analog}. The representation of $H_P$ in \eqn{ising-format-qp} is known as the \textbf{Ising format} of the problem Hamiltonian. In practice, we prefer to use an equivalent but much simpler description of $H_P$ known as the \textbf{QUBO format}, see \sec{qubo_format}.
\end{remark}

\vspace{4mm}
In the following \algo{analog-qhd}, we show how to use the dictionary \tab{HammingEncoding} to implement QHD on Quantum Ising Machines and solve quadratic programming problems \eqn{qp-analog}. This algorithm can be readily implemented on any real-world instances of Quantum Ising Machines, e.g., the D-Wave quantum sampler. The correctness of \algo{analog-qhd} will be shown in \thm{analog-implement}.

\begin{algorithm}[htbp]
    \caption{Relaxed QHD for Quadratic Programming}
    \label{algo:analog-qhd}
    Specify a resolution parameter $r\in \Z$ and a time-rescaling parameter $\lambda > 0$. Let $T$ be the total evolution time, $d$ be the dimension of $f$, and $n$ = $dr$ be the total number of qubits.\\
    Initialize an $n$-qubit quantum register to the uniform superposition state $\ket{\psi_0} = \ket{+}^{\otimes n}$.\\
    Set the quantum Hamiltonian $H(t) = -\frac{A(t)}{2}S_x + \frac{B(t)}{2}H_P$ with the parameters $A(t)$, $B(t)$ and $H_P$ as in \tab{HammingEncoding}.\\
    Set total running time $t_f = T/\lambda$. Emulate the Quantum Ising Machine for $t_f$.\\
    Measure the final state $\ket{\psi_{t_f}}$ with $\{\mathscr{M}^{(k)}_j\}$ for $k=1,2,\dots,d$ and $j = 0,\dots,r$.\\
    Repeat the simulation and measurement for $M$ times and we obtain a sample of solutions to the quadratic programming problem.\\
\end{algorithm}

\paragraph{Overheads.}
Most steps in \algo{analog-qhd} are either classical or can be efficiently implemented on quantum computers (e.g., evolution, measurement, etc.). The majority of time required to obtain a solution is to run quantum evolution on Quantum Ising Machines, which usually only takes tens to hundreds of micro-seconds. In practice, we also have to consider the time required to embed the problem Hamiltonian into the actual QPU (quantum Processing Unit) due to the limited connectivity; see ``Minor embedding and qubit counts'' in \sec{implementation_details}. Though the minor embedding step can take a relatively long time, we only have to do it once for each problem: the same minor embedding can be reused for all subsequent shots.

\paragraph{Beyond quadratic programming.} 
In \algo{analog-qhd}, we only consider the implementation of QHD with a quadratic potential function and a box constraint. Although our method does not directly apply to higher-order polynomial optimization, it is still possible to consider the \textit{quadratization} technique \cite{dattani2019quadratization} so that higher-order polynomials or non-linear functions are reduced to quadratic problems. However, it is worth noting that quadratization can introduce substantial overheads in the algorithm.

\subsection{Theory of analog implementation}\label{sec:theory_analog}
We call the Hamiltonian that we want to simulate as \textit{target Hamiltonian}, and the analog quantum simulator is specified by a \textit{simulator Hamiltonian}. When the simulator Hamiltonian does not match the target Hamiltonian, e.g., the QHD Hamiltonian \eqn{finite-qhd-discrete} is of tridiagonal form and it does not fit into the Quantum Ising Machine Hamiltonian \eqn{ising-machine}, we have to \textit{embed} the problem Hamiltonian into the simulator Hamiltonian via a technique called \textit{subspace encoding}~\cite{cubitt2018universal}.

\subsubsection{Subspace encoding}
Let $H$, $\simH$ be the problem Hamiltonian and the simulator Hamiltonian. We denote $\H$ (or $\Hs$) for the Hilbert space of the problem (or simulator) Hamiltonian. Let $\Herm(\H)$ (or $\Herm(\Hs)$) be the set of Hermitian operators over $\H$ (or $\Hs$). Given a linear subspace $\S \subset \Hs$, $P_{\S}$ denotes the projector onto $\S$; clearly, we have $I = P_{\S} + P_{\S^\perp}$, where $I$ is the identity map on the Hilbert space $\Hs$. 

\begin{definition}
    Given a linear operator $A$ on $\Hs$, we say $\S \subset \Hs$ is an invariant subspace of $A$ if $A$ maps the subspace $\S$ to itself. In other words, $\S$ is an invariant subspace of $A$ if
    \begin{align}
        P_{\S^\perp} A P_{\S} = 0.
    \end{align}
    If $\S$ is an invariant subspace of $A$, we denote $A|_{\S}$ as the restriction of $A$ on the subspace $\S$. 
\end{definition}

\begin{definition}
    Given a subspace $\S \subset \Hs$, we say an isometry $V\colon \H \to \S$ gives a \textbf{subspace encoding} of the quantum Hamiltonian $H$ into the simulator Hamiltonian $\widetilde{H}$ if (1) $\S$ is an invariant subspace of $\widetilde{H}$, and (2) $V^\dagger\left(\widetilde{H}|_{\S}\right) V = H$.
\end{definition}

\begin{lemma}\label{lem:subspace}
    Suppose the isometry $V\colon \H \to \S$ gives a subspace encoding of $H\in \Herm(\H)$ into $\simH \in \Herm(\Hs)$. Let $\mathcal{U}(t) = e^{-iHt}$ be the time evolution operator of $H$, and $\widetilde{\mathcal{U}}(t) = e^{-i\simH t}$ be the time evolution operator of $\simH$, we have
    \begin{align}
        \mathcal{U}(t) = V^\dagger \left[\widetilde{\mathcal{U}}(t)|_{\S}\right] V.
    \end{align}
\end{lemma}

\begin{proof}
    Since $\S$ is an invariant subspace of $\simH$, we split the Hamiltonian $\simH$ as a direct sum: $\simH = \simH|_{\S} \oplus \simH|_{\S^\perp}$. As a consequence,
    \begin{align*}
        \widetilde{\mathcal{U}}(t) = e^{-i\simH t} = e^{-i(\simH|_{\S})t} \oplus e^{-i(\simH|_{\S^\perp})t}.
    \end{align*}
    On the other hand, we write $ \widetilde{\mathcal{U}}(t)  = \widetilde{\mathcal{U}}(t)|_{\S} \oplus\widetilde{\mathcal{U}}(t)|_{\S^\perp}$. Comparing the two identities, we observe that the projection onto $\S$ commutes with exponentiation, i.e., 
    \begin{align*}
        e^{-i(\simH|_{\S})t} = \widetilde{\mathcal{U}}(t)|_{\S}.
    \end{align*}
    It follows that
    \begin{align*}
        V^\dagger \left[\widetilde{\mathcal{U}}(t)|_{\S}\right] V &= V^\dagger \left[e^{-i(\simH|_{\S})t}\right] V = V^\dagger \left[\sum^\infty_{k=0} \frac{1}{k!} \left(-i(\simH|_{\S})t\right)^k\right] V\\
        &= \sum^\infty_{k=0} \frac{1}{k!} \left(-iV(\simH|_{\S})V^\dagger t\right)^k = \sum^\infty_{k=0} \frac{1}{k!} \left(-iH t\right)^k = \mathcal{U}(t).
    \end{align*}
\end{proof}

The following proposition shows that it is possible to simulate the target Hamiltonian $H$ by emulating a larger simulator Hamiltonian $\simH$ that encodes $H$.

\begin{proposition}[Analog simulation via subspace encoding]\label{prop:analog-simulation}
    Suppose the isometry $V\colon \H \to \S$ gives a subspace encoding of $H\in \Herm(\H)$ into $\simH \in \Herm(\Hs)$.
    Given an initial state $\ket{\psi_0} \in \H$, we write $\ket{\tilde{\psi}_0} = V \ket{\psi_0} \in \Hs$. The quantum evolution in $\H$ (or $\Hs$) is specified by $\ket{\psi_t} = e^{-iHt}\ket{\psi_0}$ (or $\ket{\tilde{\psi}_t} = e^{-i\simH t}\ket{\tilde{\psi}_0}$). Let $O \in Herm(\H)$ be a quantum observable in $\H$, we define a new observable $\tilde{O} \coloneqq V O V^\dagger$ in $\Hs$. Then, for any $t \ge 0$, we have
    \begin{align}
        \bra{\psi_t} O \ket{\psi_t} = \bra{\tilde{\psi}_t} \tilde{O} \ket{\tilde{\psi}_t}.
    \end{align}
\end{proposition}
\begin{proof}
    Note that $\ket{\tilde{\psi}_t} = e^{-i\simH t}\ket{\tilde{\psi}_0}$ and $\tilde{O} = V O V^\dagger$, we have
    \begin{align*}
        \bra{\tilde{\psi}_t} \tilde{O} \ket{\tilde{\psi}_t} &= \bra{\tilde{\psi_0}}e^{+i\simH t} \left(V O V^\dagger\right) e^{-i\simH t}\ket{\tilde{\psi_0}}\\
        &= \bra{\psi_0} V^\dagger e^{+i\simH t} V O V^\dagger e^{-i\simH t}V\ket{\psi_0}\\
        &= \bra{\psi_0} \left(V^\dagger e^{+i\simH t} V \right) O \left(V^\dagger e^{-i\simH t}V\right)\ket{\psi_0}\\
        &= \bra{\psi_0} e^{+iHt} O e^{-iHt}\ket{\psi_0} =  \bra{\psi_t} O \ket{\psi_t}.
    \end{align*}
    Note that we invoke \lem{subspace} in the second to last step.
\end{proof}
\begin{remark}
    \prop{analog-simulation} is formulated for time-independent target Hamiltonian $H$. Actually, it also applies to time-dependent target Hamiltonian $H(t)$ as long as $V$ is a subspace encoding of $H(t)$ into $\simH(t)$ for any $t \ge 0$. Also, it is clear that $\ket{tilde{\psi}_t}\in S$ for any $t \ge 0$ because $S$ is an invariant subspace of $\simH(t)$. 
\end{remark}

\subsubsection{Hamming encoding}\label{sec:hamming-encoding}
Now, we introduce a specific subspace encoding named \textit{Hamming encoding} that will be used to embed the (relaxed) QHD Hamiltonian \eqn{relaxed-qhd-discrete} into Quantum Ising Machines.

\begin{definition}[Hamming states]
    The Hilbert space of $n$ qubits is $\Hs = \C^{2^n}$ with the standard computational basis $\{\ket{b}:b\in \{0,1\}^n\}$.
    The Hamming weight of a computational basis $\ket{b}$ is the number of ones in the bit-string $b$, denoted by $|b|$.

    Given an integer $j=0,1,...,n$, we define the $j$-th \textbf{Hamming state} as the uniform superposition over all computational basis of Hamming weight $j$, i.e., 
    \begin{align}
        \ket{H_j} \coloneqq \frac{1}{\sqrt{C_j}}\sum_{|b|=j}\ket{b},
    \end{align}
    where $C_j = \binom{n}{j}$.
\end{definition}

\begin{lemma}\label{lem:hamming-state}
    Hamming states are of norm $1$ and mutually orthorgonal, i.e., for $j,k = 0,1,...,n$,
    \begin{align}
        \braket{H_j}{H_k} = \begin{cases}
            1 & j = k\\
            0 & j \neq k\\
        \end{cases}
    \end{align}
\end{lemma}
\begin{proof}
    By the definition of Hamming states, 
    \begin{align*}
        \braket{H_j}{H_k} = \frac{1}{\sqrt{C_j C_k}} \sum_{|b|=j,|b'|=k}\braket{b}{b'}.
    \end{align*}
    If $j \neq k$, we have $\braket{b}{b'} = 0$ so that $\braket{H_j}{H_k}=0$. If $j = k$, $\braket{H_j}{H_k}=(\sum_{|b|=j}1)/C_j = 1$.
\end{proof}

We consider the subspace $\S \subset \Hs$ spanned by all $(n+1)$ Hamming states. By \lem{hamming-state}, the set of Hamming states actually form an orthonormal basis of $\S$. 

\begin{lemma}\label{lem:uni-gaussian}
    The uniform superposition state $\ket{s} = \ket{+}^{\otimes n}= \frac{1}{\sqrt{2^n}} \sum_{b\in\{0,1\}^n}\ket{b}$ is in the subspace $\S$. In particular, we have
    \begin{align}
        \ket{s} = \sum^n_{j=0} \sqrt{\frac{C_j}{2^n}} \ket{H_j}.
    \end{align}
\end{lemma}
\begin{proof}
    \begin{align*}
        \sum^n_{j=0} \sqrt{\frac{C_j}{2^n}} \ket{H_j} = \sum^n_{j=0} \sqrt{\frac{C_j}{2^n}} \left(\frac{1}{\sqrt{C_j}}\sum_{|b|=j}\ket{b}\right) = \frac{1}{\sqrt{2^n}}\sum^n_{j=0}\sum_{|b|=j}\ket{b} = \ket{s}.
    \end{align*}
\end{proof}

\lem{uni-gaussian} shows that a Quantum Ising Machine can be directly initialized to the subspace $\S$. Also, it shows that the superposition state $\ket{s}$ is a discrete Gaussian state in the Hamming basis.

\begin{lemma}\label{lem:xz-subspace}
    We denote $S_x = \sum^{n-1}_{j=0} \sigma^{(j)}_x$ and $S_z = \sum^{n-1}_{j=0} \sigma^{(j)}_z$. The subspace $\S$ is an invariant subspace of both $S_x$ and $S_z$. Furthermore, 
    \begin{align}
        S_x|_{\S} &= \sum^{n-1}_{j=0}\sqrt{(j+1)(n-j)}\Big(\ket{H_{j+1}}\bra{H_j} + \ket{H_j}\bra{H_{j+1}}\Big),\\
        S_z|_{\S} &= \sum^n_{j=0} (n-2j) \ket{H_j}\bra{H_j}.
    \end{align}
\end{lemma}
\begin{proof}
    Note that $\sigma^{(j)}_x$ acts on a computational basis state $\ket{b}$ by flipping the $j$-th qubit while keeping all other qubits unchanged. Therefore, $S_x$ is a sum of all single qubit flips. Then
    \begin{equation*}
        S_x \ket{H_0} = \sum_{j} \sigma_x^{(j)} \ket{0^{\otimes n}} = \sum_{|b|=1}\ket{b} = \sqrt{n} \ket{H_1},
    \end{equation*}
    and similarly,
    \begin{equation*}
        S_x \ket{H_n} = \sum_{j} \sigma_x^{(j)} \ket{1^{\otimes n}} = \sum_{|b|=n-1}\ket{b} = \sqrt{n}\ket{H_{n-1}}.
    \end{equation*}
   
   Let $j\in\{1,2,\dots,n-1\}$. For any bit-string $\ket{y}$ of Hamming weight $|y|=j-1$, there are $n-j+1$ bit-strings of Hamming weight $j$ that are mapped to $\ket{y}$ by a single bit flip. Similarly, for any bit-string $\ket{z}$ of Hamming weight $|z|=j+1$, there are $j+1$ bit-strings of Hamming weight $j$ that are mapped to $\ket{z}$ by a single bit flip. It follows that
    \begin{align*}
        S_x \ket{H_j} 
        &= \frac{n-j+1}{\sqrt{C_j}} \sum_{|y|=j-1} \ket{y} + \frac{j+1}{\sqrt{C_j}} \sum_{|z|=j+1} \ket{z}\\
        &= (n-j+1) \frac{\sqrt{C_{j-1}}}{\sqrt{C_{j}}} \ket{H_{j-1}} + (j+1) \frac{\sqrt{C_{j+1}}}{\sqrt{C_j}} \ket{H_{j+1}}\\
        &= \sqrt{\frac{(n-j+1)^2 n! j! (n-j)!}{n! (j-1)! (n-j+1)!}} \ket{H_{j-1}} + \sqrt{\frac{(j+1)^2 n! j! (n-j)!}{n!(j+1)!(n-j-1)!}}\ket{H_{j+1}}\\
        &= \sqrt{j(n-j+1)}\ket{H_{j-1}} + \sqrt{(j+1)(n-j)}\ket{H_{j+1}}.
    \end{align*}
    This proves the first part of the proposition. 
    
    Note that $\sigma_z$ maps $\ket{0}$ to itself, and flips the phase of $\ket{1}$: $\sigma_z\ket{1}=-\ket{1}$. For any bit-string $\ket{b}$ of Hamming weight $j$, we have
    \begin{align*}
        S_z \ket{b} = \sum_{j} \sigma_z^{(j)} \ket{b} = (n-j) \ket{b} - j\ket{b} = (n-2j)\ket{b},
    \end{align*}
    and it follows that $S_z \ket{H_j} = (n-2j)\ket{H_j}$ for any $j=0,1,...,n$. This proves the second part of the proposition.
\end{proof}

\begin{definition}[Hamming encoding in one dimension]\label{defn:1-dim-hamming}
    The \textbf{Hamming encoding} refers to the isometry $V\colon \H =\C^{n+1} \to \S$ such that
    \begin{align}\label{eqn:1-dim-hamming-encoding}
        V = \sum^n_{j=0} \ket{H_j}\bra{j}.
    \end{align}
\end{definition}

The following corollary is a direct consequence of \lem{xz-subspace}.

\begin{corollary}\label{cor:encoding_1d}
    Let $V$ be the Hamming encoding as in \defn{1-dim-hamming}.
    \begin{enumerate}
        \item Let $\hat{X}=\sum^n_{j=0}\frac{j}{n}\ket{j}\bra{j}$ be the discretized position observable as defined in \eqn{observable-discretize}, the Hamming encoding $V$ gives a subspace encoding of $\hat{X}$ into $\left(\frac{1}{2} - \frac{1}{2n}S_z\right)$;
        \item Let $\hat{A}'=\sum^{n-1}_{j=0}\sqrt{(j+1)(n-j)/n}\left(\ket{j+1}\bra{j} + \ket{j}\bra{j+1}\right)$ be the relaxed adjacency matrix as in \eqn{relaxed_adjacency}, the Hamming encoding $V$ gives a subspace encoding of $\hat{A}'$ into $\frac{1}{n^{1/2}}S_x$.
    \end{enumerate}
\end{corollary}

\begin{remark}
    \cor{encoding_1d} gives a subspace encoding of the tridiagonal matrix $A'$ into the off-diagonal part of the Quantum Ising Machine Hamiltonian, yet at the cost of using an exponentially small subspace $\S$ of full Hilbert space $\Hs = \C^{2^n}$. This seems to be an unavoidable compromise due to the limited programmability of Quantum Ising Machines.
\end{remark}

To summarize, the target Hamiltonian represents a mesh point $a_j=j/n$ by the state $\ket{j}$, while the Hamming encoding represents the same point by the $j$-th Hamming state $\ket{H_j}$. In particular, we use $n$ qubits on a Quantum Ising Machine to encode a single continuous variable. The same idea also works for higher dimensional cases. In what follows, we introduce a resolution parameter $r \in \Z$, which is the number of cells on each edge of the regular mesh discretizing $[0,1]^d$. We use $n = dr$ qubits to encode the full (discretized) $d$-dimensional QHD Hamiltonian (which is a $(r+1)^d$-dimensional matrix). Similar as in the one-dimensional case, a mesh point $(a_{j_1},...,a_{j_d})$ in $[0,1]^d$ is mapped to the tensor product of $d$ Hamming states $\ket{H_{j_1},...,H_{j_d}}$ on the simulator. 

\begin{definition}[Hamming encoding in finite dimensions]
    For an integer $d \ge 1$, the Hamming encoding in $d$ dimensions refers to the isometry $V_d: \H^{\otimes d} \to \S^{\otimes d}$ such that
    \begin{align}\label{eqn:d-dim-hamming-encoding}
        V_d = \underbrace{V \otimes \dots \otimes V}_{\text{$d$ copies of $V$}} = \sum^n_{j_1,...,j_d = 0}\ket{H_{j_1},...,H_{j_d}}\bra{j_1,...,j_d}.
    \end{align}
    In other words, $V_d$ maps $\ket{j_1,...,j_d}$ to $\ket{H_{j_1},...,H_{j_d}}$.
\end{definition}

In \lem{uni-gaussian}, we observe the uniform superposition state can be regarded as a Gaussian distribution in Hamming state basis. This observation generalizes to higher dimensions.

\begin{proposition}[Multivariate Gaussian distribution]\label{prop:multi-gaussian}
    Fix a resolution parameter $r \in \Z$ and dimension $d \in \Z$, and denote $n = dr$. Let $\Hs = \C^{2^r}$ The uniform superposition state in $\Hs^{\otimes d}$ is in the subspace $\S^{\otimes d}$:
    \begin{align}
        \ket{+}^{\otimes n} = \frac{1}{\sqrt{2^{n}}} \sum_{b\in\{0,1\}^n} \ket{b} = \sum^{r}_{j_1,...,j_d=0} \sqrt{\frac{C_{j_1}...C_{j_d}}{2^n}}\ket{H_{j_1},...,H_{j_d}}.
    \end{align}
\end{proposition}
\begin{proof}
    Using \lem{uni-gaussian},
	\begin{align*}
    	\ket{+}^{\otimes n} 
    	&= \bigotimes_{k=1}^{d} \ket{+}^{\otimes r}
    	= \bigotimes_{k=1}^{d} \sum_{j_k=0}^{r} \sqrt{\frac{C_{j_k}}{2^r}} \ket{H_{j_k}}\\ 
        &= \sum_{j_1,\dots,j_d=0}^{r} \sqrt{\prod_{k=1}^{d} \frac{C_{j_k}}{2^r}} \ket{H_{j_1},\dots,H_{j_d}} =\sum_{j_1,\dots,j_d=0}^{r} \sqrt{\frac{C_{j_1}\dots C_{j_d}}{2^n}} \ket{H_{j_1},\dots,H_{j_d}}.
    \end{align*}
\end{proof}

Define the following new operators:
\begin{align}
    S^{(k)}_\alpha = \sum^{kr-1}_{j=(k-1)r} \sigma^{(j)}_\alpha,
\end{align}
for $k = 1,...,d$ and $\alpha = x,z$. Also, we define 
\begin{align}\label{eqn:W_k}
    W^{(k)} = \frac{1}{2}-\frac{1}{2n}S^{(k)}_z.
\end{align} 

\begin{proposition}\label{prop:d-xz-subspace}
    Fix a resolution parameter $r \in \Z$ and dimension $d \in \Z$, and denote $n = dr$. The subspace $\S^{\otimes d} \subset \Hs^{\otimes}$ is an invariant subspace of $S^{(k)}_x$, $S^{(k)}_z$, and $W^{(k)}$. Furthermore, for any $k = 1,2,\ldots,d$, we have
    \begin{align}
        V_d^\dagger \left(\frac{1}{r^{1/2}}\sum^{n-1}_{j=0}\sigma^{(j)}_x\right) V_d = \hat{A}'_d,\label{eqn:relaxed_Ad}\\
        V_d^\dagger \left( W^{(k)} \right)V_d = \hat{X}^{(k)},\\
        V_d^\dagger \left(W^{(k)} W^{(\ell)}\right) V_d = \hat{X}^{(k)} \hat{X}^{(\ell)},
    \end{align}
    where the symbol $\hat{X}^{(k)}$ is defined as in \eqn{d-observable}, and $\hat{A}'_d$ is as in \eqn{d_relaxed_adjacency}.
\end{proposition}
\begin{proof}
    Since $S^{(k)}_z$, and $W^{(k)}$ only act on $r$ qubits used to represent the $k$th continuous variable, it follows from \lem{xz-subspace} that $\S^{\otimes d}$ is an invariant subspace of $S^{(k)}_x$, $S^{(k)}_z$, and $W^{(k)}$.
    
    It also follows from the one-dimensional case that $V^\dagger \left(\frac{1}{r^{1/2}} S_x^{(k)}\right) V = I\otimes \ldots \otimes \hat{A}' \otimes \ldots \otimes I$, where $\hat{A}'$ is the $k$-th operator in the tensor product.
    Thus,
    \begin{align*}
        V_d^\dagger \left(\frac{1}{r^{1/2}} \sum_{j=0}^{n-1} \sigma_x^{(j)}\right)V_d = V_d^\dagger \left(\frac{1}{r^{1/2}} \sum_{k=1}^{d} S_x^{(k)}\right)V_d = \hat{A}'_d.
    \end{align*}
    
     Since $W^{(k)}=\frac{1}{2}-\frac{1}{2r}S_z^{(k)}$, it follows from the one-dimensional case that $V^\dagger \left(W^{(k)}\right) V_d = \hat{X}^{(k)}$. Similarly, $V_d^\dagger \left(W^{(k)} W^{(\ell)}\right) V_d = \hat{X}^{(k)} \hat{X}^{(\ell)}$ as $W^{(k)}$ and $W^{(\ell)}$ operate on different qubits.
\end{proof}

\begin{corollary}[Ising format for quadratic programming]
    \label{cor:qp-encoding}
    Consider a quadratic function $f(\x)=\frac{1}{2} \x^\top \mathbf{Q} \x + \mathbf{b}^\top \x$, where $\x, \mathbf{b} \in \R^d$, and $\mathbf{Q} = (Q_{p,q}) \in \R^{d\times d}$ is a symmetric coefficient matrix. Let $\hat{X}^{(p)}$ be the same as in \eqn{d-observable}, the discretization of $f$ over the regular mesh $\mathcal{M}^d = \{a_j = \frac{j}{r}:j=0,1,...,r\}^d$ is
    \begin{align}
        \hat{F} =  \left(\frac{1}{2}\sum^d_{p=1}Q_{p,p} (\hat{X}^{(p)})^2 + \sum_{p > q}Q_{p,q}\hat{X}^{(p)} \hat{X}^{(q)}\right) + \sum^d_{p=1} b_p \hat{X}^{(p)}.
    \end{align}
    Let $V^{\otimes d}$ be the Hamming encoding isometry as in \eqn{d-dim-hamming-encoding}. The Hamming encoding gives a subspace encoding of $\hat{F}$ into $H_P$, where the Hamiltonian $H_P$ is of the form:
    \begin{align}
        H_P &= \left(\frac{1}{2}\sum^d_{p=1}Q_{p,p} (W^{(p)})^2 + \sum_{p > q}Q_{p,q}W^{(p)} W^{(q)}\right) + \sum^d_{p=1} b_p W^{(p)}\\
        &= \sum^{n-1}_{j=0} h_j \sigma^{(j)}_z + \sum_{\substack{j,k=0 \\j> k}}^{n-1} J_{j,k} \sigma^{(j)}_z\sigma^{(k)}_z + g,
    \end{align}
    with $g = \frac{1}{8}(1+\frac{1}{r})\sum_p Q_{p,p} + \frac{1}{4}\sum_{p>q} Q_{p,q} + \frac{1}{2}\sum_p b_p$, 
    \begin{align}
        h_j = - \frac{1}{4r}\left[\left(\sum^d_{q=1}Q_{p,q}\right)+2 b_p\right]
    \end{align}
    for $(p-1)r \le j \le pr-1$, and 
    \begin{align}
        J_{j,k} = \frac{Q_{p,q}}{4r^2}
    \end{align}
    for $(p-1)r \le j \le pr-1$, $(q-1)r \le k \le qr-1$.
\end{corollary}
\begin{proof}
    Note that $\bra{j_1,\dots,j_d}\hat{X}^{(p)}\ket{j_1,\dots,j_d} = a_{j_p}$ and $\bra{j_1,\dots,j_d}\hat{X}^{(p)}\hat{X}^{(q)}\ket{j_1,\dots,j_d} = a_{j_p}a_{j_q}$ for any $p,q=1,\dots,d$. Therefore, it follows that
    \begin{align*}
        f(a_{j_1},...,a_{j_d}) = \bra{j_1,...,j_d}\hat{F}\ket{j_1,...,j_d}.
    \end{align*}
    It follows from \prop{d-xz-subspace} that
    \begin{align*}
        V_d^\dagger H_P V_d
        &= \frac{1}{2} \sum_{p=1}^{d} Q_{p,p} V_d^\dagger \left(\hat{W}^{(p)}\right)^2 V_d + \sum_{p>q} Q_{p,q} V_d^\dagger \left(\hat{W}^{(p)} \hat{W}^{(q)}\right) V_d + \sum_{p=1}^{d} b_p V_d^\dagger \left(\hat{W}^{(p)}\right) V_d\\
        &= \frac{1}{2} \sum_{p=1}^{d} Q_{p,p} \left(\hat{X}^{(p)}\right)^2 + \sum_{p>q} Q_{p,q} \hat{X}^{(p)} \hat{X}^{(q)} + \sum_{p=1}^{d} b_{p} \hat{X}^{(p)}\\
        &= \hat{F}.
    \end{align*}
    The form of $H_P$ in terms of Pauli-$z$ operators is by straightforward calculation, combining terms that include the identity, 1-site, and 2-site Pauli-$z$ operators.
\end{proof}

\subsubsection{Time-energy rescaling}
\label{sec:time-energy-rescale}
Recall that we want to simulate the relaxed version of discretized QHD \eqn{relaxed-qhd-discrete} for $t \in [0,T]$. This ODE is a mathematical model and it is dimensionless. However, real-world analog quantum computers operate at certain energy levels and its Hamiltonian is of the unit of energy (e.g., Joule); see \eqn{analog-evolution}. To connect the dimensionless ODE to a realistic analog quantum computer, we need to calibrate the ``clock'' in the abstract time evolution so that it matches real-world physical devices. 

We introduce a rescaled time $t = \lambda \xi$, where $\lambda > 0$ is a time dilation parameter. The rescaled wave function is 
\begin{align}
  \tilde{\psi}(\xi) = \psi(\lambda \xi),
\end{align}
with the evolution duration $\tilde{T} = T/\lambda$. Plugging the rescaled wave function to the Schr\"odinger equation \eqn{relaxed-qhd-discrete}, together with the Hamming encoding: $A_d \to S_x/\sqrt{r}$, we obtain:
\begin{align}\label{eqn:rescaled-qhd-discrete}
    i \frac{\d}{\d \xi}\ket{\tilde{\psi}_\xi} = \left[\underbrace{-\frac{\lambda e^{\varphi_t}}{2} r^{3/2} S_x}_{\text{Kinetic part}} + \underbrace{e^{\chi_t}\lambda V_d}_{\text{Potential part}}\right]\ket{\tilde{\psi}_\xi},
\end{align} 
This means if we want to get the evolution result in \eqn{relaxed-qhd-discrete} at time $T$, we should let the Quantum Ising Machine emulate for $T/\lambda$ with the following time-dependent functions \eqn{ising-machine}: 
\begin{subequations}
\begin{align}
    \frac{A(t)}{h} &= \lambda r^{3/2} e^{\varphi(t)},\label{eqn:rescale_A}\\
    \frac{B(t)}{h} &= 2\lambda e^{\chi(t)}.
\end{align}
\end{subequations}

\subsubsection{Reconstructing quantum distribution}
To obtain the quantum distribution of the QHD wave function $\ket{\psi}$, it is sufficient to measure the quantum state with the (discretized) position observable $\hat{X}$. In particular, if a wave function in $\ket{\psi} \in \H = \C^{r+1}$ is written as $\ket{\psi} = \sum^r_{j=0} c_j\ket{j}$, we use the measurement elements $\{E_j=\ket{j}\bra{j}\}$ to recover the quantum density distribution $\{|c_j|^2\}$. From the quantum distribution, the global minimizer of $f$ is inferred. 

The Hamming encoding isometry is $V = \sum^r_{j=0} \ket{H_j}\bra{j}$. By \prop{analog-simulation}, in order to recover the same measurement results, we want to use the transformed measurement elements $\{\tilde{E}_j = V E_j V^\dagger = \ket{H_j}\bra{H_j}\}$. 
However, according to \defn{quantum-ising-machine}, we can only measure a state in Quantum Ising Machines with computational basis. At a first glance, it seems impossible to perform quantum measurement with $\{\tilde{E}_j\}$. Surprinsingly, we show that it is still possible to recover the quantum distribution $\{|c_j|^2\}$ on Quantum Ising Machines by using a new set of measurement elements: $\{\Pi_j = \sum_{|b|=j} \ket{b}\bra{b}: j = 0,1,...,r\}$. Note that this measurement is performed in computational basis on Quantum Ising Machines.

\begin{lemma}[One-dimensional case]\label{lem:reconstruct}
    Let $\H' = \C^{2^r}$ be the simulator Hilbert space, and we denote $\S$ as the subspace spanned by Hamming states $\{\ket{H_j},j=0,1,...,r\}$. Suppose $\ket{\psi} \in \S$. Consider two sets of measurement elements in $\H'$: $\{\tilde{E}_j =  \ket{H_j}\bra{H_j}:j=0,1,...,r\}$ and $\{\Pi_j = \sum_{|b|=j} \ket{b}\bra{b}: j=0,1,...,r\}$. We have
    \begin{align}
       \bra{\psi} \tilde{E}_j \ket{\psi} = \bra{\psi} \Pi_j \ket{\psi},
    \end{align}
    for all $j = 0,1,\dots,r$.
\end{lemma}
\begin{proof}
    We write the quantum state $\ket{\psi} = \sum^r_{j=0} c_\ell \ket{H_j}$. It is clear that $\bra{\psi} \tilde{E}_j \ket{\psi} = |c_j|^2$. At the same time, we compute $\bra{\psi} \Pi_j \ket{\psi}$,
    \begin{align*}
        \bra{\psi} \Pi_j \ket{\psi} &= \left(\sum_{\ell} \overline{c_{\ell}} \bra{H_\ell}\right)\left(\sum_{|b|=j}\ket{b}\bra{b}\right)\left(\sum_{m} c_{m} \ket{H_m}\right)\\
        &= \sum_{\ell,m,|b|=j} \overline{c_{\ell}}c_m \braket{H_\ell}{b}\braket{b}{H_m}.
    \end{align*}
    Note that $\braket{H_\ell}{b} = \delta_{\ell,|b|}\frac{1}{\sqrt{C_\ell}}$, where $C_\ell = \binom{r}{\ell}$. It turns out that 
    \begin{align*}
        \bra{\psi} \Pi_j \ket{\psi} = \overline{c_j}c_j \sum_{|b|=j}\frac{1}{C_j} = |c_j|^2 = \bra{\psi} \tilde{E}_j \ket{\psi}.
    \end{align*}
\end{proof}

\begin{proposition}[Finite-dimensional case]\label{prop:d-dim-reconstruct}
   Let $d$ be the dimension of the problem, and $r\in \Z^+$ be the resolution parameter. Let $\H'$, $\S$ be the same as in \lem{reconstruct}. Suppose $\ket{\psi} \in \S^{\otimes d}$. For $1\le k \le d$ and $0 \le j \le r$, we define
   \begin{align}\label{eqn:transformed_Ekj}
       \tilde{E}^{(k)}_j \coloneqq I \otimes \dots \otimes \underbrace{\tilde{E}_j}_{\text{the $k$-th operator}} \otimes \dots \otimes I,
   \end{align}
   where $\tilde{E}_j = \ket{H_j}\bra{H_j}$. Let $\mathscr{M}^{(k)}_j$ be the same as in \eqn{Mkj}. Then, for any $1\le k \le d$ and $0 \le j \le r$, we have
   \begin{align}
    \braket{\psi}{\tilde{E}^{(k)}_j\Big|\psi} = \braket{\psi}{\mathscr{M}^{(k)}_j\Big|\psi}.
   \end{align}
\end{proposition}
\begin{proof}
    \begin{align*}
        \braket{\psi}{\tilde{E}^{(k)}_j\Big|\psi} &=  \bra{\psi}\left(I \otimes \dots \otimes \tilde{E}_j \otimes \dots \otimes I\right)\ket{\psi}\\
        &= \bra{\psi}\left(I \otimes \dots \otimes \Pi_j \otimes \dots \otimes I\right)\ket{\psi} = \braket{\psi}{\mathscr{M}^{(k)}_j\Big|\psi}.
    \end{align*}
    Note that the second equality follows from \lem{reconstruct}, and the last step uses the definition of $\mathscr{M}^{(k)}_j$.
\end{proof}

\subsubsection{Correctness of our analog implementation}
Now, we prove the correctness of \algo{analog-qhd}.

\begin{theorem}[Analog implementation of QHD]\label{thm:analog-implement}
    Let $d$ be the dimension of the quadratic programming problem \eqn{qp-analog}. Given a resolution parameter $r \in \Z^+$ and a time-rescaling parameter $\lambda > 0$. Let $T$ be the total evolution time of QHD, and we define $t_f = T/\lambda$. Let $\ket{\phi_0}$, $\ket{\psi_0}$, $E^{(k)}_j$, and $\mathscr{M}^{(k)}_j$ be defined as in \tab{HammingEncoding}. We denote $\ket{\phi_T}$ as the solution of the relaxed QHD \eqn{relaxed-qhd-discrete} at time $t = T$, and $\ket{\psi_{t_f}}$ be the state in the Quantum Ising Machine at evolution time $t=t_f$ (following \algo{analog-qhd}). Then, we have
    \begin{align}
        \braket{\psi_{t_f}}{\mathscr{M}^{(k)}_j\Big|\psi_{t_f}} = \braket{\phi_T}{E^{(k)}_j\Big|\phi_T},
    \end{align}
    for any $1\le k \le d$ and $0 \le j \le r$.
\end{theorem}
\begin{proof}
    Let $V\colon \C^{r+1} \to S \subset \C^{2^r}$ be the Hamming encoding defined in \eqn{1-dim-hamming-encoding}. Let $\H = \C^{(r+1)^d}$ be the Hilbert space of discretized QHD (in $d$ dimensions), and $\H'=\C^{2^n}$ be the simulator Hamiltonian with $n=dr$ being the total number of qubits required. We denote $V_d\colon \H \subset S^{\otimes d} \subset \H'$ as the Hamming encoding in $d$ dimensions as in \eqn{d-dim-hamming-encoding}. 

    First, we check the initial state. By \prop{multi-gaussian}, we see that $V_d \ket{\phi_0} = \ket{psi_0}$. Next, we check the Hamiltonian encoding. By \prop{d-xz-subspace} and \cor{qp-encoding}, we have 
    \begin{align*}
        V^\dagger_d\left(-\frac{e^{\varphi_t}r^{3/2}}{2}S_x + e^{\chi_t} H_P\right)V_d = -\frac{r^2}{2} \hat{A}'_d + e^{\chi_t}\hat{F}_d,
    \end{align*}
    which matches the relaxed QHD Hamiltonian \eqn{relaxed-qhd-discrete}. We further introduce the time rescaling: $t = \lambda \xi$ for $\xi \in [0, T/\lambda]$, the time-rescaled simulator dynamics is described as in \eqn{rescaled-qhd-discrete}. Therefore, by \prop{analog-simulation}, we have $\braket{\phi_T}{E^{(k)}_j\Big|\phi_T} = \braket{\psi_{t_f}}{\tilde{E}^{(k)}_j\Big|\psi_{t_f}}$, where the transformed measurement element is $\tilde{E}^{(k)}_j = V_d E^{(k)}_jV^\dagger_d$ (see \eqn{transformed_Ekj}). Lastly, it follows from \prop{d-dim-reconstruct} that 
    \begin{align*}
        \braket{\phi_T}{E^{(k)}_j\Big|\phi_T} = \braket{\psi_{t_f}}{\mathscr{M}^{(k)}_j\Big|\psi_{t_f}}.
    \end{align*}

\end{proof}

\subsection{Hamming encoding v.s. Radix-2 encoding: a case study}
We consider a quadratic program in two dimensions. The objective function is
\begin{equation}
    f(x_1,x_2) = \frac{1}{2}\left(Q_{11} x_1^2 + (Q_{12}+Q_{21}) x_1 x_2 + Q_{22} x_2^2\right) + b_1 x_1 + b_2 x_2,
\end{equation}
assuming the Hessian is symmetric, i.e. $Q_{12}=Q_{21}$.

Suppose the resolution is $1/8$, so we use $r=8$ qubits for each continuous variable for the Hamming encoding.

From \cor{qp-encoding}, the Quantum Ising Machine Hamiltonian is
\begin{align}
    H_{\text{Hamming}} &= -\frac{1}{32}\left[\left(Q_{11}+Q_{12}+2b_2\right)\sum_{j=0}^{7}\sigma_z^{(j)} + \left(Q_{21}+Q_{22}+2b_2\right)\sum_{j=8}^{15}\sigma_z^{(j)} \right] \\
    &+ \frac{1}{256}\left[Q_{11}\left(\sum_{j=0}^{7}\sum_{k=0}^{j-1} \sigma_z^{(j)}\sigma_z^{(k)}\right) + Q_{21}\left(\sum_{j=8}^{15}\sum_{k=0}^{7} \sigma_z^{(j)}\sigma_z^{(k)}\right)+Q_{22}\left(\sum_{j=8}^{15}\sum_{k=8}^{j-1} \sigma_z^{(j)}\sigma_z^{(k)}\right)\right],
\end{align}
ignoring the constant term which only contributes a global phase.

For the radix-2 encoding, we only use $r=4$ qubits for each variable. To derive the quantum Ising machine Hamiltonian, let $\mathbf{s}_1,\mathbf{s}_2\in\{-1,1\}^r$ be binary vectors such that $x_1$ and $x_2$ are approximated by
\begin{align}
    x_1 &\approx \frac{\mathbf{p}^\top (\mathbf{1}-\mathbf{s}_1)}{2} = \frac{1}{2}\left(1 -  \mathbf{p}^\top \mathbf{s}_1\right) \\
    x_2 &\approx \frac{\mathbf{p}^\top (\mathbf{1}-\mathbf{s}_2)}{2} = \frac{1}{2}\left(1 -  \mathbf{p}^\top \mathbf{s}_2\right),
\end{align}
where $\mathbf{p}=\begin{pmatrix} 1/2 & 1/4 & 1/8 & 1/8 \end{pmatrix}$.

The discretized quadratic objective function is
\begin{align}
    \hat{f}(\mathbf{s_1},\mathbf{s_2})
    &= \frac{1}{8}\left(Q_{11}\mathbf{s}_1^\top \mathbf{p}\mathbf{p}^\top \mathbf{s}_1 + (Q_{12}+Q_{21}) \mathbf{s}_1^\top \mathbf{p}\mathbf{p}^\top \mathbf{s}_2 + Q_{22} \mathbf{s}_2^\top \mathbf{p}\mathbf{p}^\top \mathbf{s}_2 \right) \\
    & \qquad - \frac{1}{8}\left((2Q_{11}+Q_{12}+Q_{21}+4b_1)\mathbf{p}^\top \mathbf{s}_1 + (2Q_{22}+Q_{12}+Q_{21}+4b_2)\mathbf{p}^\top \mathbf{s}_2\right)\\
    & \qquad + \frac{1}{8}\left(Q_{11}+Q_{12}+Q_{21}+Q_{22}+4b_1 + 4b_2\right)
\end{align}

Let $p_j$ be the $j$th entry of $\mathbf{p}$, and let $p_{jk}=p_{j}p_{k}$ be the $(j,k)$-th entry of $\mathbf{p}\mathbf{p}^\top$.
Ignoring the constant terms, the quantum Ising machine Hamiltonian is
\begin{align}
    H_{\text{Radix-2}}
    &= -\frac{1}{4}\left[(Q_{11}+Q_{12}+2b_1)\sum_{j=0}^{3}p_j \sigma_z^{(0)} + (Q_{21}+Q_{22}+2b_2)\sum_{j=0}^{3}p_j \sigma_z^{(j+4)}\right]\\
    & + \frac{1}{4}\left[Q_{11}\sum_{j=0}^{3}\sum_{k=0}^{j-1} p_{jk} \sigma_z^{(j)}\sigma_z^{(k)} + Q_{22}\sum_{j=0}^{3} \sum_{k=0}^{j-1} p_{jk} \sigma_z^{(j+4)}\sigma_z^{(k+4)}\right]\\
    & + \frac{(Q_{12}+Q_{21})}{8} \sum_{j=0}^{3}\sum_{k=0}^{3}p_{jk}\sigma_z^{(j)}\sigma_z^{(k+4)}.
\end{align}

\section{Quadratic Programming on D-Wave}
\label{sec:qp-benchmark}
\subsection{Test problems}\label{sec:qp-test-problems}
We test the performance of QHD and other algorithms on quadratic programming (QP) problems. Quadratic programming problems are the ``simplest'' class of optimization problems beyond linear programs with many applications in operations research \cite{furini2019qplib}. However, solving QP problems is challenging in practice: non-convex QP is known to be NP-hard \cite{vavasis1990quadratic}, even with one negative eigenvalue \cite{pardalos1991quadratic}.

In our empirical study, we focus on the special case of quadratic programming with box constraints, where the problem is formulated as
\begin{subequations}\label{eqn:qp}
\begin{align}
    \text{minimize}&\qquad f(\x)=\frac{1}{2} \x^\top \mathbf{Q} \x + \mathbf{b}^\top \x, \label{eqn:qp_obj}\\
    \text{subject to}
    &\qquad \mathbf{0} \preccurlyeq \x \preccurlyeq \mathbf{1},
\end{align}
\end{subequations}
where $\mathbf{0}$ and $\mathbf{1}$ are $n$-dimensional vectors of all zeros and all ones, respectively. In general, such problems are still known to be NP-hard \cite{burer2009nonconvex}.

We implement quantum algorithms (QHD and QAA) on the D-Wave \texttt{Advantage\_system6.1}. While there exists extensive libraries of QP test instances in the literature (e.g., \cite{furini2019qplib}), most of these instances can not be solved by the D-Wave QPU because of the limited connectivity between qubits.\footnote{The qubits on D-Wave QPU are aranged as the Pegasus graph. The maximum degree of the Pegasus graph is 15 \cite{mcgeoch2020dwave}. This means we can not implement very dense qubit-qubit interaction on the QPU.} For this reason, we randomly generate a test benchmark set with 160 QP instances whose Hessin matrices $\mathbf{Q}$ are of low sparsity. 

\paragraph{Parameters of the QP benchmark.}
These instances are split into 4 groups by their dimension: \texttt{QP-5d}, \texttt{QP-50d-5s}, \texttt{QP-60d-5s} and \texttt{QP-75d-5s}.\footnote{Due to limited connectivity and number of qubits, the maximum dimension of the problems we are able to embed on the D-Wave QPU is roughly 75.} The sparsity of these problems (i.e., the number of non-zero entries per row/column in the Hessian matrix $\mathbf{Q}$) is fixed to be $5$. There are ten $5$-dimensional instances, and fifty instances of dimension $50$, $60$, $75$, respectively. The entries of the Hessian are uniformly sampled from $[-1,1]$. We enforce the box constraint $\x \in [0,1]^n$ for all instances.

\subsection{Optimization solvers and metric of performance}
We test 7 methods: DW-QHD, DW-QAA, IPOPT, SNOPT, MATLAB's \texttt{fmincon} (with SQP), QCQP~\cite{park2017general}, and a basic Scipy \texttt{minimize} (with TNC) method. The ground truth (global solution) is obtained by running Gurobi. For each instance and method, we perform 1000 trials (or anneals/shots, in the case of quantum methods). We implement the quantum methods (DW-QHD, DW-QAA) on the D-Wave \texttt{Advantage\_system6.1}. The results from D-Wave QPU are post-processed by the Scipy \texttt{minimize} local solver (with TNC method). For the classical methods, the initial point is a uniformly random point in the box $[0,1]^d$. For quantum algorithms, we discretize each dimension into 8 cells (i.e., resolution = 8 per continuous variable). This means 4 logical qubits per dimension in QAA and 8 logical qubits per dimension in QHD. In order to investigate the quality of the D-Wave QPU, we also numerically simulate QHD and QAA for the ten 5-dimensional instances (Sim-QHD, Sim-QAA) on a classical computer. For a fair comparison, we use the same D-Wave QPU parameters~\cite{dwave-qpu-characteristics} in the numerical simulation. All the classical computations are done using a computer with 11th Gen Intel(R) Core(TM) i7-11800H @ 2.30GHz. We use Python 3.9.7 and MATLAB R2021b.\footnote{The QCQP package has not been updated since 2018. To run the QCQP solver, we use Python 3.6.13.}

We use the time-to-solution (TTS) metric \cite{ronnow2014defining} to compare the performance of optimization algorithms. TTS is the number of runs or shots required to obtain the correct global solution at least once with $0.99$ success probability:
\begin{align}
    TTS = t_f \times \Big\lceil\frac{\ln(1-0.99)}{\ln(1-p_s)}\Big\rceil,
\end{align}
where $t_f$ is the average runtime per trial/shot, and $p_s$ is the success probability of finding the global solution in a given trial/shot. We regard a given result $x_f$ as a global solution if $|f(x_f) - f(x^*)|\le 0.01$, where $x^*$ is the solution returned by Gurobi.

\subsection{Details of the D-Wave quantum computer}
In nature, the D-Wave quantum computer is a programmable quantum Ising machine. The quantum evolution in D-Wave machines is described by the time-dependent Schr\"odinger equation ($0 \le t \le t_f$):
\begin{align}\label{eqn:dwave-evolution}
    ih \frac{\d}{\d t}\ket{\psi_t} = H_{dwave}(t)\ket{\psi_t},
\end{align}
where $h$ is the Planck's constant,\footnote{To be consistent with the data provided by the D-Wave company, we use the \textit{original} Planck's constant in the Schr\"odinger equation. Its exact value is $h = 6.626 \times 10^{-34}$ J/Hz.} and the system Hamiltonian reads \cite[Section 3.3]{dwave-manual}:
\begin{align}\label{eqn:dwave-ham}
    H_{dwave}(s) = \underbrace{- \frac{A(s)}{2} \left(\sum_j \sigma^{(j)}_x\right)}_{\text{Initial Hamiltonian}} + \underbrace{\frac{B(s)}{2} \left(\sum_j h_j \sigma^{(j)}_z + \sum_{j>k} J_{j,k} \sigma^{(j)}_z \sigma^{(k)}_z\right)}_{\text{Problem Hamiltonian}},
\end{align}
where $s$ is the \textit{normalized anneal fraction} ranging from $0$ to $1$. The operator $\sigma^{(j)}_\alpha$ stands for 1-site Pauli operators acting on $n$ qubits:
\begin{align}
    \sigma^{(j)}_\alpha = I_0 \otimes ...\otimes (P_\alpha)_j \otimes ... \otimes I_{n-1},
\end{align}
where $P_\alpha$ with $\alpha = x,y,z$ refers to the Pauli matrices. For simplicity, we denote 
\begin{align}
    S_x = \sum_j \sigma^{(j)}_x.
\end{align}

\paragraph{Programmable parameters of the D-Wave machine:}
\begin{enumerate}
    \item Problem description parameters $h_j$, $J_{j,k}$. They are used to formulate the problem Hamiltonian so that its ground state encodes the anwser to the problem that we desire to solve. Note that the parameters $h_j$, $J_{j,k}$ describe the problem in Ising format. D-Wave quantum sampler also accept problems in QUBO format, see~\sec{qubo_format}.
    \item Annealing duration $t_f$. This parameter is the end time of the quantum simulation and can be specified by users.
    \item Annealing schedule. D-Wave QPU has pre-fixed time-dependent functions $A(s)$ and $B(s)$ and users can not change these functions; for more details, see \cite{dwave-qpu-characteristics}. For a given annealing duration $t_f$, the default annealing schedule used by D-Wave QPU is the linear interpolation $s = t/t_f$, where $t$ is the actual annealing time. One can also specify custom annealing schedules~\cite{dwave-solver} by inputing a sequence $\{(t_0, 0), (t_1,s_1),\dots, (t_m,1)\}$, where $t_0 = 0$, and $t_m = t_f$ is the total annealing time. The user-specified annealing schedule will be interpreted as a non-decreasing function for $t\in[0,t_f]$ and $s\in[0,1]$. Adjacent pairs $(t_k,s_k)$ are joined by linear interpolation.
\end{enumerate}

\begin{remark}
    As mentioned above, the annealing functions $A(s)$, $B(s)$ are pre-fixed and not programmable by the users. They begin at $s=0$ with $A(s)\gg B(s)$ and end at $s=1$ with $A(s)\ll B(s)$. The actual shapes of $A$ and $B$ are non-linear, and it is reported that $B$ grows quadratically in $s$. On the D-Wave \texttt{Advantage\_system6.1}, $A(0)/h = 9.63$ GHz and $B(1)/h = 7.57$ GHz. These functions are illustrated in \fig{anneal_schedule}, see the ``DW default A(t)'' and ``DW default B(t)'' curves in the right panel.
\end{remark}

\subsection{QUBO format of quadratic programming}
\label{sec:qubo_format}
The D-Wave quantum sampler accepts two formats of problem description, i.e., the \textit{Ising} format and the \textit{QUBO} format~\cite{dwave-problem-reformulate}. These two formats are equivalent and can be converted to each other. In \cor{qp-encoding}, we give the Ising format of quadratic programming problems. In practice, we find the QUBO format is more succinct and easier to construct with standard numerical programming tools.

Quadratic Unconstrained Binary Optimization (QUBO) refers to the following minimization problem:
\begin{align}
    \min_{\vect{z}\in\{0,1\}^n}f(\vect{z}) = \frac{1}{2}\sum_{j< k}q_{j,k} z_j z_k + \sum_j q_j z_j,
\end{align}
where $q_{j,k},q_j$ are real numbers. 

To map a QP problem to a QUBO, we represent each continuous variable $x_i$ (with $i=1,\dots,d$) using a binary vector $\mathbf{w}_i \in \{0,1\}^b$ (where $b$ is the resolution in the binary expansion) such that $x_i=\mathbf{p}^\top \mathbf{w}_i$, where $\mathbf{p}\in \R^b$ is a precision vector that defines the encoding. Therefore, the full solution vector $\vect{x}\in \R^d$ is expanded to a binary variable $\vect{w} \in \R^{bd}$ such that 
\begin{align}\label{eqn:binary_expansion}
    \vect{x} = \mathbf{P} \vect{w},
\end{align}
where $\mathbf{P} \coloneqq \mathbf{I}_d \otimes \mathbf{p}^\top$ is the matrix that carries the encoding information. 

With the binary expansion \eqn{binary_expansion}, a QP problem \eqn{qp_obj} is mapped to the following QUBO problem:
\begin{align}\label{eqn:qubo_qp}
    \min_{\vect{w}\in \{0,1\}^{bd}} f(\vect{w}) =\frac{1}{2} \mathbf{w}^\top \mathbf{P}^\top \mathbf{Q} \mathbf{P}\mathbf{w} + \mathbf{b}^\top \mathbf{P} \mathbf{w}.
\end{align}

In DW-QHD, we use the Hamming encoding scheme to represent QP problems. In the Hamming encoding scheme, a real variable is expanded to a $b$-bit binary string and the real variable is evaluated from the Hamming weight of the binary string. With $b = 8$, the corresponding precision vector for DW-QHD is 
\begin{align}\label{eqn:qhd_precision}
    \mathbf{p}_\text{DW-QHD} = \begin{pmatrix} 1/8 & \dots & 1/8 \end{pmatrix}^\top.
\end{align}

Similarly, in DW-QAA, a real variable is represented by a $b$-bit binary string through the radix-2 expansion.\footnote{In other words, the encoding is the same as the base-2 floating point representation of real numbers.} We choose $b = 4$ for DW-QAA (so the smallest resolution is $1/8$, the same as DW-QHD) and the corresponding precision vector is
\begin{equation}
    \mathbf{p}_\text{DW-QAA} = 
    \begin{pmatrix}
    1/8 & 1/8 & 1/4 & 1/2
    \end{pmatrix}^\top.
\end{equation}

It is worth noting that the box constraint in the QP problem \eqn{qp} is implicitly enforced in the QUBO format \eqn{qubo_qp} as long as all elements in the precision vector $\mathbf{p}$ sum up to $1$.  

\subsection{Details of the implementation of DW-QHD and DW-QAA}\label{sec:implementation_details}
We implement two quantum methods (QHD, QAA) on the D-Wave \texttt{Advantage\_system6.1}, accessed through \href{https://aws.amazon.com/braket/}{Amazon Braket}.

\paragraph{Time-dependent functions (annealing schedules).}
The time-dependent functions $A(t)$ and $B(t)$ in DW-QAA and DW-QHD are realized by the control of the annealing schedules on D-Wave QPU \cite{dwave-solver}. We consider two annealing schedules in our experiment: 
\begin{enumerate}
    \item the D-Wave default annealing schedule $s = t/t_f$;
    \item a custom annealing schedule that resembles the time-dependent parameters in QHD. The custom annealing schedule is defined by the sequence of $(t,s)$ pairs: 
    $$\{(0,0), (400,0.3), (640, 0.6), (800, 1)\},$$ 
    where consecutive pairs are linearly interpolated.
\end{enumerate}
We set the annealing duration $t_f = \SI{800}{\micro\second}$ in our D-Wave experiment for reasonable performance. In \fig{anneal_schedule}, we illustrate the annealing schedules and the corresponding time-dependent functions.

\begin{figure}[!ht]
    \centering
    \includegraphics[width=14cm]{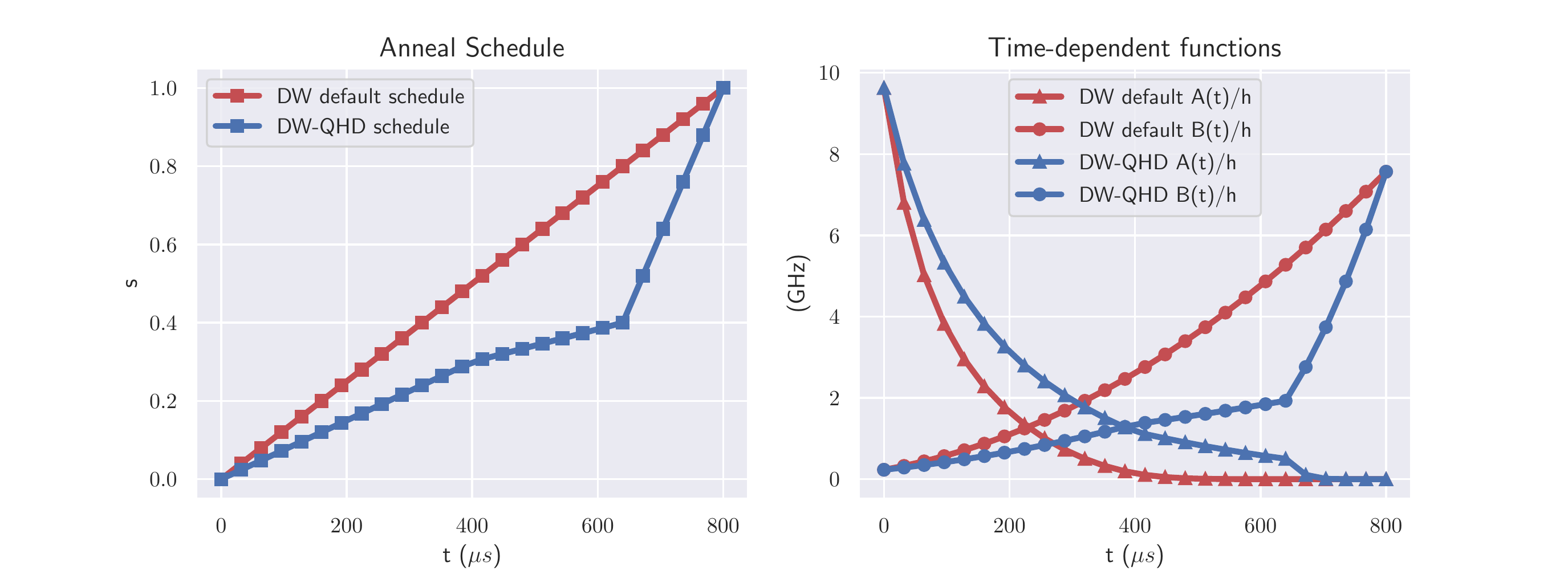}
    \caption[Anealing schedules in our experiments]{Annealing schedules of DW-QAA (default) and DW-QHD (custom).}
    \label{fig:anneal_schedule}
\end{figure}

To understand how the annealing schedules affect the performance of quantum algorithms, we conduct a pre-test of four different encoding-schedule combinations through numerical simulation:
\begin{itemize}
    \item (a) Radix-2 encoding + the default schedule; 
    \item (b) Radix-2 encoding + the custom schedule; 
    \item (c) Hamming encoding + the default schedule; 
    \item (d) Hamming encoding + the custom schedule.
\end{itemize}
The combinations (a) and (b) are for QAA, and (c) and (d) are for QHD. Each encoding-schedule combination is numerically simulated on a classical computer to solve problems in the benchmark \texttt{QP-5d} (consisting of ten 5-dimensional non-convex QP instances). The numerical simulation is performed with the machine Hamiltonian/parameters of the D-Wave \texttt{Advantage\_system6.1}, so the results reflect the performance of the quantum algorithms on an \textit{ideal} D-Wave machine. We find that the combinations (a) and (b) show similar performance (their TTS are in the magnitude of $10^{-2}$), while they are significantly worse than the combinations (c) and (d) (their TTS are in the magnitude of $10^{-3}$). Our finding suggests that the choice of quantum algorithms (or more specifically, the mapping of the QP problems to QIMs) mostly determines the quality of solutions, while the annealing schedules only play a minor role.

Based on our numerical observation, we decide to use the combination (a) in our D-Wave implementation of QAA (i.e., DW-QAA) and the combination (d) in our D-Wave implementation of QHD (i.e., DW-QHD). We believe our choices of annealing schedules best match the conventional design of QAA and QHD while giving us a faithful understanding of the performance of QAA/QHD on the D-Wave machine.

\paragraph{Minor embedding and qubit counts.}
Due to the limited connectivity of the D-Wave QPU, the target QUBO problems are complied to a sparser problem before mapping to the real quantum sampler. This compilation process is called \textit{minor embedding}~\cite{dwave-minorembedding} and it is automatically done by the D-Wave computing cloud. We list the number of logical qubits (i.e., qubits used in the target QUBO problem) and the number of physical qubits (i.e., active qubits used by the D-Wave QPU) in \tab{qhd_qubits} (DW-QHD) and \tab{qaa_qubits} (DW-QAA). Although we use only double the number of logical qubits for DW-QHD compared to DW-QAA (\tab{qaa_qubits}), the additional number of physical qubits required is much larger than double due to the limited connectivity of the D-Wave QPU. In fact, if we were to use 8 qubits for each continuous variable for QAA, the physical qubit counts would be identical to those shown in \tab{qhd_qubits}.

\begin{table}[!ht]
    \centering
    \begin{tabular}{|c|c|c|}
        \hline
        \textbf{Dimensions} & \textbf{Logical Qubits} & \textbf{Physical Qubits} \\
        \hline
        50 & 400 & 2630 \\
        \hline
        60 & 480 & 3276 \\
        \hline
        75 & 600 & 4064 \\
        \hline
    \end{tabular}
    \caption[Logical and physical qubit counts in DW-QHD]{Number of logical and physical qubits used to run DW-QHD for QP problems of dimensions 50, 60, and 75.}
    \label{tab:qhd_qubits}
\end{table}

\begin{table}[!ht]
    \centering
    \begin{tabular}{|c|c|c|}
        \hline
        \textbf{Dimensions} & \textbf{Logical Qubits} & \textbf{Physical Qubits} \\
        \hline
        50 & 200 & 662 \\
        \hline
        60 & 240 & 834 \\
        \hline
        75 & 300 & 1047 \\
        \hline
    \end{tabular}
    \caption[Logical and physical qubit counts in DW-QAA]{Number of logical and physical qubits used to run DW-QAA for QP problems of dimensions 50, 60, and 75.}
    \label{tab:qaa_qubits}
\end{table}

\paragraph{Effective evolution time of QHD.}
In our experiment, we set the annealing duration $t_f = \SI{800}{\micro\second}$ for DW-QHD. Now, we try to compute the effective evolution time $T$ as in the original QHD formulation:
\begin{align*}
    i \ket{\Psi_t} = \left[e^{\varphi_t}\left(-\frac{1}{2}\nabla^2\right) + e^{\chi_t} f(x)\right] \ket{\Psi_t}.
\end{align*}
We focus on the kinetic part, as the scaling of the potential $f$ does not change the location of the global minimizer. By \sec{time-energy-rescale}, the effective evolution time $T = \lambda t_f$, where the time dilation parameter $\lambda$ is computed by matching the coefficients in \eqn{rescale_A} at $s = 0$, 
\begin{align}
    \lambda e^{\varphi_0} r^{3/2} = \frac{A(0)}{h}.
\end{align}

Assuming that $e^{\varphi_0} = 500$ (as in \eqn{nonconvex_QHD_h}), and we note that $r=8$ and $A(0)/h = 9.63\times 10^9$ (Hz), it follows that
\begin{align}
    \lambda = \frac{A(0)/h}{r^{3/2}e^{\varphi_0}} = \frac{9.63\times 10^9}{8^{3/2} \times 500} \approx 8.51\times 10^5.
\end{align}
Therefore, the effective evolution time in QHD is 
\begin{align}
    T = \lambda t_f \approx 681.
\end{align}

\paragraph{Post-processing.}
Due to the limited QPU resources, we have low spatial resolution in the analog implementation of QHD and QAA. To better understand the performance of QHD and QAA on continuous problems, we apply Scipy \texttt{minimize} function (with the ``TNC'' solver) to post-process the quantum samples for DW-QHD and DW-QAA. The Scipy \texttt{minimize-TNC} function is a gradient-based local nonlinear solver, so it does not bring extra advantage in the post-processing for non-convex problems.

\subsection{Experiment results}\label{sec:experimental}

The TTS data for each tested solver are shown in \fig{fig4}. Raw experiment data and the code are available online.\footnote{Raw data is available on \href{https://umd.box.com/s/vq747fvjnt8qrkbxprexhoh44n0q9m0i}{Box}. The code for our experiment is available on \href{https://github.com/jiaqileng/quantum-hamiltonian-descent}{GitHub}.}

In the 5-dimensional benchmark, we also have simulated QHD and QAA (see Sim-QHD, Sim-QAA). It turns out that Sim-QHD (with evolution time $t_f = \SI{1}{\micro\second}$) still does better than DW-QHD (with annealing time $t_f=\SI{800}{\micro\second}$), indicating that the D-Wave system is subject to significant noise and decoherence. Another interesting observation is that DW-QAA (with annealing time $t_f = \SI{800}{\micro\second}$) does better than the ideal simulation Sim-QAA (annealing time $t_f=\SI{1}{\micro\second}$), which shows QAA has slower convergence. 

In the other three benchmarks (high-dimensional QP problems), DW-QHD outperforms DW-QAA, usually by a considerable margin. It is worth noting that we use different annealing schedules for DW-QHD and DW-QAA (but with the same annealing time $t_f = \SI{800}{\micro\second}$). However, our experiment shows that DW-QAA is roughly unchanged even if we use the same annealing schedule as for DW-QHD.

\subsection{Limitations of the D-Wave implementation of QHD}\label{sec:limitations}
The QHD we have implemented on the D-Wave QPU is still of the following limitations. 

\paragraph{Connectivity and programmability of the D-Wave QPU.}
Due to the limited connectivity of the physical qubits on the D-Wave \texttt{Advantage\_system6.1}, we have to employ minor embedding so that the target/problem Hamiltonian is mapped to the QPU. In theory, minor embedding only preserves the ground state of the problem Hamiltonian while it can perturb the subspace-encoded QHD Hamiltonian. This means the solutions we obtain with minor embedding can be worse than the solutions from the ideal Quantum Ising Machine setup (as in \algo{analog-qhd}). Also, the time-dependent functions $A(t)$ and $B(t)$ are not fully programmable so we can not implement arbitrary time-dependent parameters $e^{\varphi_t}$ and $e^{\chi_t}$ on D-Wave.

\paragraph{Naive algorithmic designs.}
Currently, our D-Wave implementation of QHD is exceedingly plain compared to classical methods, and we observe that DW-QHD is unable to outperform state-of-the-art branch-and-bound methods such as those used by Gurobi and CPLEX.\footnote{While branch-and-bound methods are typically used for discrete and combinatorial problems, they are also able to solve nonconvex QPs by performing spatial branching and solving linear programming relaxed subproblems \cite{mccormick1976computability}.} By effectively searching an enumeration tree, branch-and-bound solvers are able to find a solution and potentially guarantee optimality. In DW-QHD, we have neither the smart search strategy nor guarantee optimality of solutions. However, we want to remark that the branch-and-bound technique requires exponential resources in the worst-case (see \fig{branch_and_bound}), and generally problems of sufficiently large dimension and density cannot be solved to optimality in any reasonable amount of time. With larger quantum machines and better connectivity, it may be the case that QHD is able to solve problems with a better optimality gap than existing branch-and-bound methods with a reasonable but fixed runtime limit. Also, there is potential to integrate many classical ideas into QHD, e.g., use QHD as a subroutine in a branch-and-bound search strategy, etc. Our empirical study is evident to demonstrate to potential of QHD, in which we show QHD is competitive with state-of-the-art nonlinear local solvers such as Ipopt.

\begin{figure}[!ht]
    \centering
    \includegraphics[width=16cm]{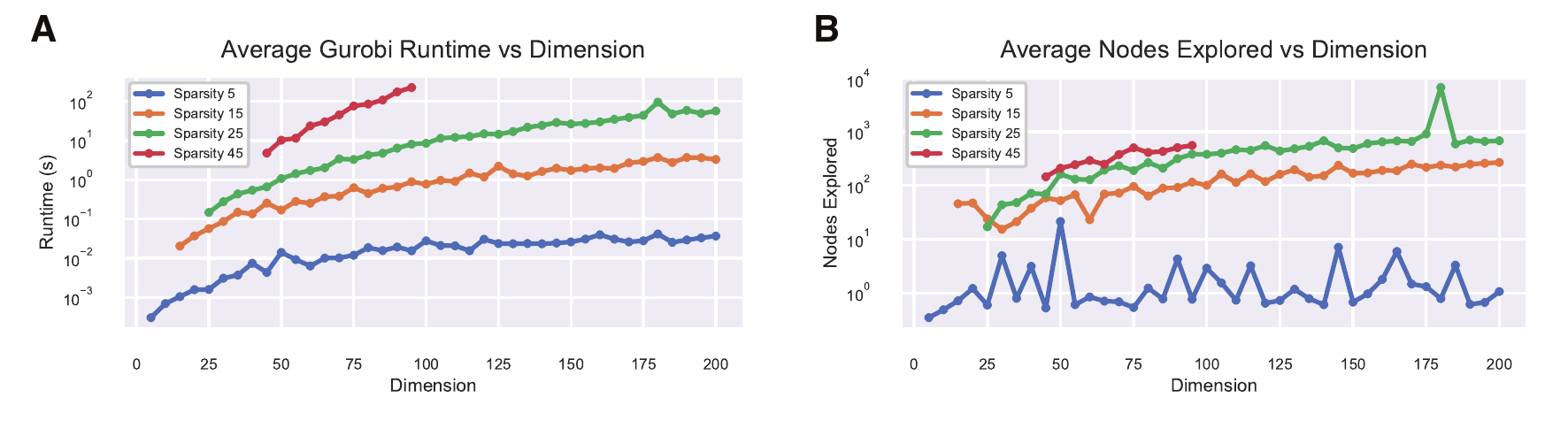}
    \caption[Gurobi runtime scales exponentially with problem parameters.]{\textbf{A:} Average runtime in seconds for Gurobi to solve random QPs of sparsity 5, 15, 25, and 45 for various dimensions. Each point denotes the average runtime over 100 random instances for a given dimension and sparsity.
    \textbf{B:} Average number of branching nodes explored by Gurobi for random QPs of sparsity 5, 15, 25, and 45 for various dimensions.}
    \label{fig:branch_and_bound}
\end{figure}

\subsection{Resource analysis for digital and analog implementation}
To better understand the actual cost of the digital implementation when applied to quadratic programming problems (with box constraints), we compute the number of T gate (i.e., T-count)\footnote{A standard universal gate set is Clifford + T. It is known that a T gate is far more expensive to implement fault tolerantly than any of the Clifford gates. Therefore, the T-count faithfully reflects the difficulty of fault-tolerant implementation.} required to solve the QP problems in different resolutions, assuming the same experiment setup as in our experiment on D-Wave QPU. 

For the digital implementation of QHD, we consider a similar product-formula method as in \eqn{prod_qhd_fft}. Note that the FFT can be replaced with Quantum Fourier Transform (QFT) and implemented on a digital quantum computer. 

We need to perform spatial discretization in \eqn{prod_qhd_fft}. The precision parameter $q$ (i.e., number of qubits per continuous variable) controls the resolution in the spatial discretization of the feasible domain. We will compute the T-count for $q = 3,16,32$. Note that the precision $q=3$ corresponds to the same resolution ($r=2^3$) we have in our experiment on the D-Wave sampler. The T-count of several quantum subroutines are listed in \tab{Tcount_subroutine}  (see \cite{haener2018quantum,nam2020approximate} for more details).\footnote{The estimated T-counts for $q=3$ are computed by extrapolation.} 

\begin{table}[!ht]
    \centering
    \begin{tabular}{ |p{5cm}|p{3cm}|p{3cm}|p{3cm}|  }
        \hline
            \textbf{Quantum subroutines} & \textbf{3-qubit} & \textbf{16-qubit}& \textbf{32-qubit}\\
        \hline
            Quantum adder & 587 & 4704 & 11144 \\
        \hline 
            Quantum multiplier & 173 & 6328 & 26642 \\
        \hline
            Approximate QFT & 170 & 1162 & 2698 \\
        \hline
    \end{tabular}
    \caption{T-count of quantum floating-point adder, multiplier, and approximate quantum Fourier transform}
    \label{tab:Tcount_subroutine}
\end{table}

\begin{table}[!ht]
    \centering
    \begin{tabular}{ |p{3.5cm}|p{3.5cm}|p{3.5cm}|p{3.5cm}| }
        \hline
            \textbf{Dimensions} & \textbf{3-qubit format} & \textbf{16-qubit format} & \textbf{32-qubit format}\\
        \hline
            50 & 5.49e+8 & 7.8386e+9 & 2.672e+10 \\
        \hline
            60 & 6.588e+8 & 9.4063e+9 & 3.2064e+10\\
        \hline
            75 & 8.235e+8 & 1.1758e+10 & 4.008e+10\\
        \hline
    \end{tabular}
    \caption[T-count of digital implementation of QHD]{T-count of digital implementation of QHD for QP problems.}
    \label{tab:resource-analysis}
\end{table}

\begin{proposition}\label{prop:tcount-algo}
    Consider a quadratic programming problem defined in \eqn{qp} without linear constraint (i.e., $\mathbf{Q}_c = 0$, $\mathbf{b}_c = 0$). Suppose the problem is of dimension $d$, and the problem matrix $\mathbf{Q}$ has at most $s$ non-zero elements on each row/column. We use $c_{add}$, $c_{mult}$, $c_{aqft}$ to denote the T-count in the adder, multiplier, and approximate QFT.  With $R$ iterations, the digital implementation has T-count $$2\Big((c_{add}+c_{mult})(s+2)+c_{aqft}\Big)dR.$$
\end{proposition}
\begin{proof}
    Suppose the initial guess state $\ket{\psi_0}$ is the uniform superposition state so it be efficiently prepared (T-count = 0). The implementation of $\exp(-is b_j \mathbf{V})$ uses two queries to $O_f$ (see Figure 1-1 in \cite[Section 1.2]{childs2004quantum}). There are $(sd+d)$ additions and multiplications, respectively, in the evaluation of $f$. Therefore, the T-count of $O_f$ is $(c_{add}+c_{mult})(sd+d)$. This means the T-count in $\exp(-is b_j \mathbf{V})$ is $2(c_{add}+c_{mult})(sd+d)$. In computing the matrix elements of $\mathbf{L}$, there are at least $d$ multiplications and $d$ additions; hence, the T-count is $2(c_{add}+c_{mult})d$ in the implementation of $\exp(-is a_j \mathbf{L})$. The QSFT for one dimension can be implemented with one approximate QFT and two diagonal operators \cite[Lemma 5]{childs2021high}. Even ignoring the T-count in the two diagonal operators, QSFT has (at least) the T-count of $c_{aqft}d$ for $d$ dimensions. Therefore, a single iterative step in the digital implementation uses $2((c_{add}+c_{mult})(s+2)+c_{aqft})d$. This number times $R$ is the total T-count in the quantum algorithm.
\end{proof}

In \tab{resource-analysis}, we evaluate the T-count in the digital implementation of QHD for solving the quadratic programming problems from \sec{qp-benchmark} based on \prop{tcount-algo} (we set $R=1000$).\footnote{In \sec{implementation_details}, we compute the effective QHD evolution time $T\approx 681$ in our D-Wave implementation. In fact, to simulate the same quantum evolution with digital quantum computer, the number of iterations $R$ will be much bigger than $T$. Here, we just take a mild estimate $R=1000$.} The T-count in the low-precision implementation (i.e., $3$-qubit representation of a continuous variable) is already greater than $5\times 10^{8}$. However, near-term digital quantum computers merely promise the controllability of hundreds of qubits and thousands of gates. This means the resources required by QHD to solve real-life problems are beyond the capability of any near-term digital quantum computers.

\bibliographystyle{myhamsplain}
\bibliography{ref-supp}

\end{document}